\documentclass[11pt,letterpaper]{article}
\usepackage{amsmath,amssymb,amsthm}
\usepackage{fullpage}
\usepackage{float,placeins}
\usepackage[colorlinks=true,citecolor=blue]{hyperref}
\usepackage{algorithm}
\usepackage[dvipsnames]{xcolor}

\newtheorem{theorem}{Theorem}[section]
\newtheorem{lemma}[theorem]{Lemma}

\newtheorem{corollary}[theorem]{Corollary}

\theoremstyle{definition}
\newtheorem{definition}[theorem]{Definition}
\newtheorem{observation}[theorem]{Observation}

\newtheorem{remark}[theorem]{Remark}
\numberwithin{equation}{section}

\floatstyle{boxed}\newfloat{Algorithm}{ht}{alg}

\newcommand{\eps}{\varepsilon}
\newcommand{\R}{\mathbb{R}}
\newcommand{\C}{\mathbb{C}}
\newcommand{\Z}{\mathbb{Z}}
\newcommand{\Q}{\mathbb{Q}}
\newcommand{\E}{\mathbb{E}}
\newcommand{\Img}{\operatorname{Im}}
\newcommand{\SL}{\operatorname{SL}}
\newcommand{\tr}{\operatorname{Tr}}

\newcommand{\prm}{\operatorname{pm}}
\newcommand{\row}{\operatorname{row}}

\newcommand{\diag}{\operatorname{diag}}
\newcommand{\capa}{\operatorname{cap}}
\newcommand{\caA}{\mathcal{A}}
\newcommand{\caB}{\mathcal{B}}

\newcommand{\caD}{\mathcal{D}}

\newcommand{\caI}{\mathcal{I}}
\newcommand{\bfzero}{\mathbf{0}}
\newcommand{\bfone}{\mathbf{1}}
\DeclareMathOperator*{\supp}{supp}
\DeclareMathOperator*{\rk}{rank}
\DeclareMathOperator*{\ncrank}{nc-rank}
\DeclareMathOperator{\poly}{poly}
\DeclareMathOperator*{\argmin}{argmin}

\usepackage{mathtools}

\DeclarePairedDelimiter{\norm}{\lVert}{\rVert}
\DeclarePairedDelimiter{\abs}{\lvert}{\rvert}
\DeclarePairedDelimiter{\inprod}{\langle}{\rangle}

\title{Shrunk subspaces via operator Sinkhorn iteration}%

\newif\ifanonymous
\anonymousfalse %
\ifanonymous
    \author{Anonymous authors}%
\else
    \author{%
    Cole Franks\thanks{Department of Mathematics, Massachusetts Institute of Technology. Email: \texttt{franks@mit.edu}.
     Part of C.F.'s work took place at the Geometric Methods in Optimization and Sampling workshop at the Simons Institute for the Theory of Computing.}
    \and
    Tasuku Soma\thanks{Department of Mathematics, Massachusetts Institute of Technology. Email: \texttt{tasuku@mit.edu}}
    \and
    Michel X. Goemans\thanks{Department of Mathematics, Massachusetts Institute of Technology. Email: \texttt{goemans@math.mit.edu}}
    } 
\fi
\begin{document}
\maketitle
\begin{abstract}
A recent breakthrough in Edmonds' problem showed that the noncommutative rank can be computed in deterministic polynomial time, and various algorithms for it were devised.
However, only quite complicated algorithms are known for finding a so-called shrunk subspace, which acts as a dual certificate for the value of the noncommutative rank.
In particular, the operator Sinkhorn algorithm, perhaps the simplest algorithm to compute the noncommutative rank with operator scaling, does not find a shrunk subspace. Finding a shrunk subspace plays a key role in applications, such as separation in the Brascamp-Lieb polytope, one-parameter subgroups in the null-cone membership problem, and primal-dual algorithms for matroid intersection and fractional matroid matching.

In this paper, we provide a simple Sinkhorn-style algorithm to find the smallest shrunk subspace over the complex field in deterministic polynomial time.  To this end, we introduce a generalization of the operator scaling problem, where the spectra of the marginals must be majorized by specified vectors. Then we design an efficient Sinkhorn-style algorithm for the generalized operator scaling problem.
Applying this to the shrunk subspace problem, we show that a sufficiently long run of the algorithm also finds an approximate shrunk subspace close to the minimum exact shrunk subspace. Finally, we show that the approximate shrunk subspace can be rounded if it is sufficiently close. Along the way, we also provide a simple randomized algorithm to find the smallest shrunk subspace.

As applications, we design a faster algorithm for fractional linear matroid matching and efficient weak membership and optimization algorithms for the rank-2 Brascamp-Lieb polytope.
\end{abstract}
\clearpage
\tableofcontents
\clearpage
\section{Introduction}
Let $\caA \subseteq \C^{n\times n}$ be a matrix space spanned by $A_1, \dots, A_p \in \C^{n\times n}$. Edmonds' problem is to decide whether $\caA$ has a matrix of full rank, denoted $\rk \caA = n$. There is a simple randomized algorithm for Edmonds' problem, but derandomizing it is a central question in polynomial identity testing. The completeness of the determinant for polynomial identity testing \cite{valiant1979complexity} implies that a positive solution to Edmonds' problem would imply historically intractable circuit lower bounds \cite{kabanets2004derandomizing}, so it comes as no surprise that Edmonds' problem remains open.

The \emph{noncommutative} Edmonds problem is to compute a relaxed notion of rank known as the \emph{noncommutative rank} $\ncrank\caA$ defined as
\begin{align*}
\ncrank\caA = \min\left\{n- \dim U + \dim \caA(U): \text{$U$ subspace of $\C^n$} \right\},
\end{align*}
where $\caA(U) = \inprod{Au: A \in \caA, u \in U}$.
Equivalently, $\ncrank\caA = n - c$, where $c$ is the largest difference between $\dim U$ and $\dim \caA(U)$.
We call such a subspace $U$ a \emph{$c$-shrunk subspace}. If $\caA$ is an yes instance of the Edmonds problem, then clearly $\ncrank\caA = \rk\caA = n$, i.e., $\dim U = \dim\caA(U)$ for all subspaces $U$. The noncommutative Edmonds problem is equivalent to testing the identity of noncommutative rational formulae \cite{Fortin2004}, and arises in diverse contexts such as  the invariant theory of quivers \cite{Derksen2017} and statistical estimation \cite{amendola2021invariant}. Maximum bipartite matching can be formulated as a special case of this problem, and the $c$-shrunk subspace amounts to the ``maximum deficiency set" from Hall's theorem.  Several independent works by Garg, Gurvits, Oliveira, Wigderson \cite{garg2016deterministic}, Ivanyos, Quiao, Subramanyam \cite{Ivanyos2018}, and Hamada and Hirai \cite{Hamada2021} exhibited deterministic polynomial-time algorithms for the noncommutative Edmonds' problem, showing that the noncommutative version of the problem is significantly easier.

The three algorithms have very different approaches. \cite{garg2016deterministic}, following \cite{Gurvits2004}, reduces the problem to the boundedness of a certain optimization problem over matrix groups known as the \emph{capacity}. The algorithm for solving this optimization problem, known as \emph{operator scaling}, is a delightfully simple generalization of Sinkhorn's algorithm \cite{sinkhorn1964relationship} for re-weighting (or ``scaling") the rows and columns of a nonnegative matrix to make it doubly stochastic (matrix scaling). However, the iterative algorithm lacks one key feature shared by a more sophisticated submodular minimization approach of \cite{Hamada2021} and algebraic approach of \cite{Ivanyos2018}: the optimization algorithm of \cite{garg2016deterministic} does \emph{not} produce a $c$-shrunk subspace.

The $c$-shrunk subspace can be viewed as a certificate bounding the noncommutative rank, and such certificates are useful in optimization applications (e.g. dual certificates) as well as computational invariant theory (it can be viewed as a $1$-parameter subgroup taking the tuple of matrices to the origin).
For instance, the $c$-shrunk subspace arises as a dual certificate in algorithms for the weighted matroid intersection problem and its generalizations~\cite{VandeVate1992,Chang2001b,Gijswijt2013}.
Also, finding a $c$-shrunk subspace is crucial in a weak separation algorithm for the Brascamp-Lieb polytope by \cite{Garg2018}.

Geometric invariant theory provides another motivation for finding shrunk subspaces - in the $m = n$ case, the group $G = \SL(n) \times \SL(n)$ (pairs of matrices with determinant one) naturally acts on the tuple $(A_1, \dots, A_p)$ generating $\caA$ by $(g,h) \cdot (A_1, \dots, A_p) = (g A_1 h^\dagger, \dots, g A_p h^\dagger).$ This is the well-studied \emph{left-right action}.~It is known that all homogeneous invariant polynomials invariant under this group action vanish on the tuple $(A_1, \dots, A_p)$, i.e. the tuple is in the null-cone, if and only if $\caA$ has a shrunk subspace \cite{burgin2006hilbert}. The Hilbert-Mumford criterion states that null-cone membership holds if and only if there is a \emph{one-parameter subgroup} (1-PSG) of $G$ driving the Euclidean norm of $(A_1, \dots, A_p)$ to zero \cite{mumford1994geometric}. 1-PSG's are deterministic certificates of null-cone membership, as opposed to oracle access to a random invariant polynomial. Shrunk subspaces encode such 1-PSG's, and $c$-shrunk subspaces correspond to 1-PSG's driving the tuple to the origin the ``fastest." Efficiently deciding null-cone membership for general linear group actions remains an interesting open problem.

\subsection{Our contribution}
We show how to modify the Sinkhorn-style operator scaling algorithm from \cite{garg2016deterministic} to find the smallest $c$-shrunk subspace of $\caA$. We view our result as evidence that continuous optimization is a viable, or at least competitive, approach to polynomial identity testing problems in invariant theory. Along the way, we find a simple randomized algorithm to find the smallest $c$-shrunk subspace of $\caA$ using Wong sequences. Even when allowing randomness, the sophisticated deterministic algorithms of \cite{Hamada2021} and \cite{Ivanyos2018} were the only known methods to find a $c$-shrunk subspace.

The idea of the deterministic iterative algorithm is as follows. Gurvits's original operator Sinkhorn algorithm \cite{Gurvits2004} can be seen as alternating minimization of an objective function (Equation \ref{eq:gurv}) on the space of (pairs of) positive-definite matrices. When the objective function is unbounded, a sufficiently long run of Gurvits's algorithm yields an \emph{approximate} shrunk subspace, a notion we define formally in Definition \ref{def:approx-indep}. It is easy to modify the objective function (Equation \ref{eq:simple-capacity}) to yield an approximate $c$-shrunk subspace; see Section \ref{sec:warm-up} for how to read off the maximum deficiency subset in the case of matrix scaling. Even so, it is not obvious how to obtain a true $c$-shrunk subspace from a mere approximate one. It is not clear that every approximately $c$-shrunk subspace is close to some $c$-shrunk subspace, let alone how to find the nearby true shrunk subspace.

To fix this issue, we further ``perturb" the objective function (Equation \ref{eq:perturb-capacity}) to find an $\eps$-approximate $c$-shrunk subspace of dimension at most the dimension $r^*$ of the true smallest $c$-shrunk subspace in time polynomial in $\log (1/\eps)$ and the input size - see Algorithm \ref{alg:approx-indep}, and for the algorithm's guarantees see Theorem~\ref{thm:indep-correct}. It is well-known that the true smallest $c$-shrunk subspace is unique, and we show it has polynomial bit complexity (Theorem \ref{thm:round}). This follows from the analysis of our simple randomized algorithm to output the smallest $c$-shrunk subspace. Finally, we show a stability result for the smallest $c$-shrunk subspace, namely that an $\exp(-\poly)$-approximate $c$-shrunk subspace of dimension at most $r^*$ is actually $\exp(-\poly)$-close to the true smallest $c$-shrunk subspace (Theorem \ref{thm:close}).
Together with our bit complexity bound, this shows our main result (Theorem \ref{thm:roundsub-correct}): we can round our approximate $c$-shrunk subspace to the true smallest $c$-shrunk subspace.

\subsection{Applications: rank-2 Brascamp-Lieb polytope and fractional linear matroid matching}

As an application of our methods, we prove several algorithmic results about the Brascamp-Lieb polytope (BL polytope for short)~\cite{Bennett2008}. The BL polytope for full row-rank matrices $B_i \in \C^{n_i \times n}$ $(i=1,\dots,p)$ is the set of $x \in \R^p$ such that
\begin{align}\label{eq:BL-V-repr}
\begin{split}
\sum_{i=1}^p x_i \dim(B_i W) &\geq \dim(W) \qquad \text{($W$: subspace of $\C^n$)} \\
\sum_{i=1}^p x_i n_i &= n, \\
x_i &\geq 0. \qquad (i=1, \dots, p)
\end{split}
\end{align}
The BL polytope with $n_i = r$ for all $i$ is called the rank-$r$ BL polytope.
We denote this polytope $BL(B_1, \dots, B_p)$ or concisely $BL(B)$, where $B = (B_1, \dots, B_p)$.

It was shown in \cite{Bennett2008} that the BL polytope determines the finiteness of the Brascamp-Lieb inequalities in functional analysis, and \cite{Garg2018} exhibited a separation algorithm for $x \in \mathbb{Q}^p$ that runs in time polynomial in the sizes of $B_i$ and the common denominator of entries of $x$.
Indeed, they reduced the separation problem to computing the noncommutative rank of a pseudopolynomial-size matrix space $\caA$ with the following property.
If $x \in BL(B)$, then $\ncrank\caA$ is full.
If $x \notin BL(B)$, then a shrunk subspace gives a violating subspace $W$, i.e., $\sum_{i=1}^p x_i \dim(B_i W) > \dim W$.
The former can be verified by the operator Sinkhorn algorithm, but the latter cannot; they used an algebraic algorithm of \cite{Ivanyos2018} separately to this end.

The BL polytope is significant in combinatorial optimization because various important polytopes are special cases of even low-rank BL polytopes.
The rank-1 BL polytope is simply the base polytope of a linear matroid~\cite{Barthe1998} (over $\C$).
\cite{Garg2018} showed that the rank-2 BL polytope includes the common base polytope of two linear matroids.
Indeed, we can observe that the rank-2 BL polytope is precisely the \emph{perfect fractional linear matroid matching polytope} \cite{VandeVate1992}, which includes the fractional perfect matching polytope and the common base polytope and is totally dual half-integral~\cite{Gijswijt2013}. The polytope obtained by removing the constraint $\sum_i n_i x_i = n$ (after rewriting the inequalities in down-monotone form) is known as the \emph{fractional linear matroid matching polytope}, and we refer to it as $P(B)$. The \emph{maximum fractional linear matroid matching problem} is to find $x \in P(B)$ that maximizes $\sum_i x_i$. This problem generalizes the maximum matroid intersection problem. The dual certificate, known as a \emph{minimum 2-cover}, is a pair $(S, T)$ of subspaces such that $\dim S + \dim T = k$ and $\dim (S \cap \row B_i) + \dim (T \cap \row B_i) \geq 2$ for all $i$, where $k$ is twice the optimal value and $\row B_i$ is the row-space of $B_i$. Finding the minimum 2-cover in the special case of linear matroid intersection amounts to the submodular minimization formulation of the dual of maximum matroid intersection, which in the special case of bipartite matching is equivalent to finding the minimum cover (the dual problem to finding the maximum matching). The BL polytope can also be seen as a slice of a \emph{moment polytope}, a certain polytope associated to a vector in a representation of a linear algebraic group. For more information about moment polytopes see \cite{burgisser2019towards}. Efficiently deciding membership in BL polytopes could lead to efficient algorithms for moment polytope membership, an open problem in computational invariant theory.

\cite{Chang2001b} gave an algorithm to solve the maximum fractional linear matroid matching problem with polynomially many calls to certain matroid oracles, and \cite{Gijswijt2013} showed how to use polynomially many calls to the algorithm of \cite{Chang2001b} to optimize linear functionals over the perfect fractional linear matroid matching polytope. Nonetheless, no algorithm for membership in rank-2 BL polytopes is known to be polynomial time! This is because the recursive nature of the oracle calls in \cite{Chang2001b,Gijswijt2013} na\"ively lead to bit complexity explosions of intermediate numbers.

The difficulty of optimization on BL polytopes is perhaps unsurprising. It is unclear how to decide if a given inequality is valid for the polytope and in fact, it is not known that membership in the rank-2 BL polytope is in \textsc{coNP}!  One could hope that a violating subspace could act as a \textsc{coNP} certificate for membership, but it is not known whether every inequality can be witnessed by a subspace of polynomial bit complexity.

As an application of our modified Sinkhorn algorithm, we prove the following two results in Section \ref{sec:bl}.
\begin{theorem}[informal, see~Theorem~\ref{thm:fracmat}]
There is a deterministic algorithm that finds a $(1-\eps)$-maximum fractional linear matroid matching in $\tilde{O}(n^5(n+p))$ time for any $\eps > 0$, where $n$ is the dimension of the ground space and $p$ is the number of matrices $B_i$.
\end{theorem}

\begin{theorem}\label{thm:bl-conp}
Membership in the rank-2 BL polytope is in \textsc{NP} $\cap$ \textsc{coNP}.
\end{theorem}

The idea of the first result is as follows.
First, we use the recent connection between a matrix space spanned by rank-2 skew-symmetric matrices and fractional linear matroid matching~\cite{Oki2022}.
They showed that the noncommutative rank of such matrix spaces is twice the size of the maximum fractional linear matroid matching.
We show that by applying our operator scaling algorithm to the matrix space, we can obtain a $(1-\eps)$-maximum fractional linear matroid matching.

Let us consider the second result.
It is easy to show that the membership in the rank-2 BL polytope is in \textsc{NP}.
Any $x \in BL(B)$ can be decomposed into a convex combination of half-integral vertices, so these vertices and convex coefficients form a certificate with polynomial bit complexity.
Note that the membership of a half-integral vector can be checked in polynomial time by the result of \cite{Garg2018} because the common denominator is a constant.
In contrast, a \textsc{coNP} certificate is more intricate. 
As mentioned above, a violating subspace may not have bounded bit complexity.
However, we can exploit half-integrality of $BL(B)$ (in the rank-2 case) and the fact that all facet-defining inequalities have coefficients in $\{0,1,2\}$ to provide a \textsc{coNP} certificate. The certificate for $x\notin BL(B)$ are $p$ affinely independent half-integral points $x^{(i)}$ ($i=1,\cdots,p$) in $BL(B)$ satisfying a valid inequality for $P(B)$ (and hence for $BL(B)$) at equality, with $x^*$ violating this inequality. To verify that the corresponding inequality is valid, the verifier needs to also check that a point close to the centroid of the $x^{(i)}$'s is not in $P(B)$, and this (and the membership of the $x^{(i)}$'s) can be checked efficiently using \cite{Garg2018} since the denominators are small. Details are given in a proof at the end of Section \ref{sec:scaling}. 

\subsubsection*{A pseudopolynomial weighted algorithm for rank-2 $P(B)$}
Another option for the \textsc{coNP} certificate is through a formulation of the weighted fractional linear matroid matching problem as another scaling problem similar to the one used for the shrunk subspace problem.
A natural extension of the Sinkhorn-style algorithm for the scaling problem (Algorithm \ref{alg:weighted-sinkhorn}) yields an algorithm that runs in time polynomial in the size of $B_i$ and the magnitude of weight:
\begin{theorem}\label{theorem:opt-pb} There is a deterministic algorithm to compute a point $x$ such that $w^T x\geq\max\{w^T y: y \in P(B)\} - \eps$ in time polynomial in the input size of $B_i$, $\norm{w}_1$, and $1/\eps$.
\end{theorem}
\noindent Theorem \ref{thm:weighted-sinkhorn-decision} gives the precise time complexity. Using this algorithm, one can define an appropriate \textsc{coNP} certificate that can be verified in polynomial time.

\subsubsection*{An $\eps$-membership algorithm for rank-2 $P(B)$}
We have mentioned the link between the BL polytope and the moment polytope. \cite{burgisser2019towards} gave a randomized algorithm to decide the $\eps$-membership problem for the moment polytope (to decide whether a point is $\eps$-far from the moment polytope or $\eps$-far from its complement) that runs in polynomial time in the input size and $1/\eps$. Their analysis showed that the dependence on the common denominator in \cite{Garg2018} is essentially not necessary. If one could show that the moment polytopes have a polynomial radius and inverse polynomial in-radius, a randomized algorithm for $\eps$-optimization in time polynomial in $1/\eps$ and the input size would follow by general reductions between membership, separation, and optimization \cite{lee2018efficient}.

However, no inverse polynomial lower bound is known for the in-radius of moment polytopes, and even if this were the case, the BL polytope is only a slice of the full moment polytope. Nonetheless, we directly (without using Theorem \ref{theorem:opt-pb}) show the following:
\begin{theorem}\label{thm:bl-membership}
There is a deterministic algorithm (Algorithm \ref{alg:bl-membership}) to decide the $\eps$-membership problem for rank-2 $P(B)$ in time polynomial in the input size and $1/\eps$.
\end{theorem}
The algorithm is yet another variant of the Sinkhorn-style algorithm described in \cite{Garg2018}. As with \cite{burgisser2019towards}, there is no dependence on the common denominator for the point in question. As $P(B)$ contains the origin and the simplex, and hence has polynomial in-radius, this result and \cite{lee2018efficient} imply a different algorithm, albeit a much slower randomized one, with guarantees matching Theorem \ref{theorem:opt-pb}. Without Theorem \ref{theorem:opt-pb}, Theorem \ref{thm:bl-membership} with only implies membership in $P(B)$ is in \textsc{coMA} due to the randomization in \cite{lee2018efficient}.

\subsection{Related work}
Matrix scaling is a classic problem commonly appearing in various fields such as optimization, statistics, economics, and physics.
See the comprehensive survey of Idel~\cite{Idel2016} and references therein.
The Sinkhorn algorithm is the standard algorithm for matrix scaling.
Although the Sinkhorn algorithm converges to an approximate doubly stochastic matrix if a given matrix admits such scaling, it is nontrivial what happens if it diverges.
\cite{Gietl2013} showed that the Sinkhorn algorithm has only two accumulation points by information geometry with the Kullback-Leibler divergence (KL-divergence).
\cite{Aas2014a} showed that the accumulation points exhibit a block diagonal structure which can be characterized combinatorially.
Recently, Hayashi and Hirai~\cite{Hayashi2022} proved that one can obtain a Hall blocker (dual certificate showing that a given matrix has no scaling) by the Sinkhorn iteration.

Their analysis also used information geometry with the KL-divergence.
However, it is unknown whether a similar result holds in operator scaling.
A major obstacle is that the underlying geometry of operator scaling is unknown.
The operator Sinkhorn algorithm is not an alternating minimization of the KL-divergence or other well-known divergences in quantum information theory~\cite{Matsuda2022}. We discuss in Section \ref{sec:warm-up} how, when specialized to matrix scaling case, our algorithm results in a quite different method to find a Hall blocker from the Sinkhorn iteration.

\paragraph{Organization of the paper.}
The paper is organized as follows.
In Section~\ref{sec:pre}, we introduce necessary concepts and definitions.
Section~\ref{sec:highlevel} describes our algorithms at a high level first for the simpler matrix scaling setting and then for the operator scaling setting.
In Section~\ref{sec:maj}, we introduce majorized operator scaling and our Sinkhorn-style algorithm.
Section~\ref{sec:decsion} discusses how to find an approximate shrunk subspace with Sinkhorn iteration and Section~\ref{sec:round} shows how to round it to the smallest $c$-shrunk subspace.
In Section~\ref{sec:bl}, we present applications of our algorithms to fractional linear matroid matching and the membership problem in the BL polytope.

\section{Definitions and preliminaries}\label{sec:pre}

\paragraph{Notation.}
Throughout the paper, $p$ will denote the number of Kraus operators.
Note that if the Kraus operators are $m \times n$ matrices then we can assume without loss of generality that $p \leq m n$. An upper bound on any entry of any of the Kraus operators is denoted by $M$.
For a positive integer $i$, we define $[i] := \{1, \dots, i\}$.
For a vector $a$, the Euclidian norm is denoted by $\norm{a}$ and the $\ell^p$ norm is by $\norm{a}_p$.
The all-one vector of length $n$ is denoted by $\bfone_n$.
For a complex matrix $A$, the Hermitian conjugate, transpose, and complex conjugate of $A$ is denoted by $A^\dagger$, $A^T$, and $\overline{A}$, respectively.
The Frobenius norm of $A$ is denoted by $\norm{A}_F$ and the operator norm (i.e., the largest singular value) is by $\norm{A}_1$.
We denote the orthogonal projection onto a subspace $U$ by $\pi_U$.
For a Hermitian matrix $A$, we denote by $\lambda(A)$ the vector whose entries are the eigenvalues of $A$ in nonincreasing order (unless stated otherwise).
The matrix square root and matrix exponential of $A$ are denoted by $\sqrt{X}$ (or $X^{1/2}$) and $e^A$, respectively.

\subsection{Matrix spaces and shrunk subspaces}

Let $\caA$ be a matrix space spanned by $A_1, \dots, A_p \in \C^{m\times n}$.
We present a few equivalent definitions of the noncommutative rank of $\caA$. The first three are easily seen to be equivalent; the fourth is non-trivial. Only the fourth hints at the reason for the name - the matrices $X_i$ can be thought of as non-commuting indeterminates.

\begin{definition}[noncommutative rank]
The \emph{noncommutative rank} of a matrix space $\caA = \langle A_1, \dots, A_p \rangle \subseteq \C^{m\times n}$ (denoted by $\ncrank\caA$) is equivalently defined as follows.
\begin{enumerate}
\item (Fortin and Reutenauer~\cite{Fortin2004}) Let $0_{s,t}$ denote an $s\times t$ zero matrix.
    \begin{align}\label{eq:FR}
    \ncrank\caA = \min\left\{m+n-s-t : P, Q \in GL(n),  PA_iQ = \begin{bmatrix}
        * & * \\
        * & 0_{s,t}
    \end{bmatrix}
    \, (i=1, \dots, p)
    \right\}
    \end{align}

\item (Hirai~and~Hamada~\cite{Hamada2021})\label{it:mvsp}
    \begin{align}
    \label{eq:MVSP}
        \ncrank\caA = \min\left\{m + n-\dim L-\dim R : \text{$L, R$ subspaces, $A_i(L, R) = \{0\}$ for all $i$}\right\},
    \end{align}
    where $A_i(L, R) = \{x^\dagger A_i y: x \in L, y \in R\}$.
    The minimization problem in~\eqref{eq:MVSP}, called the \emph{minimum vanishing subspace problem}, is submodular function minimization on the product of the lattice of all vector subspaces and its order-reversed lattice.
    Namely, if $(L, R)$ and $(L', R')$ both attain the minimum, so do their join $(L + L', R \cap R')$ and meet $(L \cap L', R + R')$.

\item (Ivanyos~et~al.~\cite{Ivanyos2015}, Garg~et~al.~\cite{Garg2019})
    \begin{align}
    \label{eq:shrunk}
        \ncrank\caA = \min\left\{n- \dim U + \dim \caA(U): \text{$U$ subspace of $\C^n$} \right\},
    \end{align}
    where $\caA(U) = \langle Au : A \in \caA, u \in U \rangle$.
    A subspace $U$ is called a \emph{$c$-shrunk subspace} of $\caA$ if $\dim U - \dim\caA(U) \geq c$.
    The right-hand side of \eqref{eq:shrunk} equals $\min\{n - c:  \text{$\exists$ $c$-shrunk subspace of $\caA$}\}$.

\item (Ivanyos, Qiao, and Subrahamnyam~\cite{Ivanyos2018})
\begin{align}
    \ncrank\caA = \sup_{d=1}^\infty \max\left\{\frac{1}{d} \rk\left(\sum_{i=1}^p A_i \otimes X_i \right) : X_i \in \C^{d\times d}\right\}.
    \label{eq:IQS}
\end{align}
The minimum $d$ that attains the supremum is called the \emph{minimum blow-up size} of $\caA$. It is known that the minimum blow-up size $d$ of any matrix subspace is at most $n-1$~\cite{Derksen2017}.
\end{enumerate}
\end{definition}

Feasible solutions of \eqref{eq:FR}, \eqref{eq:MVSP}, and \eqref{eq:shrunk} can be converted to each other.
Given a feasible solution $(P, Q)$ of \eqref{eq:FR}, let $L$ and $R$ be the subspaces spanned by the last $s$ columns and $t$ rows of $P$ and $Q$, respectively. Then, $(L, R)$ is a feasible solution of \eqref{eq:MVSP} with the same objective value.
If $(L,R)$ is a feasible solution of \eqref{eq:MVSP}, then $\caA(L) \subseteq R^\perp$, so any feasible solution of \eqref{eq:MVSP} gives that of \eqref{eq:shrunk} without increasing the objective value.
Finally, given a subspace $U$, one can make an $(n-\dim\caA(U)) \times \dim U$ zero block by an appropriate basis change.
Thus, the optimal value of \eqref{eq:shrunk} is at least that of \eqref{eq:FR}.

We focus in particular on Item \ref{it:mvsp}, the minimum vanishing subspace problem. In analogy with bipartite matching, if a $c$-shrunk subspace $U$ is analogous to a violation of Hall's condition with maximum defect, then the pair $((\mathcal A U)^\perp, U)$ is analogous to a maximum independent set. We give such sets a name.

\begin{definition}[Independent set]
A pair of subspaces $L\subseteq \C^n$ and $R\subseteq \C^m$ is said to be \emph{independent} if $L \subseteq (\mathcal A R)^\perp$, or equivalently if for all $v \in L, w \in R$ we have $$w^\dagger A_i v  =0.$$
The \emph{size} of an independent set is defined as $\dim L + \dim R$, and the maximum size of an independent set is equal to $m + n - \ncrank\caA$. %

The lattice of maximum independent sets has a unique maximum element $(L^*,R^*)$, which we call the \emph{dominant independent set}, so that $R^*$ is the smallest $c$-shrunk subspace. Let $\ell^* = \dim L^*, r^* = \dim R^*$, and $k^* = \ncrank\caA$ so that $\ell^* + r^* = n + m - k^*$.
\end{definition}

\subsection{Operator scaling}
For a tuple $A_1, \dots, A_p$ of matrices in $\C^{m \times n}$, we can associate the \emph{completely positive map} (CP map for short) $\Phi: \C^{n\times n} \to \C^{m \times m}$ given by
$$\Phi: X \mapsto \sum_{i = 1}^p A_i X A_i^\dagger$$
and $A_i, \dots, A_p$ are called the \emph{Kraus operators} for $\Phi$. The adjoint map $\Phi^* : \C^{m\times m} \to \C^{n \times n}$ is given by $\Phi^*: X \mapsto \sum_{i = 1}^p A_i^\dagger X A_i.$ For matrices $B \in \C^{m\times m}$ and $C \in \C^{n \times n}$, the operator $\Phi_{B,C}$ with Kraus operators $B A_1 C^\dagger, \dots, B A_p C^\dagger$ is called a \emph{scaling} of $\Phi$ by $B,C$. For $n= m$, the operator $\Phi$ is said to be $\eps$-\emph{doubly stochastic} if
\[
    \|\Phi(I_n) - I_n\|_F, \| \Phi^*(I_n) - I_n \|_F \leq \eps.
\]
We say that $\Phi$ is \emph{scalable to doubly stochastic} if $\Phi$ has an $\eps$-doubly stochastic scaling for every $\eps > 0$. There is a variational characterization of scalability through the following optimization problem over positive-definite matrices.
Define the \emph{capacity}\footnote{Remark that the original definition of the capacity by \cite{Gurvits2004} is $\capa'\Phi = \inf_{X \succ O}\frac{\det\Phi(X)}{\det X}$. One can check that $\capa\Phi = m + \log\capa'\Phi$.} of $\Phi$ as
\begin{align}
\capa \Phi = \inf_{X \succ 0, Y \succ 0} \tr \Phi (X^{-1}) Y^{-1} + \log \det X + \log \det Y.\label{eq:gurv}
\end{align}

Gurvits' theorem, which was proven independently in the context of quiver representations by King \cite{king1994moduli}, states the following.
\begin{theorem}[Gurvits~\cite{Gurvits2004}]\label{thm:gurv} For a CP map $\Phi: \C^{n\times n} \to \C^{n \times n}$, the following are equivalent:
\begin{enumerate}
\item $\ncrank\caA = n$, where $\caA$ is the matrix space spanned by the Kraus operators of $\Phi$.
\item $\Phi$ is scalable to doubly stochastic.
\item $\capa \Phi > - \infty.$
\end{enumerate}
\end{theorem}

\cite{Garg2019} showed that the capacity can be computed efficiently by a simple alternating minimization algorithm, which shows that deciding whether the noncommutative rank is full can be done in polynomial time.

\begin{remark}[Geodesic convexity] The objective function in the above optimization problem is not jointly convex in $X$ and $Y$, nor is it known to become convex under any simple change of variables. However, it is \emph{geodesically convex}, i.e. convex along the curves of the form $\sqrt{X} e^{H_1 t} \sqrt{X}, \sqrt{Y} e^{H_2 t} \sqrt{Y}$ through $(X,Y)$. These curves are the geodesics for a certain Riemannian metric on the manifold of positive definite matrices \cite{bhatia2009positive}. For fixed $Y$, the function is (Euclidean) convex in $X$ under the change of variables $X \gets X^{-1}$ or $X \gets e^X$. The domains $X \succ 0$ and $X \succeq I_n$ are both Euclidean convex and geodesically convex. All the objective functions we consider will have both joint geodesic convexity and Euclidean convexity in each matrix considered alone. We will not need to make use of geodesic convexity until Section \ref{sec:bl}.
\end{remark}

Lastly, we describe the interaction between independent sets and CP maps. One verifies that $(L,R)$ is an independent set if and only if
\begin{align} \tr \pi_L \Phi ( \pi_R) = 0. \end{align}
Furthermore, we can think of $(L^\perp, R^\perp)$ as a ``cover'' for $\Phi$. In particular, by expanding $I_n = \pi_R + \pi_{R^\perp}$ and $I_m = \pi_{L} + \pi_{L^\perp}$ the above identity implies the following inequality:
\begin{align}\tr \pi_{R^\perp} \Phi^*(I_m)  + \tr \pi_{L^\perp} \Phi(I_n) \geq \tr \Phi(I_n).\label{eq:cover}\end{align}

\section{Algorithm and analysis}\label{sec:highlevel}
Here we describe our algorithm at a high level.
Before explaining it for shrunk subspace, we describe our ideas in a simpler case of matrix scaling.
Then, we will see how these ideas generalize to the shrunk subspace problem.

\subsection{Warm-up: Finding Hall blockers in matrix scaling}\label{sec:warm-up}

Let $A \in \R^{m \times n}$ be a nonnegative matrix and $G = ([m], [n]; E)$ be the support graph of $A$.
That is, $G$ is the bipartite graph with the vertex sets $[m], [n]$ and the edge set $E = \{ij: A_{ij}\neq 0\}$.
Our goal is to find a Hall blocker in $G$, i.e., a column subset $S \subseteq [n]$ such that $\abs{S} - \abs{\Gamma(S)} = n - k^*$, where $\Gamma(S)$ is the set of neighbors of $S$ and $k^*$ is the size of a maximum matching in $G$.
Note that this is equivalent to finding an independent set $(L, R)$ of $G$ such that $\abs{L} + \abs{R} \geq m + n - k^*$.

The classical Sinkhorn-Knopp theorem characterizes the scalability to a doubly stochastic matrix with the existence of a perfect matching in the support graph.
We start with a generalization of the Sinkhorn-Knopp theorem to rectangular matrices and maximum matching.
\begin{definition}[Substochastic scalability of matrices]
We say a nonnegative matrix $A \in \R^{m\times n}$ is \emph{doubly substochastic} if
\[
A\bfone_n \leq \bfone_m \text{ and } A^T \bfone_m \leq \bfone_n,
\]
where $\leq$ denotes the usual element-wise order.
We refer to $\sum_{i,j} A_{ij}$ as the \emph{size} of $A$.
We say that $A$ is \emph{$k$-scalable} if for every $\eps > 0$, there is a scaling of $A$ which is doubly substochastic and has size at least $k - \eps$.
\end{definition}

Similar to the scalability to doubly stochastic matrices, one can characterize the $k$-scalability by a convex optimization problem.
Let
\begin{align}
f_k(x,y,z)& := \sum_{i, j} A_{ij} e^{z - y_i - x_j} +  \sum_{j=1}^n x_j+ \sum_{i=1}^m y_i - kz\label{eq:simple-capacity-matrix}
\end{align}
and
\begin{align}
\caD:= \{ (x,y,z): x \geq \bfzero_n, y \geq \bfzero_m, z \geq 0\} \subseteq \R^{n} \times \R^{m} \times \R. \label{eq:domain-matrix}
\end{align}
We define the $k$-capacity of $A$ by $\capa_k(A) := \inf_{(x,y,z)\in \caD}f_k(x,y,z)$.
By convex analysis, one can prove the following theorem.
We do not provide the proof here; we will prove a generalization of this theorem for operator scaling in Theorem~\ref{thm:k-gurv}.

\begin{theorem}[Matrix scaling version of Theorem~\ref{thm:k-gurv}]
For a nonnegative matrix $A \in \R^{m\times n}$, the following are equivalent:
\begin{enumerate}
\item There exists an independent set $(L, R)$ of $G$ with $\abs{L} + \abs{R} \geq m + n - k$.
\item $A$ is $k$-scalable.
\item $\capa_k A > - \infty.$
\end{enumerate}
\end{theorem}

One can check the finiteness of $\capa_k A$ and find a $k$-scaling (if it exists) with a variant of the Sinkhorn algorithm.
In each iteration, we first scale up the whole matrix so that its size is at least $k$.
Then, if the scaled matrix violates the column sum constraint (resp. the row sum constraint), we scale down columns (resp. rows) to satisfy the constraint.
We will repeat this until we find a $k$-scaling of size at least $(1-\eps)k$.
Indeed, this is equivalent to applying alternating minimization of $f_k(x,y,z)$.
One can show that if $A$ is $k$-scalable, then this algorithm finds a desired $k$-scaling in polynomially many iterations.
To find the cardinality $k^*$ of the maximum matching, we can simply do a binary search on $k$.

\subsubsection*{Finding a Hall blocker with Sinkhorn iteration}
If $A$ is not $k$-scalable, then the algorithm diverges even for $k'=k -1/2$.
Suppose that we have $(x,y,z) \in \caD$ such that $f_{k'}(x, y, z)$ is very small.
Can we find an independent set of size greater than $m + n - k$ from $(x, y, z)$?
Indeed, one can obtain such an independent set by a simple rounding method.
\begin{lemma}
Suppose that $(x,y,z)\in\caD$ satisfies $f_{k'}(x, y, z) \leq -C$ for some sufficiently large $C > 0$ and $k':=k-1/2$.
Then, there exists an independent set of size greater than $m+n-k$ and one can efficiently find it from $(x, y, z)$.
\end{lemma}
\begin{proof}
Let us first find a \emph{fractional cover} of size at most $k$. Recall that a fractional cover of $G$ is a pair $\tilde x, \tilde y \in \R^n_{\geq 0}$ such that $\tilde x_i + \tilde y_j \geq 1$ for $i,j \in G$ and its size is $\sum_i \tilde x_i + \sum_j \tilde y_j$; the minimum size of a fractional cover is dual to the maximum size of a perfect matching.

Let $S:=k'z - \sum_j x_j - \sum_i y_i$.
Then $S$ is at least $C$ because the remaining term in $f_{k'}$ is always positive.
Furthermore, $-x_i - y_j + z \leq \ln(S/A_{ij})$ for any $ij \in E$.
Set $\tilde{x} = x/z, \tilde{y} = y/z$.
Then $\tilde{x}_j + \tilde{y}_i \geq 1 - \ln(S / A_{ij})/z \geq 1 - k' \ln (S/A_{ij})/S$, and $\sum_j \tilde{x}_j + \sum_i \tilde{y}_i \leq k' - S/z < k'$. By making $C,$ and hence $S$, large enough and renormalizing $\tilde x, \tilde y$ slightly, we can further convert this to a fractional cover of value less than $k$.

Once we obtain a fractional cover $(\tilde x, \tilde y)$ of size less than $k$, we can round it to an independent set of size greater than $m+n-k$ as follows.
For simplicity, assume that $k$ is an integer (even if $k$ is not an integer, a similar argument works with $\lfloor k \rfloor$ instead of $k$).
Sort the rows and columns so that $\tilde x_1 \geq \dots \geq \tilde x_n$ and $\tilde y_1 \geq \dots \geq \tilde y_m$.
Consider the set $I$ of indices $(i, j) \in [m] \times [n]$ such that $i+j = k + 1$.
Note that $\abs{I} = k$.
Then we have
$\frac{1}{k} \sum_{(i,j) \in I} (\tilde x_j + \tilde y_i) \leq \frac{1}{k} (\sum_j \tilde x_j + \sum_i \tilde y_i) < 1$.
Therefore, there exists $(i, j) \in I$ such that $\tilde x_j + \tilde y_i < 1$.
Since $\tilde x_j, \tilde y_i$ are non-increasing, we have $\tilde x_{j'} + \tilde y_{i'} < 1$ for $i' \geq i$, $j' \geq j$.
Because $(\tilde x, \tilde y)$ is a fractional cover, this implies that $L =\{i,\dots, m\}$, $R = \{j,\dots, n\}$ form an independent set.
The size of $(L,R)$ is $(m - i + 1) + (n - j + 1) = m + n - k + 1$.
\end{proof}

Therefore, running the variant of the Sinkhorn algorithm for $k = k^* + 1$, we obtain a maximum independent set of size $m+n-k^*$ witnessing that the size of maximum matching is $k^*$.

\subsubsection*{Finding the smallest Hall blocker with perturbed capacity}
Now let us consider a slightly more refined problem.
Can we find the \emph{smallest} minimizer $S^*$ of the surplus function $\abs{\Gamma(S)} - \abs{S}$?
Note that the surplus function is submodular and hence the minimizers form a distributive lattice, so the smallest minimizer exists.
Again, this is equivalent to finding the independent set $(L^*, R^*)$ such that $\abs{L^*}+\abs{R^*} = m + n - k^*$ and $\abs{R^*}$ is smallest.

By the Dulmage-Mendelsohn decomposition, $R^*$ is the set of vertices exposed (i.e., not covered) by some maximum matching (see, e.g., \cite[Section~2.2.3]{Murota2009book}).
Therefore, there exists a $k^*$-scaling of $A$ such that the $i$th column sum is at most $1 - 1/n$ for $i \in R^*$ and at most $1$ elsewhere.
For a fixed candidate of $R^*$, we can check if there is such a scaling by matrix scaling with specified marginals.
The obvious problem here is that the number of candidates is exponential.

We can address this issue with a perturbed version of $k$-capacity.
Let $r^*$ be the size of $R^*$.
We say $(k, r)$ is \emph{violated} by an independent set $(L,R)$ either if $\abs{L}+\abs{R} > m + n - k$ or if $\abs{L}+\abs{R} = m + n - k$ and $\abs{R} > r$.
The key observation here is that one can also check whether $(k,r)$ is violated by $(L^*, R^*)$ or not by a similar scaling problem as above.
Let
\begin{align}\label{eq:perturb-capacity-matrix}
f_{k,r}(x,y,z)& := \sum_{i, j} A_{ij} e^{z - y_i - x_j} +  \sum_{j=1}^r x^\downarrow_j + \sum_{j=r+1}^ m \left(1 - \frac{1}{n} \right) x^\downarrow_j + \sum_{i=1}^m y_i - kz,
\end{align}
where $x^\downarrow_1 \geq \dots \geq x^\downarrow_n$ denotes the entries of $x$ in a non-increasing order.
Since $f_{k,r}(x,y,z) = f_k(x,y,z) - \frac{1}{n}\sum_{j=r+1}^m x^\downarrow_j$, the new function $f_{k,r}$ can be considered as a ``perturbed'' objective function.
Note that $f_{k,r}$ is still a convex function.
Let $\capa_{k,r} A := \inf_{(x,y,z) \in \caD} f_{k,r}(x,y,z)$.
Again, one can prove that $\capa_{k,r}$ gives the desired characterization of tuple $(k,r)$.

\begin{theorem}[Matrix scaling version of Theorem~\ref{thm:indep-scaling}]\label{thm:indep-scaling-matrix}
For a non-negative matrix $A \in \R^{m\times n}$ and $0 \leq k,r \leq n$ integers, the following are equivalent:
\begin{enumerate}
\item $(k,r) \leq (k^*,r^*)$ in lexicographic order, i.e. there is no independent set violating $(k,r)$.
\item $\capa_{k,r} A > -\infty$.
\end{enumerate}
\end{theorem}

We can tweak the Sinkhorn algorithm to check the finiteness of $\capa_{k,r} A$ (and even find the desired scaling if it exists) in polynomial time.
So we can find $r^*$ by binary search.
Furthermore, by a similar rounding argument as above, given $(x,y,z) \in \caD$ with sufficiently small $f_{k,r}(x,y,z)$, it is possible to find an independent set $(L, R)$ violating $(k, r)$.
So, applying it for $k = k^*$ and $r = r^* + 1$, one can find the desired smallest minimizer $R^*$.

\subsubsection*{Majorization scaling of matrices}
The above scaling problem---finding a $k$-scaling with the column sums at most $1 - 1/n$ for $r^*$ columns and at most $1$ elsewhere---can be cast as a more general scaling problem, which we call \emph{majorization scaling}.

\begin{definition}[Majorization scalability of matrices]
For a vector $x \in \R^n$, let $x^\downarrow \in \R^n$ denote $x$ sorted in non-increasing order. The vector $x \in \R^n$ is \emph{weakly majorized} by $y \in \R^n$, denoted by $x \preceq y$, if
$$ \sum_{i = 1}^s x^\downarrow_i \leq \sum_{i = 1}^s y_i^\downarrow $$ for each $s \in [n]$.
For non-increasing vectors $\alpha \in \R^n$ and $\beta \in \R^m$, we say $A$ is $(\alpha,\beta)$-\emph{majorized} if
\[
    A\bfone_n \preceq \alpha, \quad A^T \bfone_m \preceq \beta.
\]
We say $A$ is $k$-\emph{scalable} to $(\alpha, \beta)$ if for every $\eps > 0$, there is a scaling of $A$ which is $(\alpha,\beta)$-majorized and has size at least $k - \eps$.
\end{definition}

Interestingly, the majorization scalability is also characterized by convex optimization.
Define
\begin{align}
f_k^{\alpha,\beta}(x,y,z)& := \sum_{i, j} A_{ij} e^{z - y_i - x_j} +  \sum_{j=1}^n \alpha_j x^\downarrow_j + \sum_{i=1}^m \beta_j y^\downarrow_i - kz,
\end{align}
and define the $(\alpha, \beta)$-capacity of $A$ by $\capa_k^{\alpha,\beta}:= \inf_{(x,y,z)\in\caD} f_{k}^{\alpha,\beta}(x,y,z)$.
Note that the above $(k,r)$-capacity is the special case
$$\alpha = (\underbrace{1,\dots, 1}_{\text{$(n-r)$ times}}, \underbrace{1-1/n, \dots, 1-1/n}_{\text{$r$ times}}), \quad \beta = \bfone_m.$$
Again by convex optimization, we can show the following theorem.

\begin{theorem}[Matrix scaling version of Theorem~\ref{thm:maj-scal}]\label{thm:maj-scal-matrix}
For a non-negative matrix $A \in \R^{m\times n}$ and $0 \leq k,r \leq n$ integers, the following are equivalent:
\begin{enumerate}
\item there exists a $k$-scaling to $(\alpha, \beta)$
\item $\capa_{k}^{\alpha,\beta} A > -\infty$.
\end{enumerate}
Furthermore, if there exists a $k$-scaling to $(\alpha, \beta)$, we can find it with a variant of the Sinkhorn algorithm.
\end{theorem}

We will prove a generalization of this theorem to operator scaling in Section~\ref{sec:maj}.
To our knowledge, even for the matrix scaling setting, our majorization scaling is novel.

\subsection{High-level description of algorithm}
Now we describe our algorithm (Algorithm \ref{alg:scaling-shrunk}, \textsc{RoundSubspaces}) for the shrunk subspace problem at a high level, together with the theorems and concepts needed to analyze it.

Before doing so, we need to define an approximate notion of shrunk subspaces. We define approximate shrunk subspaces by instead defining approximate independent sets.

\begin{definition}[Approximate independent set]\label{def:approx-indep} Say a pair of subspaces $L\subseteq \C^n, R \subseteq \C^m$, is $\eps$-\emph{independent} if for all $v \in L, w \in R$ we have $$|w^\dagger A_i v|  \leq \eps  \norm{w} \cdot \norm{v}.$$
\end{definition}
Our algorithm will find the value of $k^* = \ncrank \mathcal A$ and then the value of $r^*$, the minimum value of $\dim R$ for $(L,R)$ a maximum independent set. For a candidate pair of values $(k,r)$, we say $(k,r)$ is \emph{violated} by a pair of subspaces $(L,R)$ if either $\dim L + \dim R > m + n - k$ or if $\dim L + \dim R = m + n - k$ and $\dim R < r$. Note that $(k^*, r^*)$ is the largest pair (in lexicographic order) with no violating independent set, so if we can test whether $(k,r)$ has a violating independent set we can find $k^*$ and $r^*$ using binary search.

\subsubsection*{Characterization of noncommutative rank}
We begin by characterizing the noncommutative rank with optimization.
Like the Sinkhorn-Knopp theorem, Gurvits' theorem (Theorem \ref{thm:gurv}) doesn't directly answer the problem of distinguishing if $\ncrank\mathcal A \geq k$ for other $1 \leq k \leq n$. There is a trick~\cite[Theorem~A.4]{Garg2019} to reduce to the case $k = n$, which is analogous to adding dummy vertices to a bipartite graph in order to reduce maximum matching to perfect matching. Here we show how to do it directly using a modified scaling problem.

\begin{definition}[Substochastic scalability]
We say a CP map $\Phi: \C^{n\times n} \to \C^{m \times m}$ is \emph{doubly substochastic} if
$$ \Phi(I_n) \preceq I_m \text{ and } \Phi^*(I_m) \preceq I_n,$$
where $\preceq$ denotes the Loewner ordering. We refer to $\tr \Phi(I_n)$ as the \emph{size} of $\Phi$. We say that $\Phi$ is \emph{$k$-scalable} if for every $\eps > 0$, there is a scaling of $\Phi$ which is doubly substochastic and has size at least $k - \eps$.
\end{definition}

There is also a variational characterization of the $k$-scalability of CP maps analogous to that of non-negative matrices. Consider the domain
\begin{align}\mathcal D:= \{ (X,Y,z): X \succeq I_n, Y \succeq I_m, z \geq 0\} \subseteq \C^{n\times n} \times \C^{m \times m} \times \R,\label{eq:domain}\end{align}
and the following function
\begin{align}
f_k(X,Y,z)& := \tr (\Phi(X^{-1}) Y^{-1} e^{z}) +  \log\det(X)+ \log\det(Y) - kz. \label{eq:simple-capacity}
\end{align}
on $\caD$. Note that $\capa_n \Phi = \capa \Phi$. We will show the following generalization of Gurvits' theorem.
\begin{theorem} \label{thm:k-gurv}
For a CP map $\Phi: \C^{n\times n} \to \C^{m \times m}$, the following are equivalent:
\begin{enumerate}
\item $\ncrank\mathcal A \geq k.$
\item $\Phi$ is $k$-scalable.
\item $\capa_k \Phi > - \infty.$
\end{enumerate}
\end{theorem}
The above theorem characterizes the maximum size of an independent set as $m + n - k^*$, where $k^*$ is the largest $k$ such that $\capa_k \Phi > -\infty$. Because $\capa_n \Phi = \capa \Phi$, Theorem \ref{thm:k-gurv} implies Gurvits' theorem. We will not explicitly prove Theorem \ref{thm:k-gurv} because it is implied by stronger theorems stated later on (Theorems \ref{thm:indep-scaling} and \ref{thm:maj-scal}).

\subsubsection*{Characterization of $r^*$ and rounding}
We would also like to characterize $r^*$, i.e. how small $R$ can be in a maximum independent set $(L,R)$. For this we will need a perturbed version of $f_k$. For a Hermitian matrix $H \in \C^{n \times n}$, let $\lambda(H) \in \R^{n}$ denote the spectrum of $H$ in non-increasing order. Define the function
\begin{align}
f_{k,r}(X,Y,z) := f_k(X,Y,z) - \frac{1}{n} \sum_{i = n - r + 1}^n \lambda_i (\log(X)) \label{eq:perturb-capacity}
\end{align}
on $\caD$, and again define
$$ \capa_{k,r} \Phi = \inf_{\Upsilon \in \caD} f_{k,r}(\Upsilon).$$
We show that this optimization problem characterizes the dimensions of the dominant independent set. In fact one can use the slightly larger quantity $\capa_{k',r} \Phi$ for $k' = k - \frac{1}{2n}$; this will turn out to be helpful in finding independent sets.
\begin{theorem}\label{thm:indep-scaling}
For a CP map $\Phi: \C^{n\times n} \to \C^{m \times m}$ and $0 \leq k,r \leq n$ integers, the following are equivalent:
\begin{enumerate}
\item\label{it:lex} $(k,r) \leq (k^*,r^*)$ in lexicographic order, i.e. there is no independent set violating $(k,r)$.
\item\label{it:cap} $\capa_{k,r} \Phi > -\infty$.
\item\label{it:cap'} $\capa_{k',r} \Phi > -\infty$ for $k' = k - \frac{1}{2n}$.
\end{enumerate}

\end{theorem}
Note that $\capa_{k,0}\Phi = \capa_k \Phi$, so Theorem \ref{thm:indep-scaling} implies the equivalence between (1) and (3) in Theorem \ref{thm:k-gurv}. There will also be a scaling characterization analogous to of (2) in Theorem \ref{thm:gurv} and Theorem \ref{thm:k-gurv}, but the cap equivalence will hold at a more general level described in Section \ref{sec:maj}. Fortunately, there is an optimization algorithm to decide the finiteness of $\capa_k \Phi$ and $\capa_{k,r} \Phi$, and if $\capa_{k,r} \Phi = - \infty$, we can also use the optimization problem to find an approximate independent set $(L,R)$ that violates $(k,r)$.

Firstly, there is an optimization algorithm, \textsc{SinkhornDecision}$(\Phi, k,r)$ (Algorithm \ref{alg:sinkhorn-dec}), which can decide whether $\capa_{k,r} \Phi$ is finite in time $\poly(\log M, n, m)$. Formally:
\begin{theorem}\label{thm:decision-correct} \textsc{DecisionSinkhorn}$(\Phi, k,r)$ (Algorithm \ref{alg:sinkhorn-dec}) correctly decides whether $\capa_{k,r} \Phi > -\infty$ or not and takes
$O((m+n+p)(k(n + m)^4 \log(m + n) + n(n + m)^2 \log Mp))$
arithmetic operations if $r = 0$ and
$O((m+n+p)(kn^2(n + m)^4 \log(m + n) + n(n + m)^2 \log Mp))$
operations if $r > 0$. \end{theorem}

Secondly, there is another algorithm, \textsc{ApproxIndep}, which on input $\Phi$ with $\capa_{k,r} \Phi= - \infty$ outputs an $\eps$-independent set $(L,R)$ violating $(k,r)$ in time $\poly(\log M, n, m, \log(1/\eps))$.
\begin{theorem}\label{thm:indep-correct}
If $\capa_{k,r} \Phi= -\infty$, \textsc{ApproxIndep}$(\Phi,k,r,\eps)$ (Algorithm \ref{alg:approx-indep}) produces an $\eps$-independent set violating $(k,r)$ and requires
$O(k^2 n^3 (m+n)^3 (m+n+p) \log (2 e^2 n k^2/\eps))$
arithmetic operations.
\end{theorem}

Here it may be helpful to consider the difference between the operator scaling and matrix scaling cases. In the matrix scaling case discussed in Section \ref{sec:warm-up}, there was no need to round the independent set (Hall blocker). This is because of the discrete nature of the covers - they are finite subsets of indices. In the operator scaling case, however, the independent sets range over all subspaces. One can slightly perturb an independent set to obtain another approximate independent set - without some additional rounding, we can never be sure we didn't find one of these slightly perturbed versions of the true one. From another viewpoint, in the operator scaling case there are many possible eigenbases for approximate minimizers of the objective function.

If we choose $k = k^*$, $r = r^* + 1$ and $\eps$ small enough, but still inverse exponential, the next theorem shows the only approximate independent set that can be found is very close to the dominant independent set. For a subspace $S \subseteq \C^n$, let $\pi_S$ denote the orthogonal projection to $S$.

\begin{theorem}\label{thm:close} Suppose the Kraus operators of $\Phi$ have integer entries of absolute value at most $M$, and that $(L,R)$ is an $\eps$-independent set for $\Phi$ that violates $(k^*,r^* + 1)$. For $\eps \leq \eps_0 = e^{ - O (n (n + m) \log (m + n)   + n \log M)}$, we have
$$\|\pi_{L} - \pi_{L^*}\|_2, \|\pi_{R} - \pi_{R^*}\|_2  = e^{ O ( n^2(n + m) \log (m + n)   + n^2 \log Mp)} \eps=: M_0 \eps$$
\end{theorem}

The proof of this theorem is inspired by the modular lattice property of the maximum independent sets -  if an $\eps$-independent set violates $(k^*, r^* + 1)$ and it is sufficiently far from $(L^*,R^*)$, then we can combine them using something similar to the lattice operation to obtain an $\eps'$-independent set $(L,R)$ where $\dim R < \dim R^*$ and $\eps'$ is larger, but not too much larger, than $\eps$. The technical work of the proof is in defining the approximate version of the lattice operation and showing that $\dim R \geq \dim R^*$ for $\eps$ small enough.

Finally, we show that the orthogonal projection to $R^*$ has rational entries with exponential denominators. This shows we only need to find an approximate independent set that is exponentially close to the true one, and then we can round to the dominant independent set.

\begin{theorem}\label{thm:round} The entries of $\pi_{R^*}$ are rational numbers with common denominator at most $M_1 = e^{O(n^5 \log(Mn))}$.
\end{theorem}

Putting these things together, we obtain Algorithm~\ref{alg:scaling-shrunk} and the following main theorem of this paper.

\begin{Algorithm}[h]
Algorithm \textsc{RoundSubspaces}$(\caA)$:
\begin{description}
\item[\hspace{.2cm}\textbf{Input:}] A matrix space $\caA$.

\item[\hspace{.2cm}\textbf{Output:}] The dominant independent set $(L^*, R^*)$ for $\caA$.

\item[\hspace{.2cm}\textbf{Algorithm:}] Set constants $M_0 = e^{ O (n^2(n + m) \log (m + n)   + n^2 \log Mp)}$ and $M_1 = e^{O(n^5 \log(Mn))}$ as in Theorems \ref{thm:close} and \ref{thm:round}.
\end{description}
\begin{enumerate}
\item Using binary search over $k$ and \textsc{DecisionSinkhorn}$(\Phi, k,0)$, determine $k^*$, the largest $k$ such that $\capa_k \Phi > -\infty$.
\item Using binary search over $r$ and \textsc{DecisionSinkhorn}$(\Phi, k^*,r)$, determine $r^*$, the largest $r$ such that $\capa_{k^*,r} \Phi > -\infty$.
\item Using \textsc{ApproxIndep}$(\Phi, k^*,r^* +1, \eps)$ for $\eps = 1/2 M_0 M_1^2$, find an $\eps$-independent set $(L,R)$ that violates $(k^*, r^* + 1)$.
\item Using the continued fraction expansion, round each entry of $\pi_R$ to the nearest rational number with denominator at most $M_1$.
Let $R'$ be the image of the rounded matrix.
\item \textbf{Output} $(\caA (R')^\perp, R')$.
\end{enumerate}
\caption{Optimization algorithm to output the smallest shrunk subspace.}\label{alg:scaling-shrunk}
\end{Algorithm}

\begin{theorem}\label{thm:roundsub-correct} \textsc{RoundSubspaces} (Algorithm~\ref{alg:scaling-shrunk}) finds the dominant independent set (and hence smallest shrunk subspace) in
$O(n^7 (m+n)^3 (m+n+p)((m+n)\log(m+n) + n^3\log(Mnp)))$
arithmetic operations.
Furthermore, the bit complexity of intermediate numbers is $\poly(m,n,p,\log M)$.
\end{theorem}
\noindent As the input size is $mnp \log M$, the time complexity is at most the tenth power of the input size.

\begin{remark}[finite arithmetic]
In the following sections, we describe our algorithms assuming that we can compute the exact eigenvalue decomposition of Hermitian matrices and the exact values of simple functions, e.g., $\log$, $\exp$, and square root, for the simplicity of exposition.
In Appendix~\ref{sec:finite}, we show how to modify them to polynomial time algorithms in the \emph{finite arithmetic} model with only polynomially many bits.
\end{remark}

\section{Majorization scaling}\label{sec:maj}
To describe \textsc{SinkhornDecision}, we'll cast it as a special case of a more general problem we call \emph{majorization scaling}.

\begin{definition}[Majorization scalability]
For a vector $x \in \R^n$, let $x^\downarrow \in \R^n$ denote $x$ sorted in non-increasing order. The vector $x \in \R^n$ is \emph{weakly majorized} by $y \in \R^n$, denoted by $x \preceq y$, if
$$ \sum_{i = 1}^s x^\downarrow_i \leq \sum_{i = 1}^s y_i^\downarrow $$ for each $s \in [n]$.
For non-increasing vectors $\alpha \in \R^n$ and $\beta \in \R^m$, we say $\Phi$ is $(\alpha,\beta)$-\emph{majorized} if
$$\lambda(\Phi^*(I_m)) \preceq \alpha, \quad \lambda(\Phi(I_n)) \preceq \beta.$$
We say $\Phi$ is $k$-\emph{scalable} to $(\alpha, \beta)$ if for every $\eps > 0$, there is a scaling of $\Phi$ which is $(\alpha,\beta)$-majorized and has size at least $k - \eps$.
\end{definition}

\begin{definition}[Down-closure of permutahedron]
We denote
$$P_\alpha := \{x \in \R^n_{\geq 0}: x \preceq \alpha\};$$
 it is the down-closure of the permutahedron of $\alpha$.
\end{definition}
Note that $\Phi$ is $k$-scalable if and only if it is $k$-scalable to $(\mathbf 1_n, \mathbf 1_m)$, the appropriate length all-ones vectors. We'll show a variational characterization of majorization scalability as well. For decreasing vectors $\alpha,\beta \in \R^d$ and $k\geq 0$, consider the following objective function:
\begin{align}
f^{\alpha,\beta}_k(X,Y,z) = \tr \Phi(X^{-1}) Y^{-1} e^{z} +  \alpha \cdot \lambda(\log X) + \beta \cdot \lambda(\log Y) - kz, \label{eq:maj-capacity}
\end{align}
and define the $\alpha,\beta$ capacity as
\begin{align*}
\capa^{\alpha,\beta}_k \Phi = \inf_{\Upsilon \in \caD} f^{\alpha,\beta}_k(\Upsilon).
\end{align*}
Note that $f^{\mathbf 1_n, \mathbf 1_m}_k = f_k$ (and hence $\capa^{\mathbf 1_n, \mathbf 1_m}_k \Phi = \capa_k \Phi$) and if we set $$\alpha_r = (\underbrace{1, \dots, 1}_{n - r}, \underbrace{1 - 1/n, \dots, 1 - 1/n}_{r}),$$
i.e. $n - r$ ones followed by $r$ entries equal to $1 - 1/n$, then we have $f^{\alpha_r, \mathbf 1_m}_k = f_{k,r}$ and $\capa^{\alpha_r, \mathbf 1_m}_k \Phi = \capa_{k,r} \Phi.$ Note that $\capa^{\alpha,\beta} \Phi$ is \emph{unitarily invariant}, i.e. $\capa^{\alpha,\beta} \Phi_{U,V} = \capa^{\alpha,\beta} \Phi$ for any $U,V$ unitary. The finiteness of $\capa^{\alpha,\beta}$ characterizes majorization scalability:
\begin{theorem}\label{thm:maj-scal}
$\capa^{\alpha,\beta} \Phi > -\infty$ if and only if $\Phi$ is $k$-scalable to $(\alpha, \beta)$.
\end{theorem}

The proof goes through a notion of scaling by triangular matrices considered in \cite{Franks2018}. We postpone the proof to Section \ref{sec:tri} because it is not needed for the description of our majorization scaling algorithm. Note that Theorem \ref{thm:maj-scal} implies the equivalence between (2) and (3) of Theorem \ref{thm:k-gurv}.

\subsection{Sinkhorn algorithm for majorization scaling}\label{sec:maj-sinkhorn}

Here we describe a Sinkhorn-style algorithm, \textsc{MajSinkhorn}, for finding a scaling of $\Phi$ to $(\alpha, \beta)$ with size $k$.
The idea is alternating minimization of the function $f_k^{\alpha, \beta}$ (see \eqref{eq:maj-capacity} for definition).
We describe how to update $X, Y, z$ efficiently.

\subsubsection*{Update of $X$.}
First, we describe the subroutine \textsc{KLProject} for how to update $X$ while fixing $Y$ and $z$.
It suffices to solve
\begin{alignat*}{3}
    &\text{minimize} & \quad & \tr(X^{-1}C) + \alpha \cdot \lambda(\log X) \\
    &\text{subject to} & \quad & X \succeq I.
\end{alignat*}
where $C = e^z\Phi^*(Y^{-1}) \succeq O$.
Let $X = Ue^{\diag(x)}U^*$, where $U$ is an unitary matrix and $x \geq \bfzero$ is a nonnegative vector with nondecreasing entries, i.e., $x_1 \geq \dots \geq x_n \geq 0$.
We can get rid of $U$ using $\tr(X^{-1}C) = \tr(Ue^{\diag(-x)}U^*C) \geq e^{\diag(-x)}\cdot\lambda(C)$, where the equality is attained when $U$ diagonalizes $C$, and the other term $\alpha \cdot \lambda(\log X)$ does not depend on $U$.
Thus, we obtain an equivalent optimization problem only in $x$.
\begin{equation}\label{eq:Sinkhorn-X}
\begin{alignedat}{3}
    &\text{minimize}   & \quad & \sum_{i=1}^n \lambda_i(C) e^{-x_i} + \alpha_ix_i \\
    &\text{subject to} & \quad & x_1 \geq \dots \geq x_n \geq 0.
\end{alignedat}
\end{equation}
We show that this problem is the dual of generalized KL-projection of $\lambda(C)$ onto $P_\alpha$, (the down-closure of) the permutahedron of $\alpha$.

\begin{definition}[Generalized KL divergence]
For two nonnegative vectors $p, q \in \R^n$, define
$$D(p|| q) = \sum_{i=1}^n p_i \log\frac{p_i}{q_i} - p_i + q_i$$
define $D(P_\alpha|| q)$ to be $\min_{p \in P_\alpha} D(p|| q)$.
\end{definition}

KL-projection onto the permutahedron of $\alpha$ can be solved in $O(n^2)$ time, see e.g., \cite{Suehiro2012}.
Furthermore, we can construct the optimal $x$ from a dual optimal solution.

\begin{lemma}\label{lem:dual}
    Let $q = \lambda(C)$. The dual of \eqref{eq:Sinkhorn-X} is equivalent to the following convex optimization:
    \begin{equation}\label{eq:KL-permuta}
    \begin{alignedat}{3}
        &\text{minimize}   & \quad & D(p|| q) \\
        &\text{subject to} & \quad & p(S) \leq \sum_{i=1}^{\abs{S}} \alpha_i \quad (S \subseteq [n]) \\
        &                  & \quad & p \geq 0,
    \end{alignedat}
    \end{equation}
i.e. the generalized KL projection of $q$ onto $P_\alpha$.
In particular, the value of $\eqref{eq:Sinkhorn-X}$ is $\tr C - D(P_\alpha|| q)$.
Furthermore, we can solve \eqref{eq:Sinkhorn-X} as well as \eqref{eq:KL-permuta} in $O(n^2)$ time.
\end{lemma}
\begin{proof}
    Note that $q_1 \geq \dots \geq q_n$.
    First observe that any optimal solution $p^*$ also obeys the same ordering, i.e., $p^*_1 \geq \dots \geq p^*_n$.
    To see this, suppose for the contrary that $p_i^* < p_j^*$ for some $i < j$.
    Let $p$ be a vector obtained by swapping the $i$th entry with the $j$th entry in $p^*$.
    Then, $p$ is also feasible and
    \[
        D(p || q) - D(p^* || q)
        = (p^*_j - p^*_i) \log\frac{p^*_i}{q_i} + (p^*_i - p^*_j)\log\frac{p^*_j}{q_j}
        \leq (p^*_j - p^*_i) \log\frac{p^*_i}{q_j} + (p^*_i - p^*_j)\log\frac{p^*_j}{q_j}
        = (p^*_i - p^*_j) \log\frac{p^*_j}{p^*_i} < 0,
    \]
    so $D(p || q) < D(p^* || q)$, which is a contradiction.

    Hence, it suffices to consider constraints $\sum_{i=1}^j p_i \leq \sum_{i=1}^{j} \alpha_i$ for $j=1,\dots, n$.
    The Lagrangian of the problem \eqref{eq:KL-permuta} is
    \begin{align*}
        L(p, y) = \sum_{i=1}^n \left( p_i \log\frac{p_i}{q_i} - p_i + q_i \right) + \sum_{j=1}^n y_j \sum_{i=1}^j (p_i - \alpha_i),
    \end{align*}
    where $y_j \geq 0$ ($j=1,\dots,n$) is a dual variable.
    Solving $\frac{\partial L}{\partial p_i} = 0$, we obtain
    $
        p_i = q_i \exp\left(-\sum_{i=j}^n y_j\right).
        $
    Substituting this, we obtain dual problem:
    \begin{equation}\label{eq:KL-permuta-y}
    \begin{alignedat}{3}
        &\text{maximize}   & \quad & -\sum_{i=1}^n q_i \exp \left(-\sum_{j=i}^n y_j\right) - \alpha_i \sum_{j=i}^n y_j + \sum_{i=1}^n q_i \\
        &\text{subject to} & \quad & y_1, \dots, y_n \geq 0.
    \end{alignedat}
    \end{equation}
    Chainging a variable to $x_i = \sum_{j=i}^n y_j$, we obtain \eqref{eq:Sinkhorn-X} (with an extra constant $\sum_i q_i = \tr C$ in the objective function).
    Note that \eqref{eq:Sinkhorn-X} satisfies the Slater condition, so the strong duality holds.

    The algorithm for solving \eqref{eq:Sinkhorn-X} works as follows.
    First, it computes the KL-projection of $q$ onto the permutahedron of $\alpha$ as well as the corresponding dual variable $y'$, which can be found in $O(n^2)$ time by e.g. \cite{Suehiro2012,Lim2016a}.
    Then, the optimal solution $y$ of \eqref{eq:KL-permuta-y} can be obtained by setting negative entries of $y'$ to zero~\cite[Theorem~3.3]{Lim2016a}.
    Once we have an optimal solution $y$ of \eqref{eq:KL-permuta-y}, the optimal solution $x$ of \eqref{eq:KL-permuta} is given by $x_i = \sum_{j=i}^n y_j$. \end{proof}

    \begin{remark}[The case $\alpha = \mathbf 1_n$]\label{rem:simple-update} In this case, the update step is very simple. The algorithm reduces to the variant of the Sinkhorn algorithm mentioned in Section \ref{sec:warm-up} - one simply chooses $X$ to be diagonal in the same basis as $C$, and then sets each eigenvalue $x_i \geq 1$ of $X$ to scale down the corresponding eigenvalue $\lambda_i$ of $C$ so that $x_i^{-1} \lambda_i \geq 1$ (and if already $\lambda_i \leq 1$, $x_i = 1$). In the matrix scaling case, this amounts to scaling down the columns which have sums larger than 1 so that they have unit sum.
    \end{remark}

Computing the matrix $C$ takes $O(m^2n + mn^2 + pn^2)$ time and computing the eigen-decomposition of $C$ takes $O(n^3)$ time.
Overall, the time complexity of update of $X$ is $O(n^3 + m^2n + mn^2 + pn^2) = O((m+n)^3 + pn^2)$ time.

\subsubsection*{Update of $Y$.}
We can also update $Y$ in $O((m+n)^3 + pm^2)$ time mutatis mutandis.

\subsubsection*{Update of $z$.} Let $a = \tr(\Phi(X^{-1})Y^{-1}).$ Set $z \gets\log(k/a)$ if $a \leq k$. Otherwise, set $z = 0$.

\begin{Algorithm}
Algorithm \textsc{MajSinkhorn}$(\Phi, \alpha, \beta, \eps)$:
\begin{description}
\item[\hspace{.2cm}\textbf{Input:}] A CP map $\Phi:\C^{n\times n} \to \C^{m\times m}$, non-increasing vectors $\alpha \in \R^n, \beta \in \R^m$, and a parameter $\eps > 0$.

\item[\hspace{.2cm}\textbf{Output:}] A scaling of $\Phi$ that is $(\alpha, \beta)$-majorized and has size at least $(1 -  \eps)k$.

\item[\hspace{.2cm}\textbf{Algorithm:}] Set $X = I_n, Y = I_m, z = 0$. For $t \in \{0,..,T\}$:
\end{description}
\begin{enumerate}
\item
Let $a = \tr \Phi(X^{-1})Y^{-1}$. Set $z \gets\log(k/a)$ if $a \leq k$. Otherwise, set $z = 0$.
\item Set $\mu = \lambda( e^{z} X^{-1} \Phi^*(Y^{-1}))$ and $\nu = \lambda( e^{z} Y^{-1} \Phi(X^{-1}))$. Using \textsc{KLProject}, compute the divergences $\eps_X := D(P_\alpha || \mu)$ and $\eps_Y:=D( P_\beta || \nu)$, where $P_\alpha$ and $P_\beta$ are the down-closure of the permutahedron of $\alpha$ and $\beta$, respectively.
\begin{description}
\item[\textbf{If}]$\eps_X, \eps_Y \leq \min\{\alpha_1, \beta_1\} \frac{\eps^2}{2}$: \textbf{Output} $e^{z - \eps} \Phi_{Y^{-1/2}, Z^{-1/2}}.$
\item[\textbf{Else: }] $ $
\begin{description}
\item[\textbf{If}] $\eps_X \geq \eps_Y$: Using \textsc{KLProject}, set
$$X = \argmin_{X \succeq I_n} \tr(X^{-1} \Phi^*(Y^{-1}) e^z) + \alpha \cdot \lambda(\log X).$$
\item[\textbf{Else:}] Using \textsc{KLProject}, set
$$Y = \argmin_{Y \succeq I_m} \tr(Y^{-1} \Phi(X^{-1}) e^z) + \beta \cdot \lambda(\log Y).$$
\end{description}
\end{description}
\end{enumerate}
\caption{Algorithm for majorization scaling.}\label{alg:maj-scaling}
\end{Algorithm}

\subsection{Analysis of \textsc{MajSinkhorn}}
Here we describe the runtime of \textsc{MajSinkhorn}. The proof of our progress bounds requires an inequality relating products of matrices and spectra of logarithms by majorization.

\begin{lemma}\label{lem:subadd}
Let $X,Z \succ 0$ and $\alpha \in \R^n$ nondecreasing. Then
$$\alpha \cdot \lambda ( \log X^{1/2} Z X^{1/2}) \leq \alpha \cdot \lambda( \log Z) + \alpha \cdot \lambda ( \log X).$$
\end{lemma}
This is equivalent to $\lambda(\log (XZ)) \preceq \lambda(\log X ) + \lambda (\log Z)$ for $X,Z$ positive definite, which is due to Lidskii \cite[Corollary~III.4.6]{bhatia2013matrix}. We include a different proof in Section \ref{sec:tri} for completeness. We now analyze \textsc{MajSinkhorn}. The following lemma describes progress in terms of the KL projection.

\begin{lemma}\label{lem:prog-bound}
Let $X \succeq I_n$ and $C \succeq 0$ be $n \times n$ matrices, and let $\alpha \in \R^n$ be a non-increasing vector. Let $\mu = \lambda( X^{-1/2} C X^{-1/2}).$ Then
$$ \tr (X^{-1} C) + \alpha \cdot \log(X) \geq D(P_\alpha || \mu) + \tr C - D(P_\alpha|| \lambda(C)).$$
In particular, the progress made in the $X$ update of \textsc{MajSinkhorn} is at least
$D(P_\alpha || \mu) \geq \eps^2/2.$
\end{lemma}
\begin{proof} The proof is a change of variables argument. By Lemma \ref{lem:dual} and because $X \succeq I_n$, we can write
\begin{align*}
 \tr C - D(P_\alpha|| \lambda(C)) & = \inf_{\tilde X \succeq I_n} \tr (\tilde X^{-1} C) + \alpha \cdot \lambda(\log \tilde X)\\
&\leq\inf_{\tilde X \succeq X} \tr(\tilde X^{-1} C) + \alpha \cdot \lambda(\log \tilde X)\\
&=\inf_{Z \succeq I_n} \tr (Z^{-1} X^{-1/2} C X^{-1/2}) + \alpha \cdot \lambda(\log (X^{1/2} Z X^{1/2} )) \\
&\leq \inf_{Z \succeq I_n} \left[\tr (Z^{-1}  X^{-1/2} C X^{-1/2}) + \alpha \cdot \lambda( \log Z)\right] + \alpha \cdot \lambda ( \log X)\\
&= \tr (X^{-1/2} C X^{-1/2}) - D(P_\alpha || \mu) + \alpha \cdot \lambda (\log X).
\end{align*}
The second equality used the change of variables $\tilde{X} = X^{1/2} Z X^{1/2}$, the second inequality used Lemma \ref{lem:subadd}, and the final equality used Lemma \ref{lem:dual} again. Rearranging the terms yields the inequality.

For the progress claim, observe that in the $X$ update the terms in the function depending on $X$ change from $\tr X^{-1} C +  \alpha \cdot \log(X)$ to $\min_X (\tr (X^{-1} C) + \alpha \cdot \log(X))$, which is equal to $\tr C - D(P_\alpha|| \lambda(C))$ by Lemma \ref{lem:dual}. \end{proof}

We can now prove our running time bounds.
\begin{theorem}\label{thm:maj-sinkhorn}
Suppose $\Phi$ has size at most $1$. If $\capa^{\alpha, \beta} \Phi$ is finite, \textsc{MajSinkhorn} terminates in
$$2 \cdot \frac{1 - \capa^{\alpha, \beta} \Phi }{\min \{\alpha_1, \beta_1\} \eps^2 }$$ steps, and the output is an $(\alpha, \beta)$ scaling of $\Phi$ of size at least $(1 - \eps) k$. \end{theorem}

\begin{proof} First we verify that, assuming termination, the output is an $(\alpha, \beta)$ scaling of $\Phi$ of size $k - \eps$. Set $\tilde \Phi = e^{z} \Phi_{Y^{-1/2}, X^{-1/2}}$. By the assumption that we have terminated, for $\mu = \lambda(\tilde \Phi^* (I_m))$ and $\nu = \lambda(\tilde \Phi (I_n))$, we have $D(P_\alpha || \mu),D( P_\beta || \nu) \leq \min\{\alpha_1, \beta_1\} \frac{\eps^2}{2}$. By Lemma \ref{lem:kl-scale}, we have $ e^{- \eps} \mu \preceq \alpha$ and $e^{- \eps}\nu  \preceq \beta$, so the output $e^{z-\eps} \Phi_{Y^{-1/2}, X^{-1/2}}$ is $(\alpha,\beta)$-majorized. Because $\tilde{\Phi}$ has size at least $k$ by the $z$-update, $e^{z-\eps} \Phi_{Y^{-1/2}, X^{-1/2}}$ has size $e^{- \eps} k \geq (1 -\eps) k$.

Now let's verify the running time. First note that the algorithm is alternating minimization, so the objective function decreases throughout the algorithm. The only step for which this is non-trivial is the $z$ update step. This follows from the fact that the function $g_\lambda(z) := \lambda e^z - 1$ is minimized on the domain $z > 0$ at $z = \max\{0, -\log(\lambda)\}$, which can be shown by straightforward calculus. We now claim that each time the algorithm does not terminate, it decreases the objective function $f^{\alpha, \beta}_k$ by $\min\{\alpha_1, \beta_1\} \frac{\eps^2}{2}$. Initially, the function value is at most $1$ by our assumption on the size. By definition, $f^{\alpha,\beta}_k (X,Y,z) \geq \capa^{\alpha,\beta}_k \Phi.$ This implies the desired running time bound. For the progress bound, by Lemma \ref{lem:prog-bound} the change in the objective function by an $X$ (resp. $Y$) update is at least $D(P_\alpha || \mu)$ (resp. $D( P_\beta || \nu)$), whichever is larger.
Whenever the algorithm doesn't terminate, the larger of the two of these is at least $\min\{\alpha_1, \beta_1\} \frac{\eps^2}{2}$. \end{proof}

The next lemma shows that if a vector is closed to being majorized, one can scale it down slightly and it will become majorized.

\begin{lemma}[Rescaling to majorize]\label{lem:kl-scale} Let $\eps > 0$. If $D(P_\alpha || q)\leq \alpha_1 \frac{\eps^2}{2}$, then $e^{- \eps} q \preceq \alpha$.
\end{lemma}
\begin{proof}
Assume $q$ is in decreasing order. If $q \preceq \alpha$, the claim holds vacuously, so assume that this is not the case. Let $L = \log \max_k (\sum_{i = 1}^k q_i)/(\sum_{i = 1}^k \alpha_i) > 0$, so that $L$ is the smallest number such that $e^{- L} q \preceq \alpha$. Let $i^*$ be the constraint where $e^{- L} q \preceq \alpha$ is tight, i.e., $\sum_{i=1}^{i^*} e^{-L}q_i = \sum_{i=1}^{i^*} \alpha_i$.
It suffices to show that $L \leq \eps$. We know from Lemma \ref{lem:dual} that
$$D(P_\alpha || q) = \sum_{i = 1}^n  q_i - \min_{x_1 \geq x_2 \geq \dots \geq x_n} \sum_{i=1}^n q_i e^{-x_i} + \alpha_ix_i.$$
If we set $x_i = L \cdot \mathbf 1_{[i^*]}$, then
\begin{align*}D(P_\alpha || q) &\geq \sum_{i = 1}^{i^*} q_i - e^{- L} \sum_{i = 1}^{i^*} q_i  - L (\sum_{i = 1}^{i^*} \alpha_i)\\
&= (e^L - 1 - L)\sum_{i = 1}^{i^*} \alpha_i \geq \frac{L^2}{2} \sum_{i = 1}^{i^*} \alpha_i \geq \frac{L^2}{2} \alpha_1.\qedhere
\end{align*}
\end{proof}
\FloatBarrier

\subsection{Triangular scaling and the proof of Theorem~\ref{thm:maj-scal}}\label{sec:tri}

The proof of Theorem~\ref{thm:maj-scal} will go through triangular scaling, a notion studied in \cite{Franks2018}. Our capacity lower bounds in the next section will also rely on triangular scaling.

\begin{definition}[Scalability and triangular scalability]
We say $\Phi_{B,C}$ is a \emph{triangular scaling} of $\Phi$ if $B \in \C^{m \times m},C \in \C^{n\times n}$ are upper triangular. For non-increasing, nonnegative vectors $\alpha \in \R^n, \beta \in \R^m$ we say $\Phi$ is $(\alpha, \beta)$-\emph{scalable} (resp. \emph{triangularly scalable}) if for every $\eps > 0$ there is a scaling (resp. triangular scaling) $\tilde{\Phi}$ satisfying
$$\|\tilde \Phi^*(I_m) - \diag(\alpha)\|_F \leq \eps \text{ and } \|\tilde{\Phi}(I_n) - \diag(\beta)\|_F \leq \eps.$$
\end{definition}
We next introduce the corresponding notion of capacity.
\begin{definition}[Triangular scaling capacity]\label{def:other-cap}
First, for a positive-semidefinite matrix $Y$ we define
$$
\log \det(\nu, Y) := \sum_i (\nu_i - \nu_{i + 1}) \log \prm_i Y,
$$
where $\prm_i$ denotes the $i^{th}$ leading principal minor of $Y$. Let
$$
g_{\alpha,\beta}(X,Y) = \tr \Phi(X^{-1}) Y^{-1}  + \log \det(\alpha , X) + \log \det(\beta , Y),
$$
and let $\capa_{\alpha,\beta} \Phi$ be defined by $\capa_{\alpha, \beta} = \inf_{X\succ 0, Y \succ 0} g_{\alpha,\beta}(X,Y)$.
\end{definition}
It is known that $\capa_{\alpha, \beta} \Phi > -\infty$ characterizes $(\alpha, \beta)$-triangular scalability, and $\capa_{\alpha,\beta} \tilde \Phi$ for \emph{random} scalings $\tilde \Phi$ of $\Phi$ characterizes $(\alpha, \beta)$-scalability. The next theorem collects this and several other known properties of the capacity.
\begin{theorem}[\cite{Franks2018}]\label{thm:tri} Let $\Phi:\C^{n\times n} \to \C^{m\times m}$ be a CP map.
\begin{enumerate}
\item\label{it:tri-cap} $\Phi$ is $(\alpha, \beta)$-triangularly scalable if and only if $\capa_{\alpha,\beta} \Phi > -\infty.$
\item\label{it:scal} $\Phi$ is $(\alpha,\beta)$-scalable if and only if there are unitary matrices
$U \in \C^{m\times m},V \in \C^{n\times n}$ such that
$$\capa_{\alpha,\beta} \Phi_{U,V} > -\infty.$$
\item\label{it:degree} If $\alpha, \beta$ are integral, the set of matrices $A \in \C^{m\times m},B \in \C^{n\times n}$ such that $\capa_{\alpha,\beta} \Phi_{A,B} =-\infty$ is an affine variety generated by polynomials of degree at most $2 (\sum_i \alpha_i)^2$.
\item \label{it:polytope} The set of nonnegative, decreasing vectors $\alpha,\beta$ such that $\capa_{\alpha,\beta} \Phi > -\infty$ (i.e. such that $\Phi$ is $(\alpha, \beta)$-triangularly scalable) is a polyhedral cone cut out by inequalities with Boolean coefficients. The same is true of the set of $\alpha, \beta$ such that $\Phi$ is $(\alpha,\beta)$-scalable.
\item\label{it:concave} The function $(\alpha, \beta) \mapsto \capa_{\alpha,\beta,k} \Phi$ is concave over the set of nonnegative, decreasing vectors $(\alpha, \beta)$ when we take $0 \log 0:=0$.
\end{enumerate}
\end{theorem}
\cite{Franks2018} actually used a slightly different expression for the capacity, but we show in Lemma \ref{lem:compare-capacities} that one is (an affine function of) the logarithm of the other. We now collect a few properties of $\log\det(\mu, \cdot)$.
\begin{lemma}\label{lem:logdet} $ $
\begin{enumerate}
\item\label{it:tri} If $b$ is lower triangular, then
$$\log \det(\mu, b X b^\dagger) = \log\det(\mu,  b b^\dagger) + \log\det(\mu, X).$$
\item\label{it:maj} If $\mu \preceq \alpha$ and $X \succeq I_n$, then
$$\log \det(\mu, X) \geq \lambda(\log X) \cdot \alpha.$$
If also $\sum \alpha_i = \sum \mu_i$, then the inequality holds with only the constraint $X \succ 0$.
\item \label{it:diag}If $X$ is diagonal then
$$\log \det(\alpha, X) = \lambda(\log X) \cdot \alpha.$$
\end{enumerate}
\end{lemma}

\begin{proof}
Item \ref{it:tri} is \cite[Lemma~15]{Franks2018}. For Item \ref{it:maj}, we have
\begin{align*}
\log \det(\mu, X) &= \sum_i (\mu_i - \mu_{i + 1}) \log \prm_i X\\
& \leq \sum_i (\mu_i - \mu_{i + 1}) \log \prod_{j = 1}^i \lambda_j(X)\\
&= \sum_i (\mu_i - \mu_{i + 1}) \sum_{j = 1}^i \lambda_j (\log X) = \mu \cdot \lambda(\log X) \leq \alpha \cdot \lambda(\log X).
\end{align*}
In the inequality we used interlacing, in the final equality we used partial summation, and in the final equality we used $\nu \preceq \beta$ and the fact that $\lambda(\log X) \geq 0$. If $\sum_i \alpha_i = \sum_i \mu_i$ we can reduce to the case $X \succeq I_n$ by multiplying by an appropriate scalar. The third item is a straightforward calculation.
\end{proof}

We have the following inequality relating the capacity for majorization scaling with the one for triangular scaling.
\begin{lemma} \label{lem:caps-ineq}
For any $\mu \preceq \alpha, \nu \preceq \beta$ and $\sum_i \mu_i = \sum_i \nu_i = k$, we have
 $$\capa_{\mu,\nu} \Phi  \leq \capa^{\alpha,\beta}_k \Phi.$$
\end{lemma}
\begin{proof}
By the inequality in Item \ref{it:maj} of Lemma~\ref{lem:logdet},
\begin{align*}\capa^{\alpha,\beta}_k \Phi &\geq \inf_{(X,Y,z) \in \caD} \tr \Phi(X^{-1}) Y^{-1} e^{z}  + \log \det(\mu , X) + \log \det(\nu , Y) - kz\\
&= \inf_{X \succeq I_n, Y \succeq I_n, z_1, z_2 \geq 0} \tr \Phi((X e^{-z_1})^{-1}) (Y e^{-z_2})^{-1}   + \log \det(\mu , X e^{-z_1}) + \log \det(\nu , Y e^{-z_2} )\\
&= \inf_{\tilde X, \tilde Y \succeq 0} \tr \Phi(\tilde X^{-1}) \tilde Y^{-1}   + \log \det(\mu , \tilde X) + \log \det(\nu , \tilde Y).\qedhere
\end{align*}

\end{proof}
Finally, we prove Theorem~\ref{thm:maj-scal}.

\begin{proof}[Proof of Theorem \ref{thm:maj-scal}]
We have already seen in Theorem~\ref{thm:maj-sinkhorn} that if $\capa^{\alpha, \beta} \Phi$ is finite, then for any $\eps > 0$, \textsc{MajSinkhorn} finds a scaling of $\Phi$ that is majorized by $(\alpha, \beta)$ and has size $k(1-\eps)$.
Therefore, it suffices to show that $\Phi$ is not $k$-scalable to $(\alpha, \beta)$ if $\capa^{\alpha,\beta}\Phi = -\infty$.

Suppose for the contrary that $\capa^{\alpha,\beta} \Phi = - \infty$ and $\Phi$ is $k$-scalable to $(\alpha, \beta)$.
Then, by compactness, there exists some $\mu \leq \alpha, \nu\leq  \beta$ such that $\sum_i \mu_i= \sum_i \nu_i = k$ such that $\Phi$ is $(\mu, \nu)$-scalable.
By unitary invariance of $\capa^{\alpha,\beta} \Phi$, we have $\capa^{\alpha,\beta} \Phi_{U,V} = - \infty$ for all unitary matrices $U,V$. By Lemma \ref{lem:caps-ineq}, $\capa_{\mu,\nu} \Phi_{U,V} = - \infty$. By Item \ref{it:scal} of Theorem \ref{thm:tri}, this contradicts $\Phi$ being $(\mu, \nu)$-scalable. \end{proof}
As a side benefit, the properties of $\log\det(\mu, \cdot)$ give a quick proof of Lemma \ref{lem:subadd}.
\begin{proof}[Proof of Lemma \ref{lem:subadd}]
 Without loss of generality, $X^{1/2} Z X^{1/2}$ is diagonal. Then by Item \ref{it:diag} of Lemma \ref{lem:logdet},
$$\alpha \cdot \lambda ( \log X^{1/2} Z X^{1/2}) = \log\det(\alpha, X^{1/2} Z X^{1/2})$$ Choose $U$ unitary so that $X^{1/2} U$ is lower triangular. Then
\begin{align*}
\log\det(\alpha, X^{1/2} Z X^{1/2})
&= \log\det(\alpha, X^{1/2} U U^\dagger Z U U^\dagger X^{1/2})\\
&= \log\det (\alpha, X) + \log\det(\alpha, U^\dagger Z U) \\
&\leq \alpha \cdot \lambda ( \log X) + \alpha \cdot \lambda( \log Z).
\end{align*}
The second equality used Item \ref{it:maj} of Lemma \ref{lem:logdet}, and the inequality used Item \ref{it:maj} of Lemma \ref{lem:logdet} and the unitary invariance of $ \lambda( \log Z)$.
\end{proof}

\subsection{Capacity lower bounds}\label{sec:cap-lb}

\begin{theorem}\label{thm:cap-lb} Suppose the Kraus operators of $\Phi$ have Gaussian integer entries. Then $\capa^{\alpha, \beta}_k \Phi$ is either $- \infty$ or
$$-\capa^{\alpha, \beta}_k \Phi =  O(k (n + m) \log (m + n) + k \log p).$$ \end{theorem}
Before proving the theorem, we combine it with Theorem \ref{thm:maj-sinkhorn} and one small lemma about how normalizing affects the capacity to obtain a concrete bound for how many steps \textsc{MajSinkhorn} requires.

\begin{lemma}[Capacity under scalar multiplication]\label{lem:cap-scalar} Let $0 < c < 1$. Then
$$\capa^{\alpha,\beta}_k c \Phi \geq \capa^{\alpha,\beta}_k \Phi + (\sum \alpha_i) \log c.$$
\end{lemma}

\begin{proof} This is a change of variables argument. Let $\tilde{X} = X/c$.
\begin{align*}
\capa_{k,r} c \Phi &= \inf_{X,Y \succeq I, z \geq 0} \tr Y^{-1} \Phi((X/c)^{-1})e^{z} + \alpha \cdot \lambda(\log X) + \beta \cdot \lambda(\log Y) - k r\\
& = \inf_{\tilde{X} \succeq I/c,Y \succeq I, z \geq 0} \tr Y^{-1} \Phi(\tilde X^{-1})e^{z} + \alpha \cdot \lambda(\log \tilde X + (\log c) I_n) +  \beta \cdot \lambda(\log Y) - k r\\
& \geq \inf_{\tilde{X} \succeq I,Y \succeq I, z \geq 0} \tr Y^{-1} \Phi(\tilde X^{-1})e^{z} + \alpha \cdot \lambda(\log \tilde X) + (\sum_i \alpha_i) \log c +  \beta \cdot \lambda(\log Y) - k r\\
&= \capa_{k,r} c \Phi + (\sum_i \alpha_i) \log c. \qedhere
\end{align*}

\end{proof}

\begin{corollary}\label{cor:sink-running}
Suppose the Kraus operators of $\Phi$ have Gaussian integer entries and $\capa^{\alpha, \beta} \Phi$ is finite. Then \textsc{MajSinkhorn} terminates in
$$O \left(\frac{k (n + m) \log (m + n) + (\sum_i \alpha_i)\log pmnM}{\min \{\alpha_1, \beta_1\} \eps^2 }\right)$$ steps, and the output is an $(\alpha, \beta)$ scaling of $\Phi$ of size at least $(1 - \eps) k$.
\end{corollary}
\begin{proof} At the first step of \textsc{MajSinkhorn} we normalize $\Phi$ to obtain $\tilde \Phi$ of size at most $1$. This entails dividing by a constant $1 \leq c \leq pmn M^2$. By Lemma \ref{lem:cap-scalar}, $\capa^{\alpha,\beta}_k \tilde \Phi \geq \capa_{k,r} \Phi + (\sum_i \alpha_i) \log c.$ Theorem \ref{thm:maj-sinkhorn} shows that \textsc{MajSinkhorn} terminates in at most $2 (1 - \capa^{\alpha, \beta} \tilde \Phi )/\min \{\alpha_1, \beta_1\} \eps^2 $ further steps; plugging in our estimate for $\capa^{\alpha,\beta}_k \tilde \Phi$ completes the proof. \end{proof}

We now collect a few lemmas needed for the proof of Theorem \ref{thm:cap-lb}. Our main tool is Theorem \ref{lem:caps-ineq} and known capacity lower bounds from \cite{Franks2018} for the capacity $\capa_{\mu, \nu}$ for triangular scaling. From a compactness argument similar to the proof of Theorem \ref{thm:maj-scal}, we know that there is some $\mu \preceq \alpha, \nu \preceq \beta$ and unitaries $U,V$ such that $\capa_{\mu, \nu} \Phi_{U,V} > 0$. If the bit complexity of $U,V$ is controlled, then we can use known capacity lower bounds. Fortunately, Theorem \ref{thm:tri} tells us that we have $\capa_{\mu, \nu} \Phi_{A,B} > 0$ for \emph{random} matrices. The next lemma combines Theorem \ref{thm:tri} with Schwarz-Zippel to see how large a range we need to take for the random entries of $A,B$. A similar argument appeared in the capacity lower bounds of \cite{burgisser2019towards}. We'll need a quick folklore bound on the bit complexity of vertices of polyhedra generated by Boolean inequalities.

\begin{observation}\label{obs:vertices}
Let $P = \{x: a_i \cdot x \geq b_i, i \in [m]\}\subseteq \R^n$ be a polytope where $a_i \in \{0,1\}^n$ and $b_i \in \{0,1\}$. Then every vertex of $P$ is the ratio of an integral vector with entries at most $e^{O(n \log n)}$ in absolute value with an integer at most $e^{O(n \log n)}$. As a consequence, any polyhedral cone inside the nonnegative orthant of $\R^n$ cut out by inequalities with Boolean coefficients is the conic hull of some number of integer vectors with entries at most $e^{O(n \log n)}$.
\end{observation}
\begin{proof}
Each vertex is a basic feasible solution, i.e. is obtained from setting some entries to zero and the rest according to the equation $A x = b$ where $A$ is some invertible square submatrix of the $m \times n$ matrix with entries $(a_i)_j$ and $b$ is the corresponding subset of the entries of $b_1, \dots, b_m$. By Cramer's rule, the $i^{th}$ entry of $x = A^{-1} b$ is obtained by $\det A(i)/\det A$ where $A(i)$ has the $i^{th}$ column of $A$ replaced by $b$. Thus $\det A$ is a common denominator. By Hadamard's inequality, $\det A, \det A(i) \leq n^{n/2} = e^{O( n \log n)}.$ For the statement about cones, we simply take the intersection of the cone with the simplex and apply the theorem to this polytope. The extreme points of the polytope generate the extremal rays of the cone.
\end{proof}

\begin{lemma}\label{lem:schwarz-zippel}
Suppose there are some invertible matrices $A',B'$ such that $\capa_{\mu, \nu} \Phi_{A', B'} > 0$. Then there is a pair of invertible matrices $A, B$ with integer entries at most $e^{ O((n + m) \log (n + m))}$ such that
$\capa_{\mu, \nu} \Phi_{A,B} > 0$.
\end{lemma}
\begin{proof}

By Item \ref{it:polytope} of Theorem \ref{thm:tri}, the set of $\mu, \nu$ such that there exists invertible $A,B$ such that $\capa_{\mu, \nu} \Phi_{A,B} > -\infty$ (i.e. such that $\Phi_{A,B}$ is $(\mu, \nu)$-scalable) is a polyhedral cone by Item \ref{it:polytope} of Theorem \ref{thm:tri}. Thus it is enough to check that $\capa_{\mu_i, \nu_i} \Phi_{A,B} > - \infty$ for $n + m$ (the ambient dimension) many extremal rays $(\mu_i, \nu_i)$ which have $(\mu,\nu)$ is in their conic hull. Because the cone is cut out by Boolean inequalities, the extremal rays of the cone have integral entries of magnitude $e^{O((n + m) \log (n + m))}$ by Observation \ref{obs:vertices}. By Item \ref{it:degree} of Theorem \ref{thm:tri}, as $\mu_i, \nu_i$ are integral then the set of pairs $A,B$ such that $\capa_{\mu_i, \nu_i} \Phi_{A,B} = -\infty$ is a variety $\mathcal V(\mu_i,\nu_i)$ generated by polynomials of degree $2(\mu_i \cdot \mathbf 1)^2 = e^{O((n + m) \log (n + m))}$. For each $i$ there exists $A,B$ such that $\capa_{\mu_i, \nu_i} \Phi_{A,B} > -\infty$, so there is some nonzero polynomial $p_i$ of degree $e^{O(m + n) \log (n + m)}$ in the ideal of $\mathcal V(\mu_i,\nu_i)$ such that $p_i(A,B) \neq 0$. By the Schwarz-Zippel lemma, for $A,B$ with i.i.d. random entries in $[0, (m + n) e^{O((m + n) \log (m + n))}]$ there is a nonzero probability that $p_i(A,B) \neq 0$ for all $i$ - so that there is some $A,B$ such that $\capa_{\mu_i, \nu_i} \Phi_{A,B} > -\infty$ for all $i$ and hence $\capa_{\mu,\nu} \Phi_{A,B} > -\infty$. \end{proof}

Next we recall the capacity lower bound we need from \cite{Franks2018}, which is for a slightly different capacity.
\begin{theorem}[{\cite[Theorem~50]{Franks2018}}]\label{thm:fra-cap-lb}  Define
$$\capa(\Phi, \mu, \nu) := \inf_{g,h \textrm{ upper triangular, invertible}} e^{-H(\mu)}\frac{\det(\nu, \Phi( hh^\dagger))}{\det(\mu, h^\dagger h)}.$$
If $\mu, \nu$ are rational with common denominator at most $e^b$ and $\sum \mu_i = \sum \nu_i = 1,$
then
$$\textrm{either }\capa(\Phi, \mu, \nu) = 0 \textrm{ or } \capa(\Phi, \mu, \nu) \geq e^{- O(b  + \log p)}.$$
\end{theorem}

The expression for capacity in the previous theorem is related to our expression $\capa_{\mu, \nu} \Phi$ by an affine transformation:

\begin{lemma} \label{lem:compare-capacities}
We have
$$\capa_{\mu, \nu} \Phi = \nu \cdot\mathbf 1 + \log \capa(\Phi, \mu, \nu) + H(\nu) + H(\mu)$$
where $H$ denotes the Shannon entropy. In particular, $\capa_{\mu, \nu} \Phi$ and
$\log \capa(\Phi, \mu, \nu)$ are either both finite or both $-\infty$.
\end{lemma}

\begin{proof} The identity follows by holding one of the variables fixed and optimizing over the other in the objective function for $\capa_{\mu,\nu} \Phi$. Recall that
$$\capa_{\mu, \nu} \Phi = \inf_{X,Y \succ 0} \tr \Phi(X^{-1}) Y^{-1}  + \log \det(\mu , X) + \log \det(\nu , Y).$$
Infimizing over $Y$ with $X$ held fixed leads us to choose $Y = g^\dagger g$ for $g$ upper triangular such that $g^{-\dagger} \Phi(X^{-1}) g^{-1} \to \diag(\nu)$, or $g^\dagger\diag (\nu) g \to \Phi(X).$ If no such $g$ exists, then the infimum over $Y$ alone is infinite. The proof is straightforward calculus, but we omit it and refer to \cite[Theorem~5.7]{franks2020minimal}.

 Plugging this family of $Y$ values and taking a limit yields
 $$\capa_{\mu, \nu} \Phi = \inf_{X \succ 0} \nu \cdot \mathbf 1 +  \log\det(\mu, X) + \log\det(\nu,  g^\dagger g).$$
 One checks using Item \ref{it:tri} of Lemma \ref{lem:logdet} that
 $$\log\det(\nu,  g^\dagger \diag(\nu) g) = \log\det(\nu, g^\dagger g) + \log\det(\nu,\diag(\nu)) = \log\det(\nu, g^\dagger g) - H(\nu).$$
Thus $\log\det(\nu, g^\dagger g) = \log\det(\nu,  g^\dagger \diag(\nu) g) + H(\nu) = \log\det(\nu,  \Phi(X^{-1}))$. We now have
 $$\capa_{\mu, \nu} \Phi =  \nu \cdot\mathbf 1  + \inf_{X \succ 0} \log\det(\mu, X) + \log\det(\nu,  \Phi(X^{-1})) + H(\nu).$$
 Writing $X^{-1} = hh^\dagger$ for $h$ upper triangular using the Cholesky decomposition, we find
\begin{align*}
\capa_{\mu, \nu} \Phi &= \nu \cdot\mathbf 1 + \inf_{h} \log\det(\mu, h^{-\dagger} h^{-1}) + \log\det(\nu,  \Phi(hh^{\dagger})) + H(\nu)\\
 &= \nu \cdot\mathbf 1  + \inf_{h} -  \log\det(\mu, h^\dagger h) + \log\det(\nu,  \Phi(hh^{\dagger})) + H(\nu)\\
 &=\log \capa(\Phi, \mu, \nu) + H(\nu) + H(\mu) + \nu \cdot\mathbf 1.\qedhere
\end{align*}
\end{proof}

We now combine the previous lemma with Theorem \ref{thm:fra-cap-lb} to obtain a lower bound on $\capa_{\mu, \nu} \Phi$.

\begin{theorem}[Capacity lower bound for $\capa_{\mu, \nu} \Phi$]\label{thm:arb-cap-lb}
Let the $p$ Kraus operators of $\Phi$ have Gaussian integer entries, and let $\mu \cdot \mathbf 1= \nu \cdot \mathbf 1 = k$. If $\capa_{\mu, \nu} (\Phi) > - \infty$ we have
$$ -\capa_{\mu, \nu} (\Phi) = O( k (n + m) \log (m + n) + k \log p).$$
\end{theorem}
\begin{proof}
From Theorem \ref{thm:fra-cap-lb}, we know that either $\capa(\Phi, \mu, \nu) = 0$ or $\capa(\Phi, \mu, \nu) \geq e^{- O(b  + \log p)}$ where $e^b$ is a common denominator of $\mu, \nu$ for $\mu, \nu$ summing to $1$. In our setting $\mu$ and $\nu$ sum to $k$, so we observe
\begin{align*}
\capa_{k\mu, k\nu} \Phi
&= k \nu \cdot\mathbf 1 + \log \capa(\Phi, k\mu, k\nu) + H(k\nu) + H(k\mu)\\
&\geq k ( \log \capa(\Phi, \mu, \nu) + H(\mu))\\
&= - O(k(b + \log p)).
\end{align*}
where we used in the first inequality that the expression for $\log \capa(\Phi, \mu, \nu) + H(\mu)$ scales linearly in $k$ and that the other terms are nonnegative. The last line was by Lemma \ref{lem:compare-capacities} and the fact that $H(\mu) \geq 0$. Changing variables, for $\mu, \nu$ summing to $k$, we have $ -\capa_{\mu, \nu}(\Phi) = O( k ( b + \log p ))$, where $e^{b}$ is a common denominator of $\mu/k, \nu/k$.

Using the concavity of $\capa_{\mu, \nu} \Phi$ in $\mu, \nu$ (Item \ref{it:concave} of Theorem \ref{thm:tri}), we know that the capacity is minimized at an extreme point of the polytope of $\mu, \nu$ such that $\capa_{\mu, \nu} \Phi>-\infty$ and $\sum \mu_i = \sum \nu_i = 1$. It remains to bound the common denominators of the extreme points of the polytope, which is cut out by Boolean inequalities by Item \ref{it:polytope} of Theorem \ref{thm:tri}. By Observation \ref{obs:vertices}, the each vertex of such a polytope is a rational vector with common denominator $e^{O((n + m) \log (n + m))}.$
\end{proof}

\begin{lemma}[capacity under scaling]\label{lem:change-vars}
For $h,g$ upper triangular,
$$ \capa_{\mu, \nu} \Phi_{g,h} = \capa_{\mu, \nu} \Phi - \log\det(\mu, g^\dagger g) - \log\det(\nu, h^\dagger h).$$

\end{lemma}
\begin{proof} This is a straightforward change of variables argument.
\begin{align*}
\capa_{\mu, \nu} \Phi_{g,h} &= \inf_{X,Y \succ 0} \tr g^\dagger Y^{-1}g \Phi(h^\dagger X^{-1} h) + \log\det(\mu, X) + \log\det(\nu, Y) \\
& = \inf_{X,Y \succ 0} \tr g^\dagger Y^{-1}g \Phi(h^\dagger X^{-1} h) + \log\det(\mu, h^{-\dagger} X h^{-1}) + \log\det(\nu, g^{-\dagger} Y g^{-1}) \\
&= \capa_{\mu, \nu} \Phi - \log\det(\mu, h^\dagger h) - \log \det (\nu, g^\dagger g).
\end{align*}
The last equality used Item \ref{it:tri} of Lemma \ref{lem:logdet}. \end{proof}

\begin{proof}[Proof of Theorem \ref{thm:cap-lb}]
By the unitary invariance of $\capa^{\alpha, \beta}_k \Phi$ and Lemma \ref{lem:caps-ineq}, $\capa^{\alpha, \beta}_k \Phi$ is bounded below by \emph{every} $\capa_{\mu, \nu} \Phi_{U,V}$ such that $\mu \preceq \alpha, \nu \preceq \beta $ and $U,V$ unitary. By compactness, there exists some $\mu \preceq \alpha, \nu\preceq  \beta$ such that $\sum_i \mu_i= \sum_i \nu_i = k$ such that $\Phi$ is $(\mu, \nu)$-scalable. By Theorem \ref{thm:tri}, there is a pair of unitaries $U,V$ such that $\capa_{\mu, \nu} \Phi_{U,V} > - \infty.$ By Lemma \ref{lem:schwarz-zippel} there is a pair $A,B$ of invertible matrices with entries $e^{O((m + n) \log (m + n))}$ such that $\capa_{\mu, \nu} \Phi_{A,B} > - \infty$. If we write $A = g \tilde U, B = h \tilde V$ such that $\tilde U,\tilde V$ are unitary and $h,g$ upper triangular, by Lemma \ref{lem:change-vars} we have
$$ \capa_{\mu, \nu} \Phi_{A,B} = \capa_{\mu, \nu} \Phi_{g\tilde{U},h \tilde{V}} = \capa_{\mu, \nu} \Phi_{\tilde U,\tilde V} - \log\det(\mu, h^\dagger h) - \log\det(\nu, g^\dagger g).$$ Both $g^\dagger g$ and $h^\dagger h$ have the same spectra as $A, B$, so all their eigenvalues are in the range $e^{O(( m + n) \log (m + n)}$. Hence the last two terms are $O(k (m + n) \log (m + n))$. The operator $\Phi_{A,B}$ has Gaussian integral Kraus operators and hence $- \capa_{\mu, \nu} \Phi_{A,B}  = O( k (n + m) \log (m + n) + k \log p)$ by Theorem \ref{thm:arb-cap-lb}. Finally,
\begin{align*}
\capa^{\alpha, \beta}_k \Phi &\geq  \capa_{\mu, \nu} \Phi_{\tilde U,\tilde V}\\
& =  \capa_{\mu, \nu} \Phi_{A,B} + \log\det(\mu, h^\dagger h) + \log\det(\nu, g^\dagger g) = O( k (n + m) \log (m + n) + k \log p).\qedhere
\end{align*}
\end{proof}

\section{Analysis of \textsc{DecisionSinkhorn} and \textsc{ApproximateIndep}}\label{sec:decsion}
We now specialize $\alpha, \beta$ to our specific values $\alpha_r$ and $\mathbf 1_m$ in order describe and analyze \textsc{DecisionSinkhorn} and \textsc{ApproxIndep} (Algorithms \ref{alg:sinkhorn-dec} and \ref{alg:approx-indep}) and prove Theorem~\ref{thm:indep-scaling}.

\subsection{Deciding finiteness of capacity with \textsc{DecisionSinkhorn}}
In this section we describe and analyze \textsc{DecisionSinkhorn} (Algorithm \ref{alg:sinkhorn-dec}).

\begin{Algorithm}
Algorithm \textsc{DecisionSinkhorn}$(\Phi, k,r)$:
\begin{description}
\item[\hspace{.2cm}\textbf{Input:}] A CP map $\Phi:\C^{n\times n} \to \C^{m\times m}$, integers $k, r \leq n$.

\item[\hspace{.2cm}\textbf{Output:}] Either \textbf{Unbounded} if $\capa_{k,r} \Phi = -\infty$ or \textbf{Bounded} if not.%
\item[\hspace{.2cm}\textbf{Algorithm:}]
\end{description}
\begin{enumerate}
\item Normalize $\Phi$ to have size at most $1$.
\item If $r = 0$, run \textsc{MajSinkhorn}$(\Phi, \mathbf 1_n ,\mathbf 1_m, k, \eps_1)$ for $\eps_1 =  O( (m + n)^{-1/2})$ and $$T = O(k  (n + m)^2 \log (m + n) + n \log Mp)$$
iterations. If \textsc{MajSinkhorn} terminates, output \textbf{Bounded}; otherwise output \textbf{Unbounded}.
\item If $r > 0$, run \textsc{MajSinkhorn}$(\Phi, \alpha_r,\mathbf 1_m, k, \eps_2)$ for $\eps_2 =  O(1/ \sqrt{ (m + n) n^{2}})$ and $$T = O(k n^2 (n + m)^2 \log (m + n) + n \log Mp)$$
iterations. If \textsc{MajSinkhorn} terminates, output \textbf{Bounded}; otherwise output \textbf{Unbounded}.
\end{enumerate}
\caption{Algorithm for deciding finiteness of $\capa_{k,r} \Phi$.}\label{alg:sinkhorn-dec}
\end{Algorithm}

The main work of the section will be proving the correctness of the algorithm (Theorem \ref{thm:decision-correct}). To do this, we first show that for small enough $\eps$, \textsc{MajSinkhorn} will never terminate if $\capa_{k,r} \Phi = -\infty.$

\begin{lemma}\label{lem:term} If $\capa_k \Phi = -\infty$, then any scaling $\tilde{\Phi}$ of $\Phi$ satisfies one of the following:
\begin{enumerate}
\item \label{it:small-size} $\tr \tilde \Phi(I_n) \leq k - 1/3$.
\item $D(P_{\mathbf 1_n} || \mu) = \Omega(1/n)$
\item $D(P_{\mathbf 1_m} || \nu) = \Omega(1/m)$,
\end{enumerate}
where $\mu$ and $\nu$ denote the spectrum of $\tilde\Phi(I_m)$ and $\tilde\Phi^*(I_n)$, respectively.
Similarly, if $\capa_{k,r} \Phi = -\infty$, then any scaling $\tilde{\Phi}$ of $\Phi$ satisfies one of the following:
\begin{enumerate}
\item \label{it:small-size-r} $\tr \tilde \Phi(I_n) \leq k - 1/4n$.
\item\label{it:left} $D(P_{\mathbf \alpha_r} || \mu) = \Omega(1/n^3)$
\item\label{it:right} $D(P_{\mathbf 1_m} || \nu) = \Omega(1/m n^2).$
\end{enumerate}
\end{lemma}

\begin{proof} We begin with $f_k$; the proof for $f_{k,r}$ is similar. If $f_k$ is unbounded, then there is an independent set of $\Phi$ of size bigger than $m + n - k$. This implies that the same is true for the scaling $\tilde \Phi$ at each step; let $(L,R)$ be such an independent set for $\tilde \Phi$. By Equation \ref{eq:cover},
\begin{align}\tr \Pi_{R^\perp} \tilde{\Phi}^*(I_m)  + \tr \Pi_{L^\perp} \tilde{\Phi}(I_n) \geq \tr \tilde{\Phi}(I_n).\label{eq:cover-tilde}
\end{align}
Let $\ell' = \dim L^\perp, r' = \dim R^\perp$ and note that $\ell' + r' \leq k - 1$.
By the Courant-Fisher min-max theorem,
$\tr \Pi_{R^\perp} \tilde{\Phi}^*(I_m)  + \tr \Pi_{L^\perp} \tilde{\Phi}(I_n) \leq \sum_{i = 1}^{r'} \mu_i + \sum_{i = 1}^{\ell'} \nu_i$. By Equation \ref{eq:cover-tilde}, either $\tilde{\Phi}$ has size at most $k - 1/3$ or $\sum_{i = 1}^{r'} \mu_i + \sum_{i = 1}^{\ell'} \nu_i \geq k - 1/3$.
If the latter holds, then either $\sum_{i = 1}^{r'} \mu_i \geq r' + 1/3$ or $\sum_{i = 1}^{\ell'} \nu_i \geq \ell' + 1/3$ holds since $r' + \ell' \leq k - 1$.
If the former of these two inequalities holds, then $\|p - \mu\|_1 \geq 1/3$ for any $p \in P_{\mathbf 1_n}$.
Now either $\sum_{i = 1}^n \mu_i > 2n$ in which case Lemma~\ref{lem:kl-scale} implies $D(P_{\mathbf 1_n} || \mu)$ is at least a constant, or $\sum_{i = 1}^n \mu_i \leq 2n$ in which case Lemma~\ref{lem:pinsk} implies $D(P_{\mathbf 1_n} || \mu) = \Omega(1/n)$. Similarly, if $\sum_{i = 1}^{\ell'} \nu_i \geq \ell' + 1/3$, then $D(P_{\mathbf 1_m} || \nu) = \Omega(1/m)$.

For $f_{k,r}$, the proof of the first item is the same. It remains only to prove the second item. The proof proceeds the same, except if $f_{k,r}$ is unbounded we have that $(L,R)$ violates $(k,r)$. So either $\ell' + r' < k$, in which case the proof proceeds as for $f_k$, or $\ell' + r' = k$ and $r' > n - r$. In that case, we obtain that either the size of $\tilde{\Phi}$ is at most $k - 1/4n$ or $\sum_{i = 1}^{r'} \mu_i + \sum_{i = 1}^{\ell'} \nu_i \geq k - 1/4n.$ If the latter holds, then either either $\sum_{i = 1}^{r'} \mu_i \geq r' - 1/2n$ or $\sum_{i = 1}^{\ell'} \nu_i \geq 1 + 1/4n$. If the former holds, note that $\sum_{i = 1}^{r'} p_i \leq r' - 1/n$ for $p \preceq \alpha_r$ because $r' > n - r$. Thus $\|p - \mu\|_1$ is $\Omega(1/n)$. Similar reasoning shows $\|p - \nu\|_1$ is $\Omega(1/n)$ if $\sum_{i = 1}^{\ell'} \nu_i \geq 1 + 1/4n$. We can now proceed using Pinsker's inequality as we did for $f_k$ to get the bound $\Omega(1/ n^2 \max\{n,m\})$ for either divergence.
\end{proof}

\begin{lemma}[Unnormalized Pinsker (cf.~\cite{Van2020})] \label{lem:pinsk}
Let $p$ and $q$ be vectors in $\R^n_{\geq 0}$. Suppose $\supp p \subseteq \supp q$. Then
$$D(p||q) \geq \min\left\{\frac{1}{4 \|q\|_1} \|p - q\|_1^2, (1 - \ln 2)\| p - q\|_1.\right\}.$$

\end{lemma}
\begin{proof} Let $s = 1/\sum q_i$. Then $D(s p || sq) = s D(p || q)$. By \cite[Lemma~2.1]{Van2020}, $D(s p || sq) \geq \frac{1}{4} \|sp - sq\|_1^2$ if $\|sp - sq\|_1 \leq 1$ and $(1 - \ln 2) \|sp - sq\|_1$ otherwise.
In the first case $D(p || q) \geq s \| p - q\|^2$, and in the second case $D(p || q) \geq (1 - \ln 2) \| p - q\|.$ \end{proof}
We are now ready to prove the correctness of the algorithm.
\begin{proof}[Proof of Theorem \ref{thm:decision-correct}]
If $\capa_{k,r} \Phi > -\infty$, Corollary \ref{cor:sink-running} implies that \textsc{Majsinkhorn} will terminate in the number of iterations $T$ described in both the $r = 0$ and $r > 0$ cases of \textsc{DecisionSinkhorn}. Now suppose $\capa_{k,r} \Phi = -\infty$. First consider the case $r = 0$, i.e. $\capa_k \Phi = -\infty$. We claim \textsc{MajSinkhorn}$(\Phi, \mathbf 1_n ,\mathbf 1_m, k, \eps_1)$ never terminates for $\eps_1 = O((m + n)^{-1/2})$. Because $\capa_k \Phi = -\infty$, one of the three conditions \ref{it:small-size}, \ref{it:left}, \ref{it:right} in Lemma \ref{lem:term} holds for the scaling $\tilde \Phi = e^{z} \Phi_{X^{-1/2}, Y^{-1/2}}$ after the $z$ update step. After the $z$ update, Item \ref{it:small-size} cannot hold. Thus one of Items \ref{it:left} or \ref{it:right} holds, so the algorithm does not terminate. The proof for the case $r > 0$ is identical but we use the part of Lemma \ref{lem:term} that assumes $\capa_{k,r} \Phi = -\infty$.

For the running times, we have seen in Section \ref{sec:maj-sinkhorn} that both the $X$ and $Y$ update take $O((m+n)^2(m+n+p))$ arithmetic operations. Combining with the bound on the number of iterations yields the stated bound on the number of arithmetic operations.
\end{proof}

\subsection{Finding approximate independent sets with \textsc{ApproximateIndep}}

We now describe and analyze the algorithm \textsc{ApproximateIndep} for finding an approximate independent set.

\begin{Algorithm}
Algorithm \textsc{ApproxIndep}$(\Phi, k,r,\eps)$:
\begin{description}
\item[\hspace{.2cm}\textbf{Input:}] A CP map $\Phi:\C^{n\times n} \to \C^{m\times m}$, integers $k, r \leq n$ such that $\capa_{k,r} \Phi = -\infty$, and a parameter $\eps > 0$.

\item[\hspace{.2cm}\textbf{Output:}] An $\eps$-independent set $(L,R)$ violating $(k,r)$.
\item[\hspace{.2cm}\textbf{Algorithm:}]
\end{description}
\begin{enumerate}
\item Normalize $\Phi$ to have size at most $1$ and set $k' = k - 1/2n$.
\item Run \textsc{MajSinkhorn}$(\Phi, \alpha_r,\mathbf 1_m, k', 0)$ for $T = 4 k ^2 (m + n) n^3 \log (2 e^2 n k^2/\eps)$ iterations. Let $X,Y,z$ be the variables in the final step of the algorithm.
\item\label{it:get_basis} Let $u_1, \dots, u_n$ and $v_1, \dots, v_m$ be the eigenvectors of $X,Y$ in order of decreasing eigenvalue. Let $S$ be the set of integral pairs $i \in 0 \leq i \leq m - 1, 0 \leq j \leq n - 1$ such that $i + j = k - 1$ or $i + j \in \{k-1,k\}$ and $j > n - r.$
\item\label{it:check_basis} For each $i,j \in S$, check if $L = \langle u_{i + 1}, \dots, u_m \rangle, R = \langle v_{j + 1}, \dots v_{n}\rangle$, is $\eps$-independent; it suffices to check if
$$\tr \pi_L \Phi (\pi_R) \leq \eps.$$
If so, \textbf{output} $(L,R)$.
 \end{enumerate}
\caption{Algorithm for finding approximate independent sets.}\label{alg:approx-indep}
\end{Algorithm}

\begin{theorem}\label{thm:find-indep}
Let $k' = k - 1/2n$. If $(X,Y,z) \in \caD$ are such that $f_{k',r}(X,Y,z) \leq - C$, then the procedure in steps \ref{it:get_basis} and \ref{it:check_basis} of Algorithm \ref{alg:approx-indep} for
$$\eps:= 2 e n k^2 e^{- C/nk^2}$$ produces an
 $\eps$-independent set $(L,R)$ that violates $(k,r)$.
\end{theorem}

\begin{proof}[Proof of Theorem \ref{thm:find-indep}]

Suppose $(X,Y,z) \in \caD$ are such that $f_{k',r}(X,Y,z) \leq - C$. As in Algorithm \ref{alg:approx-indep}, let $u_1, \dots, u_n$ and $v_1, \dots, v_m$ be the eigenvectors of $X,Y$ in order of decreasing eigenvalue. In this basis we can write $X,Y = e^{\diag(x)}, e^{\diag(y)}$. Consider the matrix $M_{ij} = \tr v_i v_i^\dagger \Phi(u_j u_j^\dagger)$ so that
$$ f_{k',r}(X,Y,z) = \sum_{ij} M_{ij} e^{- y_i - x_j + z} + x \cdot \alpha_r + y \cdot \mathbf 1 - k' z \leq - C.$$
Note that $z \geq C/k$ by nonnegativity of $x,y$.

Let $(\tilde{x},\tilde{y}) = (x,y)/z$ so that $\tilde x \cdot \alpha_\ell + \tilde y \cdot \mathbf 1 \leq k' - C/z\leq k'$. Our goal is to show that $\tilde{x}_{t + 1} + \tilde{y}_{s + 1} \leq 1 - 1/2nk$ for some $0 \leq s \leq m-1,0 \leq t \leq n-1$ such that $s + t = k - 1$ or $s + t \in \{k-1,k\}$ and $t > n-r.$ By the ordering on $x$ and $y$, for $i > t, j > s$ we have $M_{ij} e^{z - x_i - y_j} \geq M_{ij} e^{z/2nk}$, so by the assumption that $f_{k',r}(X,Y,z) \leq - C$ we have $S e^{z/2nk} - k z \leq -C$ where $S:=\sum_{i > s, j > t} M_{ij}$.
We can write $S e^{z/2nk} - k z = 2nk^2 g_\lambda (z')$ for $g_\lambda (z) := \lambda e^{z} - z$ with $\lambda:= S/2nk^2$ and $z' = z/2nk$. As the function $g_\lambda (z)$ takes its minimum at $z = \max\{0, - \log \lambda\},$ which attains value $\lambda$ if $\lambda > 1$ and $1 + \log \lambda$ otherwise, we must have $\lambda \leq 1$ and $1 + \log \lambda \leq - C/2nk^2$. That is, $\log S/2nk^2 \leq 1- C/2nk^2 $, or $S \leq 2nk^2e^{1 - C/2nk^2}=:\eps$. Hence the pair $L = \langle u_{s + 1}, \dots, u_m \rangle, R = \langle v_{t + 1}, \dots v_{n}\rangle$ satisfies $\tr \pi_L \Phi (\pi_R) \leq \eps,$ which implies $\eps$-independence. The assumption that $s + t = k - 1$ or $s + t \in (k,k-1)$ and $t > n-r$ implies $(L,R)$, which has dimensions $m - s$ and $n - t$, violates $(k,r)$.

We show that such a pair $(s,t)$ exists using the probabilistic method, assigning certain probabilities to each pair $0 \leq i\leq m - 1, 0 \leq j \leq n-1$ of integers such that $i + j = k - 1$ or $i + j \in \{k-1,k\}$ and $j \geq n - r.$ We will assign them such that the joint random variable $(i,j)$ drawn from this probability distribution has $\E[\tilde x_{j + 1} + \tilde y_{i+1}] \leq 1 - 1/2nk,$ and hence some particular $(i,j)$ must satisfy $\tilde x_{j + 1} + \tilde y_{i+1} \leq 1 - 1/2nk$.

We now show how to sample $(i,j)$. We will construct a function $h:[0,k] \to [0,m-1]$, choose $\gamma$ uniformly at random on $[0,k]$, and let $(i,j) = \lfloor \gamma \rfloor, \lfloor h(\gamma) \rfloor$. That is, $(i,j)$ is chosen proportional to the time the curve $(\gamma, h(\gamma))$ spends in the cell $[i,i + 1] \times [j, j + 1]$ for time parameter $\gamma \in [0,k]$. The function $g$ is defined as follows:
$$
h(\gamma) = \left\{ \begin{array}{cc}
m - \ell + \frac{1}{\mu} (m - \ell - \gamma) & 0 < \gamma < m - \ell\\
k - \gamma & m - \ell \leq \gamma \leq k
\end{array} \right.
$$

\begin{figure}
\centering
\includegraphics[width=.3\textwidth]{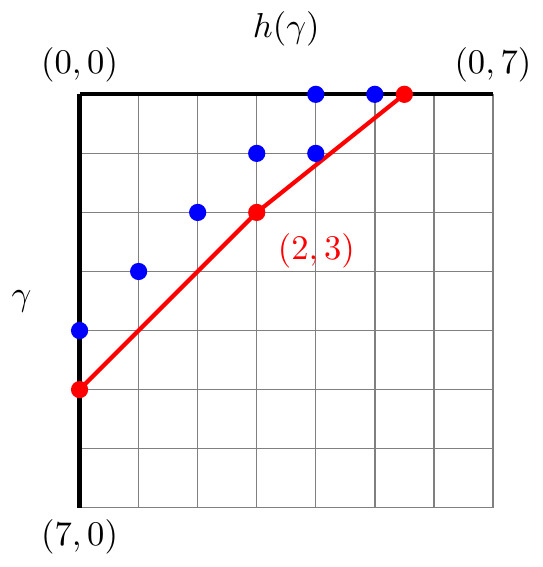}
\caption{An example of the function $h$ for $n = m = 7$, $\ell = 5, r = 4$ and hence $k = 14 - 9 = 5.$ We choose this coordinate system for consistency with matrices. The red line is the curve $(\gamma, h(\gamma))$ for $\gamma \in [0,k]$. The blue points are the support of the probability distribution on $(i,j)$, for these points $(i,j)$ one can see that $i + j = k - 1$ or $i + j \in \{k-1,k\}$ and $j > n - r.$
}\label{fig:example}
\end{figure}

See Figure \ref{fig:example} for an example. We first compute $\E[\tilde x_{j + 1} +  \tilde y_{i+1}] = \E[\tilde x_{j + 1}] + \E[\tilde y_{i + 1}]$. First note that $\gamma, g(\gamma)$ are both non-integral with probability $1$, so $\E[\tilde x_{j + 1}] = \E[ \tilde x_{\lceil g(\gamma)\rceil}]$ and $\E[\tilde y_{i + 1}] = \E[ \tilde y_{\lceil \gamma\rceil}].$ As $\gamma$ is chosen uniformly in $[0,k]$,
$$\E[ \tilde y_{\lceil \gamma\rceil}] = \frac{1}{k} \left(\sum_{i = 1}^k \tilde y_i\right) \leq \frac{1}{k} (\tilde y \cdot \mathbf 1). $$

By the same reasoning, the chance $\lceil h(\gamma)\rceil$ takes a particular value $t$ is $\frac{1}{k}$ if $1\leq t \leq h(m - \ell) = n - r$. If, however, $t > n - r$, then the chance $\lceil h(\gamma)\rceil = t$ is at most $\mu/k$ because this corresponds to the region $\gamma \leq m - r$ where $h$ has slope $1/\mu$. Hence $\E[ \tilde x_{\lceil \gamma\rceil}] \leq (\alpha_r \cdot \tilde x)/k,$
and so
$$\E[\tilde x_{j + 1} +  \tilde y_{i+1}] \leq \frac{1}{k}(y \cdot \mathbf 1 + \alpha_r \cdot \tilde x) \leq k'/k = 1 - \frac{1}{2nk}.$$
It remains to show that all $(i,j)$ in the support of the distribution satisfy $i + j = k - 1$ or $i + j \in \{k-1,k\}$ and $j > n - r.$ Conditioned on $\gamma$ in $(m - r, k)$, i.e. the unit slope part of $h$, with probability $1$ we have $\lfloor \gamma \rfloor + \lfloor h(\gamma) \rfloor = k - 1$ because the line passes along the diagonal of the cells $[i,i + 1] \times [j, j + 1]$ where $i + j = k-1$. Because $\mu = 1 - 1/n$, the curve $(\gamma, h(\gamma)$ lies between the curves $(\gamma, k - \gamma)$ and $(\gamma, k - \gamma + 1)$. Hence $\lfloor \gamma \rfloor + \lfloor h(\gamma) \rfloor$ is always $k - 1$ or $k$. If $h(\gamma) \in (n - r, n - r + 1)$ then $\gamma \in (m - \ell - 1, m - \ell)$ and we have $\lfloor \gamma \rfloor + \lfloor h(\gamma) \rfloor = k - 1.$ Thus if we have $\lfloor \gamma \rfloor + \lfloor h(\gamma) \rfloor  = k$ we also have $h( \gamma ) \geq n -r + 1$ and hence $\lfloor h(\gamma) \rfloor > n - r$. \end{proof}

The following is an easy but handy corollary.
\begin{corollary}\label{cor:alphar-scalable}
    If $(k,r)$ is not violated by any independent set $(L,R)$ of $\Phi$, then $\Phi$ is $k$-scalable to $(\alpha_r, \bfone_m)$.
\end{corollary}
\begin{proof}
    We show the contraposition.
    If $\Phi$ is not $k$-scalable to $(\alpha_r, \bfone_m)$, then $\capa_{k,r} \Phi = -\infty$ by Theorem~\ref{thm:maj-scal}.
    Then, by the precious theorem, there exists an $\eps$-independent set violating $(k, r)$ for any $\eps > 0$.
    By the compactness of approximate independent sets, there exists an independent set violating $(k,r)$.
\end{proof}

\subsection{Independent sets and capacity: Proof of Theorem~\ref{thm:indep-scaling}}

First we have a decomposition lemma using the dominant independent set. Recall that $k^* = \ncrank \caA$ and $(\ell^*,r^*) = (\dim L^*, \dim R^*)$ where $(L^*, R^*)$ are the dominant independent set of $\Phi$. Thus we have $k^* = m + n - \ell^* - r^*$.
\begin{lemma}\label{lem:ncrank-decomp}
Let $(L^*,R^*)$ be the dominant independent set of $\Phi$. Then we can choose orthonormal bases for $\C^m, \C^n$ in which
$$ A _i = \begin{bmatrix} B_i & C_i \\ D_i & 0 \end{bmatrix}$$
where the operator $\Gamma$ with $(m - \ell^*) \times r^*$ Kraus operators $C_i$ is $(m - \ell^*)$-scalable to $((1 - 1/n) \bfone_{r^*}, \bfone_{m - \ell^*})$, and the operator $\Psi$ with $\ell^* \times (n - r^*)$ Kraus operators $D_i$ is $(n-r^*)$-scalable. As a consequence, $\Phi$ is $k^*$-scalable to $(\alpha_{r^*}, \mathbf 1_m)$.\end{lemma}

\begin{proof}
We choose the orthonormal bases $\{u_i\},\{v_i\}$ for $\C^m, \C^n$ so that $L^* = \langle u_{m- \ell^* + 1}, \dots, u_m\rangle $ and $R^* = \langle v_{n - r^* + 1}, \dots, v_{n} \rangle $. Define $\Gamma, \Psi$ be defined as in the statement of the lemma. We claim that $\Gamma$ has no non-trivial independent set of size at least $r^*$ other than $(\{0\},R^*)$. For the sake of contradiction, suppose there exists an independent set $(L, R)$ for $\Gamma$ of size $r^*$ with $L \neq \{0\}$. Then $(L^* + L, R)$ is an independent set for $\Phi$ of size $\ell^* + \dim L + \dim R = \ell^* + r^*$ with $\dim R < r^*$, which contradicts the dominance of $(L^*,R^*)$.
Thus, by Corollary~\ref{cor:alphar-scalable}, $\Gamma$ is $(m - \ell^*)$-scalable to $((1-1/n)\bfone_{r^*}, \bfone_m)$.

Now let us see why $\Psi$ is $(n - r^*)$-scalable.
If not, there exists an independent set $(L, R)$ for $\Psi$ of size larger than $\ell^*$ by Theorem~\ref{thm:k-gurv}.
By the definition of $\Psi$, $(L, R + R^*)$ is an independent set of $\Phi$ of size $\dim L + \dim R + r^* > \ell^* + r^*$, contradicting the dominance of $(L^*,R^*)$.

We now prove the claim about the scalability of $\Phi$. Suppose the family of scalings guaranteed by the first part of the lemma for $\Psi,\Gamma$ are $\Psi_{X_1,Y_1}$ and $\Psi_{X_2,Y_2},$ respectively. Let $X(\gamma) := e^{\gamma \diag \mathbf 1_{[n -r^*]}}$, $Y(\gamma) := e^{\gamma \diag \mathbf 1_{[m - \ell^*]}}$.
One verifies that for $(X,Y) := X(\gamma) \cdot (X_1 \oplus X_2), Y(\gamma) \cdot  (Y_1 \oplus Y_2)$ the scaling $e^{\gamma} \Phi_{X^{-1/2},Y^{-1/2}}$ is the desired scaling as $\gamma \to \infty$.
\end{proof}

\begin{proof} [Proof of Theorem \ref{thm:indep-scaling}]
Recall that we take $k' = k - \frac{1}{2n}$. First note that $\capa_{k,r} \Phi < \capa_{k',r} \Phi$ because $f_{k,r}$ is dominated pointwise by $f_{k',r}$ on $\caD$. Hence \ref{it:cap} implies \ref{it:cap'}.

We now show \ref{it:lex} implies \ref{it:cap}. We first show that If $L^*,R^*$ is the dominant independent set, then the scaling $\tilde \Phi$ guaranteed by Lemma \ref{lem:ncrank-decomp} is a $k^*$-scaling to $(\alpha_{r^*}, \mathbf 1_m)$, and hence also a $k^*$-scaling to $(\alpha_r, \mathbf 1_m)$ for any $r \leq r^*$. Furthermore, $(1 - 1/k^*) \tilde \Phi$ is a $k$-scaling to $(\alpha_r, \mathbf 1_m)$ for any $k < k^*$ and any $r$. The pairs $(k,r)$ satisfying these conditions are exactly those dominated by $(k^*,r^*)$ in lexicographic order, and hence by Theorem \ref{thm:maj-scal} for any such $(k,r)$ we have $\capa^{\alpha_{r^*}, \mathbf 1_m}_k \Phi = \capa_{k,r} \Phi > - \infty$.

Finally, we show \ref{it:cap'} implies \ref{it:lex}. Suppose $\capa_{k',r}$ is finite. By Theorem \ref{thm:maj-scal}, for any $\eps > 0$ the operator $\Phi$ has a scaling $\tilde \Phi$ that is $(\alpha_r, \mathbf 1_m)$-majorized and has size at least $k' - \eps$. The dominant independent sets of $\tilde{\Phi},\Phi$ are related by linear transformations, so it is enough to bound $r^* = \dim R \geq r$ where $L,R$ is the dominant independent set for $\tilde{\Phi}.$ First suppose $k = k^*$ and $r > r^*$.  Because $(L,R)$ is independent for $\tilde \Phi$, we have $$ \tr \Pi_{R^\perp} \tilde{\Phi}^*(I_m)  + \tr \Pi_{L^\perp} \tilde{\Phi}(I_n) \geq \tr \tilde{\Phi}(I_n) \geq k' - \eps.$$
Without loss of generality assume $\tilde{\Phi}(I), \tilde{\Phi}(I)^*$ are diagonal with decreasing elements along the diagonal, and the diagonals are majorized by $\mathbf 1, \alpha_r, $ respectively.

By the Courant-Fisher min-max theorem, $\tr \tilde{\Phi}^*(I_m) \Pi_{R^\perp} \leq \sum_{i = 1}^{n - r^*} \mu_i$ where $\mu$ is the spectrum of $\tilde{\Phi}^*(I)$. As $\mu$ is majorized by $\alpha_r$, we have $\sum_{i = 1}^{n - r^*} \mu_i \leq (n - r) + (1 - 1/n) (r-r^*) \leq n -  r^* - 1/n$. As $\tilde{\Phi}(I_n) \preceq I_m$, we have $\tr \Pi_{L^\perp} \tilde{\Phi}(I_n) \leq m - \ell^*$. Thus $\tr \tilde{\Phi}^*(I_m)  \pi_{R^\perp} + \tr \Pi_{L^\perp} \tilde{\Phi}(I_n) \leq k^* - 1/n.$ If $\eps < 1/2n$, this leads to the contradiction $k^* - 1/n < k' - \eps \leq k^* - 1/n$. A similar argument leads to a contradiction in the case $k > k^*$.\end{proof}

Now we prove Theorem~\ref{thm:indep-correct}.

\begin{proof}[Proof of Theorem~\ref{thm:indep-correct}]
If $\capa_{k,r} \Phi= -\infty$, then by Theorem \ref{thm:indep-scaling} we have $\capa_{k',r} \Phi= -\infty$. By (the proof of) Theorem \ref{thm:decision-correct}, \textsc{MajSinkhorn}$(\Phi, \alpha_r, \mathbf 1_m, k', \eps_2)$ does not terminate for $\eps_2 =  O(1/ \sqrt{ (m + n) n^{2}})$. By Lemma \ref{lem:prog-bound}, $f_{k',r}(X,Y,z) \leq 1 - \frac{1}{4} T \eps_2^2$ after $T$ steps.  Thus in $4(C + 1)/\eps_2^2$ steps we have $f_{k',r}(X,Y,z) \leq C$. By Theorem \ref{thm:find-indep}, the procedure in steps \ref{it:get_basis} and \ref{it:check_basis} of \textsc{ApproxIndep} for $\eps= 2 e n k^2 e^{- C/nk^2}$ produces an
$\eps$-independent set $(L,R)$ that violates $(k,r)$. Thus $C = n k^2 \log (2 e n k^2 /\eps)$ and hence we require $T =4  k ^2 (m + n) n^3 \log (2 e^2 n k^2/\eps).$

The number of arithmetic operations follows from the bound $O((m + n)^2(m+n+p))$ from Section \ref{sec:maj-sinkhorn} on the number operations required for each iteration of \textsc{MajSinkhorn}. The number of operations for Step \ref{it:check_basis} is dominated by the number of operations for Step \ref{it:get_basis}.
\end{proof}

\section{Rounding approximate shrunk subspaces}\label{sec:round}
In this section, we prove that we can round an $\eps$-approximate dominant independent set to the exact dominant independent set.
To this end, we need a bit complexity bound of the dominant independent set.
We present a simple \emph{randomized} algorithm for finding the minimum shrunk subspace using the Wong sequence, from which the desired bit complexity bound follows.

\subsection{Wong sequence}
Here we describe how to use Wong sequences to find the minimum maximally shrunk subspace. We refer frequently to \cite{Ivanyos2015}; in that work they refer to a $c$-shrunk subspace as a $c$-singularity witness.

\begin{definition}[Wong sequence]
    For a matrix space $\caA \subseteq \C^{n \times n}$ and $A \in \caA$, the (second) Wong sequence of $(A, \caA)$ is defined as
    \begin{alignat*}{3}
        W_0 &:= \{\bfzero\},  & \quad W_i &:= \caA(A^{-1}(W_{i-1}))  \quad (i=1, 2, \dots),
    \end{alignat*}
    where $A^{-1}(W_{i-1})$ is the preimage of $W_{i-1}$ under $A$.
\end{definition}

Here we summarize the necessary properties of the Wong sequence.
Let $c:= n - \ncrank\caA$.

\begin{lemma}[{\cite[Lemma~9]{Ivanyos2015}}]\label{lem:Wong}
    For the Wong sequence $(W_i)$ of $(A, \caA)$, the following holds.
   \begin{enumerate}
       \item $W_i$ is nonincreasing, i.e., $W_{i-1} \subseteq W_i$ for $i=1, 2, \dots$. (Hence there exists the limit of $W_i$, which we denote by $W_\infty$)
       \item If $A \in \caB$ satisfies $\rk A = \ncrank\caA$, then $U^* := A^{-1}W_\infty$ is a $c$-shrunk subspace.
       \item $A^{-1}W_\infty \geq \caA^{-1} W_\infty$, where $\caA^{-1} W_\infty := \bigcap_{B \in \caA} B^{-1}W_\infty$.
       \item If a subspace $U$ satisfies $A^{-1}U \supseteq \caA^{-1} U$, then $U \supseteq W_\infty$.
   \end{enumerate}
\end{lemma}

\begin{lemma}[implicit in \cite{Ivanyos2015}]\label{lem:Wong-minimality}
If $A\in \caA$ satisfies $\rk A = \ncrank\caA$, then $U^* = A^{-1}W_\infty$ is the minimum $c$-shrunk subspace.
In particular, the limit $W_\infty$ of the Wong sequence is independent from the choice of $A$.
\end{lemma}
\begin{proof}
    Take an arbitrary $c$-shrunk subspace $U$.
    Then,
    \begin{align*}
        \ncrank\caA
           &= \rk A \\
           &\leq \dim A(U) + \dim A(U^\perp)  \\
           &= n - \dim U - \dim U ^\perp + \dim A(U) + \dim A(U^\perp) \\
           &\leq n - \dim U  + \dim \caA(U) + \dim A(U^\perp) - \dim U ^\perp \tag{$A(U) \subseteq \caA(U)$} \\
           &\leq n - \dim U + \dim \caA(U) \tag{$\dim A(U^\perp) \leq \dim U^\perp $}\\
           &\leq n - c \tag{$\dim U - \dim\caA(U) \geq c$} \\
           &=\ncrank\caA.
    \end{align*}
    Therefore, the above inequalities are tight.
    In particular, we have $A(U) = \caA(U)$ and $\dim U - \dim\caA(U) = n - \ncrank\caA = n - \rk A = \dim\ker A$.
    From $\dim U - \dim A(U) = \dim\ker A$, we have $\ker A \leq U$.
    Therefore, $U = A^{-1}(A(U))$.
    Then, $\caA^{-1}(A(U)) = \caA^{-1}(\caA(U)) \supseteq U = A^{-1}(A(U))$.
    By Lemma~\ref{lem:Wong}, $W_\infty \subseteq A U$.
    Hence, $U^* = A^{-1} W_\infty \subseteq A^{-1}(A(U)) = U$.
\end{proof}

Computing a basis of the Wong sequence $(W_i)$ may lead to an explosion of the bit complexity.
Thankfully, one can use a clever trick proposed in \cite{Ivanyos2015} to compute a basis of $(W_i)$ with the bit complexity bounded. Let $M_A$ denote the maximum absolute value of any of the Gaussian integer entries of $A$.
Let $A^{+}$ be a pseudoinverse of $A$ scaled to have Gaussian integer entries.
Then, $W_i = (\caA A^{+})^i \ker(A A^{+})$ for $i \geq 1$~(see \cite[Lemma~10]{Ivanyos2015}).
First, we compute a basis $v_1, \dots, v_{d_0}$ of $\ker(A A^{+})$ by the Gaussian elimination.
A basis of $W_1$ can be found by finding a maximum independent set from $\{A_jA^{+}v_q : j \in [p], q \in [d_0] \}$.
To this end, it suffices to find a maximum independent column subset from a $n \times pd_0$ matrix, which takes $O(n^2pd_0) = O(pn^3)$ time.
Similarly, a basis a basis $w_1, \dots, w_{d'}$ of $W_i$, we can find a basis of $W_{i+1}$ by finding a maximum independent set from $A_jA^+ w_q$ ($j \in [p]$, $q \in [d']$) in $O(mn^3)$ time.
In total, it takes $O(pn^4)$ time to compute a basis of $W_\infty$.
The basis vector of $W_i$ constructed in this way is in the form of
\[
    (A_{j_i} A^{+}) \cdots (A_{j_1}A^+) v_k
\]
for some $j_1, \dots, j_i \in [p]$ and $k \in [d_0]$, so the bit complexity is polynomially bounded.
Specifically, can take $A^+$ and $v_k$ to have integer entries at most $e^{O(n\log(M_An))}$ in absolute value. For $v_k$ this follows because $\ker(AA^+) = \Img(A)^\perp$. Hence we can take any basis vectors of $W_i$ to be integral with maximum entry $e^{O(in\log(M_An)+i\log M)}$. Therefore we have the following.

\begin{lemma}\label{lem:Wong-alg}
A basis of the limit $W_\infty$ of the Wong sequence of $(A, \caA)$ can be computed in $O(pn^4)$ time. If $A$ has Gaussian integer entries with absolute value at most $M_A$ and $\caA$ has a basis with Gaussian integer entries of absolute value at most $M$, then the output basis will have Gaussian integer entries at most $e^{O(n^2\log(M_An)+n\log M)}$ in absolute value, as will all intermediate numbers.
\end{lemma}

\subsection{A randomized algorithm for finding minimum optimal shrunk subspace}
The matrix subspace
\[
    \caA^{\{d\}} = \langle A_i \otimes E_{jk} : i \in [p], j,k \in [d] \rangle
\]
is called the $d$th blow-up space of $\caA$, where $E_{jk}$ is the matrix with $1$ at the $(j,k)$-entry and $0$ elsewhere.
If $d$ is the minimum blow-up size of $\caA$, then there exists a constant matrix $A \in \caA^{\{d\}}$ such that $\rk A = d\ncrank\caA = \ncrank(\caA^{\{d\}})$.
We can find such a matrix $A$ (with high probability) by substituting random integers to the entries of $X_i$.
Then, one can compute the limit $W_\infty$ of the Wong sequence of $(A, \caA^{\{d\}})$ and $U^* = A^{-1}W_\infty$ is the minimum $dc$-shrunk subspace of $\caA^{\{d\}}$ by Lemma~\ref{lem:Wong-minimality}.
One can obtain the minimum $c$-shrunk subspace of $\caA$ from $U^*$ by the following lemmas.
In the following, we denote the set of $d \times d$ complex matrices by $M(d, \C)$.

\begin{lemma}[{cf. \cite[discussion above Proposition~5.2]{Ivanyos2017}}]
    If a subspace $W \subseteq \C^{nd}$ satisfies $(I \otimes M(d, \C)) W = W$, then
    $W = W_0 \otimes \C^d$. Moreover, we can take $W_0 = \{w \in \C^n : u \otimes e_i \in W \text{ for all $i \in [d]$} \}$.
\end{lemma}
\begin{proof}
    To see $W_0 \otimes \C^d \subseteq W$, take $w = w_0 \otimes v$ for $w_0 \in W_0$ and $v \in \C^d$ arbitrarily.
    Then, $w = w_0 \otimes (\sum_{i=1}^d v_i e_i) = \sum_{i=1}^d v_i (w_0 \otimes e_i) \in W$.
    On the other hand, take an arbitrary $w \in W$ and write $w = \sum_{j=1}^{d} w_j \otimes e_j \in W$, where $w_j \in \C^d$ is the $j$th $d$-dimensional subvector of $w$.
    Then, we have $(I \otimes E_{ij}) w = w_j \otimes e_i \in W$ for $i, j \in [d]$ by $(I \otimes M(d, \C)) W = W$.
    Therefore, $w_j \in W_0$ for $j \in [d]$, which means $w \in W_0 \otimes \C^d$.
    Hence $W_0 \otimes \C^d \supseteq W$.
\end{proof}

\begin{lemma}\label{lem:shrunk-blowup}
    Let $U$ be a $dc$-shrunk subspace of $\caA^{\{d\}}$. Then, $U = U_0 \otimes \C^d$ and $U_0$ is a $c$-shrunk subspace for $\caA$.
    If $U$ is minimum, then $U_0$ is also minimum.
\end{lemma}
\begin{proof}
    Below we denote $\caA^{\{d\}}$ by $\caB$ to simplify the notation.
    First note that, by construction, $\caB$ satisfies $(I \otimes M(d, \C)) \caB = \caB$ and $\caB (I \otimes M(d, \C)) = \caB$.
    Let $W = \caB U$. Then,
    \[
        (I \otimes M(d, \C)) W = (I \otimes M(d, \C)) \caB  U = \caB U = W.
    \]
    By the previous lemma, $W = W_0 \otimes \C^d$.
    Let $U' = (I \otimes M(d, \C)) U$.
    Then, $U' \supseteq U$ and $\caB U' = \caB (I \otimes M(d, \C)) U = \caB U = W$, so
    \begin{align*}
       dc \geq \dim U' - \dim \caB U' \geq \dim U - \dim \caB U = dc.
    \end{align*}
    Thus, $\dim U' = \dim U$ and hence $U' = (I \otimes M(d, \C)) U = U$.
    By the previous lemma, $U = U_0 \otimes \C^d$.
    We have $W_0 \subseteq \caA U_0$ since $\caB = \caA \otimes M(d, \C)$.
    Thus,
    \[
     \dim U_0 - \dim \caA U_0 \geq \dim U_0 - \dim W_0 = \frac{\dim U - \dim W}{d} = c
     \]
     and $U_0$ is a $c$-shrunk subspace for $\caA$.

     Now suppose $U$ is the minimum $dc$-shrunk subspace of $\caA^{\{d\}}$. If $U'$ is a $c$-shrunk subspace of $\caA$, then $U' \otimes \C^d$ is a $dc$-shrunk subspace of $\caA^{\{d\}}$.
     Therefore, $U' \otimes \C^d \supseteq U = U_0 \otimes \C^d$ by minimality of $U$.
     Hence, $U' \supseteq U_0$. As this holds for any $c$-shrunk subspace of $\caA$, $U_0$ is minimum.
\end{proof}

We summarize our algorithm in Algorithm~\ref{alg:shrunk}.

\begin{Algorithm}
Algorithm \textsc{RandomizedShrunkSubspace}$(A_1, \dots, A_p, d, k)$:
\begin{description}
\item[\hspace{.2cm}\textbf{Input:}]
Basis $A_i \in \Z[\sqrt{-1}]^{n \times n}$ ($i = 1, \dots, p$) of $\caA$, the minimum blow-up size $d$ of $\caA$, and an upper bound $k$ of the noncommutative rank of $\caA$.
\item[\hspace{.2cm}\textbf{Output:}] The minimum $c$-shrunk subspace of $\caA$, where $c = n-\ncrank\caA$.
\item[\hspace{.2cm}\textbf{Algorithm:}]
\end{description}
\begin{enumerate}
    \item Draw random $d \times d$ matrices $X_i$ ($i = 1, \dots, p$), where each entry of $X_i$ is drawn from $\{0,1,\dots,2dk-1\}$ independently.
    \item Compute $A = \sum_{i=1}^p A_i \otimes X_i$.
    \item Compute the Wong sequence $(W_i)$ of $(A, \caA^{\{d\}})$ and compute a basis of $U^* = A^{-1}W_\infty$.
    \item Compute a basis of $U^*_0 = \{ u \in \C^n : u \otimes e_i \in U^* \text{ for $i \in [d]$} \}$ by Gaussian elimination.
    \item \textbf{output} $U_0^*$.
\end{enumerate}
\caption{Randomized algorithm for finding the minimum shrunk subspace.}\label{alg:shrunk}
\end{Algorithm}

\begin{theorem}\label{thm:shrunk}
    Algorithm~\ref{alg:shrunk} finds the minimum $c$-shrunk subspace with probability at least $1/2$.
    Furthermore, it makes $O(p d^4n^4)$ arithmetic operations over $\Q$ and the output as well as all intermediate vectors have Gaussian integer entries with magnitude at most $e^{O(n^2d^2\log(Mndk)+nd\log M)}$.
\end{theorem}
\begin{proof}
    If $A$ is rank maximum in $\caA^{\{d\}}$, the output $U_0^*$ is the minimum $c$-shrunk subspace of $\caA$ by the above discussion.
    Hence, it suffices to bound the probability that $A$ attains the maximum rank in $\caA^{\{d\}}$.
    Let $k^*$ denote $\ncrank\caA$.
    Since $\caA^{\{d\}}$ contains a constant matrix of rank $dk^*$, there exists a $dk^* \times dk^*$ nonvanishing minor in $\sum_{i=1}^p A_i \otimes X_i$, where we regard each entry of $X_i$ as an indeterminate.
    The degree of the minor (as a polynomial) is $dk^*$.
    By the Schwartz-Zippel lemma, the minor does not vanish even with probability at least $1/2$ after substituting random integers from $\{0,1,\dots, 2dk-1\}$ to the indeterminates.
    Thus, the rank of $A$ is at least $dk$ with probability at least $1/2$, which completes the proof of the validity of the algorithm.

    The running time is dominated by computing the Wong sequence of $(A, \caA^{\{d\}})$, which is $O(p d^4n^4)$ time by Lemma~\ref{lem:Wong-alg}. By Lemma~\ref{lem:Wong-alg}, the bases of the Wong sequence can be written down with Gaussian integer entries of magnitude at most $e^{O(n^2d^2\log(M_And)+nd\log M)}$. By construction we have $M_A \leq (2dk-1)M$, so this is at most $e^{O(n^2d^2\log(Mndk)+nd\log M)}$. \end{proof}

Theorem \ref{thm:round} follows.
\begin{proof}[Proof of Theorem \ref{thm:round}]
By Theorem \ref{thm:shrunk}, the smallest shrunk subspace $R^*$ has a basis with Gaussian entries of magnitude at most $e^{O(n^2d^2\log(M ndk)+nd\log M)} = e^{O(n^4\log(Mn))}$, where we use $d, k \leq n$. The orthogonal projection to $R^*$ is then proportional to a matrix with Gaussian integer entries of magnitude at most $e^{O(n^5\log(Mn))}$; see Theorem~\ref{thm:Gram-Schmidt}.
\end{proof}

\subsection{Stability of approximate shrunk subspaces}

We begin by proving a converse to Theorem \ref{thm:find-indep}, namely that approximate independent sets violating $(k,r)$ of certain dimensions imply low capacity is small. Combined with our capacity lower bounds, this implies a lower bound how close to independent a pair of subspaces violating $(k^*,r^*)$ can be.

\begin{lemma}\label{lem:cap-upper}
 Suppose $\Phi$ has size at most $1$ and $(L,R)$ is an $\eps$-independent set that violates $(k,r)$. Then
 $$\capa_{k,r}\Phi  \leq 1  +  \frac{1}{n}( 1+ \log ( n  p  \eps^2)).$$
\end{lemma}
\begin{proof}
Let $\dim L = \ell'$ and $\dim R = r$. Choose orthonormal bases so that $L = \langle e_{m- \ell' + 1}, \dots, e_m \rangle $ and $R = \langle e_{n - r' + 1}, \dots, e_{n} \rangle $. In this basis the Kraus operators will take the form
$$ A _i = \begin{bmatrix} B_i & C_i \\ D_i & E_i \end{bmatrix}$$
where $E_i$ has largest singular value at most $\eps$ and hence squared Frobenius norm at most $ n \eps^2$. If we choose $X = e^{\diag(x)}, Y = e^{\diag(y)}$, the objective function is of the form
$$ f(X,Y,z) = \sum_{i,j}  e^{ - y_i  - x_j + z} \sum_s |(A_s)_{ij}|^2 + x^\downarrow \cdot \alpha_r + y \cdot \mathbf 1 - kz. $$
Now take $x =  t\mathbf 1_{\overline{R}}$, $y =  t\mathbf 1_{\overline{L}}$, and $z = t$. Observe that
$$ f(X,Y,z) \leq \sum_s \|A_s\|_F^2 + e^{t} \sum_s \|E_s\|_F^2 +  t \mathbf 1_{\overline{R}} \cdot \alpha_r + (m - \ell') t - k t.$$
Note that $\mathbf 1_{\overline{R}} \cdot \alpha_r < (n - r')$, so if $\ell' + r' > \ell + r$ then $t \mathbf 1_{\overline{R}} \cdot \alpha_r + (m - \ell') t - k t \leq -t$. If instead $\ell' + r' = \ell + r$ but $r < \ell$, then $t \mathbf 1_{\overline{R}} \cdot \alpha_r + (m - \ell') t - k t \leq ((n - r') - 1/n)t + (m - \ell')t - k t \leq - t/n$. In either case,
$$ f(X,Y,z) \leq \sum_s \|A_s\|_F^2 + e^{t} \sum_s\|E_s\|_F^2 - t/n \leq 1 + n p  \eps^2 e^{t} - t/n.$$
The last inequality follows from the assumption that $\tr \Phi(I_n) \leq 1$ and that $L,R$ is $\eps$-independent. The result follows from setting $t = \log ( 1/n p  \eps^2)$. \end{proof}
We now combine the above result with our capacity lower bounds.
\begin{theorem}\label{thm:eps0}Let $\eps_0 = \eps_0 (\Phi)$ be the minimum value of $\eps$ such that there is an $\eps$-independent set of $\Phi$ violating $(k^*, r^*).$ If $\Phi$ has Kraus operators with integer entries at most $M$, then
$$\eps_0 = e^{ - O (k^* (n + m) \log (m + n)   + n \log Mp)}$$
\end{theorem}
\begin{proof}
Suppose there is an $\eps$-independent set violating $(k^*,r^*)$. We rescale $\Phi$ so that Lemma \ref{lem:cap-upper} applies. By the Gaussian integer assumption, $c = \tr \Phi(I_n)^{-1} \leq 1.$ Thus $(L,R)$ remains an $ \eps$-independent set for $c \Phi$.

From Lemma \ref{lem:cap-upper} we have $\capa_{k^*,r^*}c \Phi  \leq 1  +  \frac{1}{n}( 1+ \log ( n p \eps^2)).$ Therefore from Lemma \ref{lem:cap-scalar} we have
$$ 1  +  \frac{1}{n}( 1+ \log ( np \eps)) \geq  \capa_{k^*,r^*}\Phi + n \log c.$$
Note that $\log c \geq - \log (M^2nm p)$.
Rearranging the terms and plugging in Theorem \ref{thm:cap-lb} yields
$$ \eps \geq \frac{1}{np} e^{n (\capa_{k^*,r^*}\Phi + n \log c - 1) - 1} = e^{ - O (k^* (n + m) \log (m + n) + k^* \log p + n \log (M^2 p n m))}. $$
Rearranging terms and using $k^* \leq n$ completes the proof.\end{proof}

The bound on $\eps_0$ from Theorem \ref{thm:eps0} together with the next theorem implies Theorem \ref{thm:close}.

\begin{theorem} Assume $\|A_i\|_F \leq 1$ for all $i$, and that $(L,R)$ is an $\eps$-approximate independent set of size $q := \dim L^* + \dim R^*$ such that $\dim R \leq  \dim R^*$ and $\eps \leq \eps_0$. Then
$$\|\pi_{L} - \pi_{L^*}\|_2, \|\pi_{R} - \pi_{R^*}\|_2  = O( \eps_0^{- 2n - 1} \eps).$$\end{theorem}

\begin{proof}
The idea is to find new approximate independent sets as follows: write $L = H \oplus H'$ and $R = K \oplus K'$, and hope that $(L^* + H, K')$ and $(H' , R^* + K)$ are $\eps'$-approximate independent sets for $\eps' \leq \eps_0$. One can check that the independent sets sizes sum to $2q$ provided $H \cap L^*, K \cap R^*$ are empty, so either one is strictly larger than $q$ or both are of size $q$ and $\dim (L^* + H) > \dim(L^*)$ and hence $\dim K' < \dim R^*$. Unfortunately, this only works if every element in $H$ is reasonably far from $L^*$ and every element of $H'$ is very close to $L^*$ and the analogous property holds for $K,R$. There's no reason to expect such a gap.

In order to fix this, we modify $A_i$, $L$ and $R$ so that there exist subspaces $H,K$ with this property. Choose a basis such that the Kraus operators $A_i$ are written
$$ A _i = \begin{bmatrix} B_i & C_i \\ D_i & 0 \end{bmatrix}$$
and the dominant independent set corresponds to the $0$ block. Define $X(t) = \diag(t^{-1} \mathbf 1_{m - \ell}  \oplus t \mathbf 1_{\ell})$ and $Y(t) = \diag(t^{-1} \mathbf 1_{m - r}  \oplus t \mathbf 1_{r})$, and let the operator $\tilde{\Phi}$ have Kraus operators
$$ \tilde A _i := X(t) A_i Y(t) =   \begin{bmatrix}  t^{-2} B_i & C_i \\ D_i & 0\end{bmatrix}.$$
By Lemma \ref{lem:same-cap} below we have $\capa_{k,r} \tilde{\Phi} = \capa_{k,r} \Phi$. Clearly $(L^*,R^*)$ is still an independent set for $\tilde \Phi$ and $\Phi$, and because $\capa_{k,r} \tilde \Phi > -\infty$ it is also dominant.

We take $\tilde L = X(t)^{-1} L$ and $\tilde R = Y(t)^{-1} R$. First note that $(\tilde L, \tilde R)$ is a $t^2 \eps$-independent set of $\tilde{\Phi}$: for any unit vectors $v \in L, w \in R$ we have
\begin{align*}
\frac{(X(t)^{-1} v)^\dagger \tilde A_i (Y(t)^{-1} w)}{\|X(t)^{-1} v\| \|Y(t)^{-1} w \|}
\leq t^{-2} \| v^\dagger A_i w\| \leq t^{2} \eps.
\end{align*}
For $t$ large, $X(t)^{-1}$ has the effect of moving vectors in $L$ that are not in $L^*$ closer to $(L^*)^\perp$.
Setting $t$ large tends to make more orthogonal with $L^*$. If the distance from $v$ to $L^*$ is $\delta$, then the distance from $X(t)^{-1} v/\|X(t)^{-1} v\|^2$ to $L^*$ is
$$
\frac{t\delta}{ \sqrt{t^2 \delta^2 + (1 - \delta^2) t^{-2}}},
$$
as can be seen from writing $v = \pi_{L^*} v + (I - \pi_{L^*}) v$.
We will choose $t$ carefully so that vectors that are very close to $L^*$ stay close, and vectors that are not close get mapped far from $L^*$.
Let $\delta_i$ denote the singular values of the projections $L \to (L^*)^\perp$ and $R \to (R^*)^\perp$ in increasing order (so there are $\ell + r = q$ total of them).
More specifically, $\delta_1 \leq \dots \leq \delta_{q}$ are the singular values of $Q_{(L^*)^\perp}^\dagger Q_L$ and $Q_{(R^*)^\perp}^\dagger Q_R$, where $Q_L$ denotes a matrix with $Q_L^\dagger Q_L = I$ and $\Img Q_L = L$, and the others are defined similarly.
By the $\sin(\Theta)$ theorem, they are the sine of principal angles between $L$ and $L^*$ (resp. $R$ and $R^*$).
Let $\delta := \max \delta_i$.
This is the sine of the maximum principal angle, so we have $\delta = \max\{\norm{\pi_{L^*} - \pi_L}_2, \norm{\pi_{R^*} - \pi_R}_2\}$.

Because $q \leq 2n$, there is some $\delta_p \geq ( 0.01\eps_0)^{2n} \delta $ such that $\delta_{p - 1}/\delta_{p} \leq 0.01\eps_0. $ where we define $\delta_0:= 0$. Let $S$ denote the sets of indices where $\delta_i$ came from $L$. Without loss of generality assume $\delta$ belongs to $L$, i.e. $q\in S$. Set $t = \delta_p^{-1/2}$.

For $i \in S$, write the $i$th singular vector $v_i \in L$ as $v_i = \sqrt{1 - \delta_i^2}  x_i + \delta_i y_i$ for unit vectors $x_i \in L^*, y_i \in (L^*)^\perp$.
Note that $v_i \in L$ is the $i$th principal vector, and $x_i \in L^*$ and $y_i \in (L^*)^\perp$ are corresponding the principal vectors of $L^*$ and $(L^*)^\perp$, respectively.
Moreover, because they are right singular vectors, $y_i$ are orthogonal for $i \in S$ and hence $x_i$ are as well. Then $\tilde{L}$ will be spanned by the vectors
$$w_i = X(t)^{-1} v_i = t^{-1} \sqrt{1 - \delta_i^2} x_i + t \delta_i y_i.$$
The $w_i$ will remain orthogonal for all $t \in (0,\infty)$, and normalizing each to a unit vector yields
$$\hat{w}_i = \frac{t^{-1} \sqrt{1 - \delta_i^2}}{ \sqrt{(1 - \delta_i^2) t^{-2} + \delta_i^2 t^2 }} x_i + \frac{t \delta_i}{{ \sqrt{(1 - \delta_i^2) t^{-2} + \delta_i^2 t^2 }}} y_i.$$
Note that the second coefficient is a monotone increasing function of $t \in (0,\infty)$ and $\delta_i \in (0,1)$ and vice versa for the first.
Take
\begin{alignat*}{2}
    H &= \langle \hat{w}_i: i \geq p, i \in S\rangle, & \quad  H' &= \langle \hat{w}_i: i < p, i \in S\rangle \\
    K &= \langle \hat{w}_i: i \geq p, i \not \in S\rangle, & \quad  K' &=  \langle \hat{w}_i: i < p, i \not \in S\rangle.
\end{alignat*}
Obviously, we have $\tilde L = H \oplus H'$ and $\tilde R = K \oplus K'$.
The inner product between $\hat{w}_i \in H$ and any unit vector in $L^*$ is at most
\begin{align}\label{eq:L^*-H}
\frac{\delta_p^{1/2} \sqrt{1 - \delta_p^2}}{ \sqrt{(1 - \delta_p^2)\delta_p + \delta_p }} \leq \frac{1}{\sqrt{2}}.
\end{align}
by our choice of $t$ and $H$.
In particular, this implies $L^* \cap H = \{\bfzero\}$.
On the other hand, for any unit vector in $H'$, the inner product with any unit vector in $L^*$ is at most
\begin{align}\label{eq:L^*-H'}
\frac{\delta_p^{-1/2} \delta_i}{{ \sqrt{(1 - \delta_i^2) \delta_p + \delta_i^2 \delta_p^{-1} }}} \leq \frac{0.01 \delta_p^{1/2}  \eps_0}{{ \sqrt{.5 \delta_p + 0.01 \eps_0 \delta_p }}} \leq 0.03  \eps_0.
\end{align}

The analogous inequalities hold for $K, K'$.
Now we can check the quality of our independent sets. We prove it for $(L^* + H, K')$; the proof for $(H', R^* + K)$ follows mutatis mutandis. Let $v \in L^*, h \in H$ and let $w + z \in K'$ be a unit vector with $w \in R^*$ and $\|z\| \leq 0.03\eps_0$ as guaranteed by the upper bound~\eqref{eq:L^*-H'} of inner product between $K'$ and $R^*$. Then
\begin{align*}
(v + h)^\dagger \tilde A_i (w + z)
&= v^\dagger \tilde A_i w +  h^\dagger \tilde A_i (w + z) + v^\dagger \tilde A_i z \\
&\leq 0 + \delta_p^{-1} \eps \|h\| + 0.03 \eps_0  \|v\|\\
&\leq (\delta^{-1} \eps_0^{-2n} \eps + 0.03 \eps_0) (\| h\| + \|v\|) \| w + z\|.
\end{align*}
The first inequality used that $(L^*,R^*)$ is independent and that $(\tilde{L}, \tilde{R})$ is $t^2\eps = \delta_p^{-1}\eps $-independent for $\tilde \Phi$.
It remains to bound $\| h\| + \|v\|$ by $\| h + v\|.$
But by our upper bound~\eqref{eq:L^*-H} of inner product,
\begin{align*}
\|h + v\|^2 &\geq \|h\|^2 + \|v\|^2 - 2 |\langle h, v \rangle|  \geq \|h\|^2 + \|v\|^2 - \sqrt{2} \|h\| \|v\| \\
&\geq 2 (\frac{1}{2} (\|h\| + \|v\|))^2 - \sqrt{2}( \frac{1}{2}(\|h \| + \|v\|))^2
\geq \frac{1}{4} (2 - \sqrt{2}) (\|h\| + \|v\|)^2
\end{align*}
by two applications of Jensen's inequality. %

Putting this all together, we find that $(L^* + H, K')$ and $(H', R^* + K)$ are $(10 \delta^{-1} \eps_0^{-2n} \eps  + .5 \eps_0)$-independent sets of size $q$. Because $H$ was nonempty and $L^* \cap H = \{\bfzero\}$, we have $\dim (L^* + H) > \dim L^*$. Thus, by Theorem~\ref{thm:eps0}, we must have
$10 \delta^{-1} \eps_0^{-2n} \eps  \geq .5 \eps_0,$ or $\delta  = O( \eps_0^{ - 2n - 1} \eps)$,
or else there is some $< \eps_0$-independent set $(L',R')$ with size strictly bigger than $q$ or of size $q$ with $\dim R' < \dim R^*$.
\end{proof}

\begin{lemma} \label{lem:same-cap}
If $\Phi$ has Kraus operators with block decomposition
$$ A _i = \begin{bmatrix} B_i & C_i \\ D_i & 0 \end{bmatrix}$$
where $E_i$ is $\ell \times r$ for $\ell + r = m + n - k$, and $\tilde \Phi$ has Kraus operators
$$ \tilde A _i = \begin{bmatrix} t^{-2} B_i & C_i \\ D_i & 0 \end{bmatrix},$$
for $t > 0$ then $\capa_{k,r} \tilde \Phi = \capa_{k,r} \Phi$.
\end{lemma}
\begin{proof}We show that a family of $\eps$-minimizers to both objective functions asymptotically achieves the same function value. Let $X(\gamma) := e^{\gamma \diag \mathbf 1_{n -r}}$, $Y(\gamma) := e^{\gamma \diag \mathbf 1_{m - \ell}}$. We have seen (Lemma \ref{lem:ncrank-decomp}) that there are $X_1, X_2, Y_1, Y_2$ such that $ (X,Y,z) := X(2\gamma) \cdot (X_1 \oplus X_2), Y(2\gamma) \cdot  (Y_1 \oplus Y_2), 2\gamma $
achieves $f_{k,r} (X,Y,z) \leq \capa_{k,r} \Phi + \eps$ as $\gamma \to \infty$. We also have $f_{k,r} \geq \capa_{k,r}(\Phi)$ because $(X,Y,z)$ is feasible. Note that $X,Y$ are Hermitian because the block diagonal matrices commute with $X(\gamma), Y(\gamma)$ resp. One checks that for any $\gamma$,
\begin{align*} f_{k,r} (X,Y,z) &= \tr (Y_1 \oplus Y_2)^{-1} \Phi_{e^{\gamma} X(- \gamma) ,  Y(- \gamma)} ((X_1 \oplus X_2)^{-1}) \\
& + \alpha_r \cdot \lambda \log X(2\gamma) + \alpha_r \cdot \lambda (X_1 \oplus X_2)\\
& + \log\det(Y(2\gamma)) + \log\det (Y_1 \oplus Y_2) - 2 k \gamma.
\end{align*}
Here we have used that $X(\gamma), (X_1 \oplus X_2)$ commute. Some terms cancel:
$$ \alpha_r \cdot \lambda \log X(2\gamma) + \log\det(Y(2\gamma)) - 2k \gamma = 2\gamma ( n -r + m - \ell - k) = 0. $$
Therefore
$$f_{k,r} (X,Y,z) = \tr (Y_1 \oplus Y_2)^{-1} \Phi_{e^{\gamma} X(- \gamma) , Y(- \gamma)} ((X_1 \oplus X_2)^{-1}) +  \alpha_r \cdot \lambda (X_1 \oplus X_2) + \log\det (Y_1 \oplus Y_2).$$
By Lemma \ref{lem:ncrank-decomp}, the same $(X,Y,z)$ are also $\eps$-minimizers for the objective function for $\tilde{\Phi}$, with the expression
$$ \tr (Y_1 \oplus Y_2)^{-1} {\tilde \Phi}_{e^{\gamma} X(- \gamma) , Y(- \gamma)} ((X_1 \oplus X_2)^{-1}) +  \alpha_r \cdot \lambda (X_1 \oplus X_2) + \log\det (Y_1 \oplus Y_2).$$
But $\Phi_{e^{\gamma} X(- \gamma) , Y(- \gamma)}$ and ${\tilde \Phi}_{e^{\gamma} X(- \gamma) , Y(- \gamma)}$ tend to the same operator as $\gamma \to \infty$ - namely the one with Kraus operators
$$ \begin{bmatrix} 0 & C_i \\ D_i & 0 \end{bmatrix}.$$
Thus the function values are asymptotically the same. Letting $\eps \to 0$ completes the proof.
\end{proof}

\subsection{Proof of correctness of \textsc{RoundSubspaces}}

\begin{proof}[Proof of Theorem \ref{thm:roundsub-correct}]
Theorem \ref{thm:decision-correct} implies each call to \textsc{DecisionSinkhorn} produces correct results and takes
$O((m+n)^2(m+n+p)(k(n + m)^2 \log (m + n) + n \log M))$
operations. The number of operations taken for the calls to \textsc{DecisionSinkhorn} is dominated by the calls in the binary search for $r^*$, and there are $O(\log n)$ such calls. Hence the total arithmetic operations here is at most
$O((m+n)^2(m+n+p)(k(n + m)^2 \log^2(m + n) + n \log M \log n))$.

By Theorem~\ref{thm:indep-correct}, \textsc{ApproxIndep} is correct and takes
$O(k^2 n^3 (m+n)^3 (m+n+p) \log (n k/\eps)) = O(n^5 (m+n)^3(m+n+p) \log(n/\eps))$
operations, where
$\eps = 1/2 M_0 M_1^2 = e^{-O((mn^2+n^3)\log(m + n) + n^5 \log (Mnp))}$.
Hence the total number of operations is
$O(n^7 (m+n)^3 (m+n+p)((m+n)\log(m+n) + n^3\log(Mnp)))$.

The application of  \textsc{ApproxIndep} produces an $\eps$-independent pair $(L,R)$ violating $(k^*, r^* + 1)$ for $\eps = 1/2 M_0 M_1^2$. It remains to show that our rounding procedure is correct, i.e. that after rounding $\pi_R$ we obtain $\pi_{R^*}$. By Theorem \ref{thm:round}, there is some number $0 < q \leq M_1 = e^{O( n^5 \log (Mn))}$ such that every entry of $\pi_{R^*}$ takes the form $p/q$ for $p \in \Z$. On the other hand, by Theorem \ref{thm:close} we know that $\|\pi_{R} - \pi_{R^*}\|_2  \leq M_0 \eps,$ so by Lemma \ref{lem:entries} every entry $\alpha$ of $\pi_{R}$ differs from the corresponding entry $p/q$ of $\pi_{R^*}$ by at most $M_0 \eps$. By our choice of $\eps$, $M_0 \eps \leq 1/2M_1^2$. If we round $\alpha$ to the nearest rational number $a/b$ satisfying $b \leq M_1$, we must have $|a/b  - \alpha | \leq M_0\eps$. But by our choice of $\eps$ we have $M_0 \eps < 1/2M_1^2$, and so $|a/b - p/q| < 1/M_1^2$. Hence we must have $a/b = p/q$. The rounding takes time $O(\log M_1) = O( n^5 \log (Mn))$, which is dominated by the other steps.
\end{proof}
\begin{lemma}\label{lem:entries} Let $A \in \C^{n\times n}$ be a Hermitian matrix. Then $\|A\|_2 \geq \sup_{i,j \in [n]} |A_{ij}|.$
\end{lemma}
\begin{proof} By interlacing it is enough to prove this for $2\times 2$ matrices. For the diagonal entries it is true by $\|A\|_2 \geq e_i^\dagger A e_i$. For the off-diagonal entries, note that $\lambda_1(A)^2 + \lambda_2(A)^2 = A_{11}^2 + A_{22}^2 + 2|A_{12}|^2$. The larger of $\lambda_1(A)^2, \lambda_2(A)^2$ is at least half this expression, which is at least $|A_{12}|^2 = |A_{21}|^2$. \end{proof}

\section{Application to fractional linear matroid matching and rank-2 Brascamp-Lieb polytope}\label{sec:bl}
In this section, we describe applications of our algorithms to fractional linear matroid matching and rank-2 BL polytope.

\subsection{Fractional linear matroid matching and rank-2 Brascamp-Lieb polytope}
Let us recall the definition of \emph{fractional matching polytope}~\cite{VandeVate1992}.
Let $M = (E, \caI)$ be a matroid.
Note that the size of the ground set $E$ may be infinite but $M$ must have a finite rank.
A subset of $E$ of rank one or two is called a \emph{line}.
Note that a line need not to be a flat in general.
Given a matroid $M$ and a set of lines $L$ ($p = |L|$), the fractional matroid polytope is given by
\begin{align}\label{eq:frac-matroid-matching}
    \begin{split}
    a(F) \cdot x &\leq r(F), \qquad \text{($F$: flat of $M$)} \\
    x &\geq \bfzero,
    \end{split}
\end{align}
where $a(F) \in \{0,1,2\}^L$ is defined as
\begin{align*}
    a(F)_\ell =
    \begin{cases}
    0 & (F \cap \ell = \emptyset) \\
    2 & (F \supseteq \ell) \\
    1 & \text{(otherwise)}
    \end{cases}
\end{align*}
and $r$ is the rank function of $M$.
A vector in this polytope is called a fractional matroid matching.
A fractional matroid matching $x$ is said to be \emph{perfect} if $2\sum_{i \in L} x_i = n$.

In particular, we are interested in the case that the matroid is defined by the linear independence in $\C^n$ and each line is a two-dimensional subspace in $\C^n$.
We call it the \emph{fractional linear matroid matching polytope}.
Assume that each line $\ell_i$ is given by the row space of $B_i \in \C^{2\times n}$.
Then, the perfect fractional matching polytope is given by
\begin{align}\label{eq:frac-linear-matroid-matching}
    \begin{split}
    \sum_{i=1}^p x_i \dim(\row B_i \cap V) &\leq \dim V, \qquad \text{($V$: subspace of $\C^n$)} \\
    2\sum_{i=1}^p x_i  &= n, \\
    x_i &\geq 0 \qquad (i=1, \dots, p)
    \end{split}
\end{align}

If all $B_i$ are real matrices of rank-2, the perfect fractional linear matroid matching polytope coincides with the BL polytope of a tuple of operators $B = (B_1, \dots, B_p)$.
Recall that the BL polytope for full row-rank operators $B_i \in \R^{n_i \times n}$ is given by \eqref{eq:BL-V-repr}.

\begin{lemma}
     If $n_i = 2$ for all $i$, \eqref{eq:frac-linear-matroid-matching} and \eqref{eq:BL-V-repr} define the same polytope.
\end{lemma}
\begin{proof}
    For $V \leq \C^n$, let $W = V^\perp$.
    Then $\dim(\row B_i \cap V) = n_i - \dim(B_i W)$.
    So \eqref{eq:frac-linear-matroid-matching} implies \eqref{eq:BL-V-repr}.
    The other direction is analogous.
\end{proof}
Thus, linear optimization on rank-2 BL polytopes can be reduced to the algorithms on fractional matroid matching polytopes~\cite{Chang2001b,Gijswijt2013}.
However, the algorithms of \cite{Chang2001b,Gijswijt2013} are polynomial only in the arithmetic model, i.e., the bit complexity of the intermediate numbers may explode.
Furthermore, the algorithm of \cite{Chang2001b} is fairly complicated. We might hope to use the following ``scaling'' characterization of the BL polytope to find a simpler algorithm.

\begin{theorem}[\cite{Bennett2008, Garg2018}] \label{thm:bl}
Suppose each $B_i$ has rank $n_i$ and $x \in \R^p_{\geq 0}$. The following are equivalent:
\begin{enumerate}
\item $x \in BL(B)$.
\item For every $\eps > 0$ there is an invertible matrix $g \in \R^{n \times n}$ and invertible matrices $h_i \in \R^{n_i \times n_i}$ such that $\tilde{B}_i:= h_i B g^T$ satisfies
\begin{align*}
(1 - \eps) I_n \preceq \sum x_i \tilde{B_i}^\dagger \tilde{B_i} &\preceq (1 + \eps) I_n\\
\text{and } \tilde{B_i} \tilde{B_i}^\dagger &= I_{n_i}\;\forall i \in [p].\end{align*}
\item The following quantity is finite:
$$\capa_x B:= \inf_{X\succ 0, Z_i \succ 0} \sum_{i = 1}^n \tr B_i X^{-1} B_i^\dagger Z_i^{-1} + x_i \log\det Z_i + \log\det X.$$
\end{enumerate}
Moreover, suppose $x = c/d$ for $c \in \Z_{\geq 0}^p$, $d \in \Z_{> 0}$, and that the entries of $B$ are Gaussian integers with absolute value at most $\log M$. Then it is possible to decide if $x \in BL(B)$ in time $\poly( k, d, n, \log M)$ where $k = \sum_i c_i n_i$.
\end{theorem}
The third item looks different from the optimization problem considered in \cite{Bennett2008}, but it is easily shown that the conditions are equivalent. In the following, we provide a simple Sinkhorn-style algorithm for finding a maximum fractional linear matroid matching.

\subsection{Algorithm for finding maximum fractional linear matroid matching}
We are to solve the following linear program.
\begin{alignat*}{2}
    &\text{maximize}  \quad && \sum_{i=1}^p x_i  \\
    &\text{subject to}\quad && \sum_{i=1}^p x_i \dim(\row B_i \cap V) \leq \dim V, \qquad \text{($V$: subspace of $\C^n$)} \\
    &\quad && x_i \geq 0 \qquad (i=1, \dots, p)
\end{alignat*}
Denote the feasible region of this program by $P(B)$. The dual problem is the \emph{minimum 2-cover problem}.
A pair $(S, T)$ of flats in $M$ is called a \emph{2-cover} if $a_i(S) + a_i(T) \geq 2$ for each $i$.
A 2-cover $(S, T)$ is said to be \emph{nested} if $S \subseteq T$ and \emph{minimum} if $r(S) + r(T)$ is minimized.
\cite{Chang2001} showed that the size of maximum fractional matroid matching is equal to the size of a minimum 2-cover.

While we could define a Sinkhorn-style algorithm directly using the $B_i$'s, we instead relate 2-covers for fractional linear matroid matching with shrunk subspaces of a matrix space in order to re-use the algorithms we have already defined. %

Let $a_i, b_i$ be arbitrary basis of line $\ell_i = \row B_i$ and let $A_i = a_i \wedge b_i := a_i b_i^T - b_i a_i^T \in \C^{n\times n}$ for each $i \in [p]$.
Let $\caA$ be the matrix space spanned by $A_1, \dots, A_p$.
Then, it is easy to see that $(S, T)$ is a nested minimum 2-cover if and only if $(S^\perp, T^\perp)$ is a maximum independent set; see \cite{Oki2022}.
The lattice structure of maximum independent sets implies the existence of minimum 2-cover $(S^*, T^*)$ such that $S^* \subseteq S$ and $T \subseteq T^*$ for any nested minimum 2-cover $(S,T)$.
This special minimum 2-cover is called the \emph{dominant 2-cover}~\cite{Chang2001}.
Thus, if $(L, R)$ is the dominant independent set, then $(L^\perp, R^\perp)$ is the dominant 2-cover.
So we can find the dominant 2-cover with Algorithm~\ref{alg:approx-indep}.

A maximum fractional matching can be also found by \textsc{MajSinkhorn} for a CP map $\Phi: \R^{n \times n} \to \R^{n \times n}$ with the Kraus operator $A_i$.
Note that $\Phi(\overline{X}) = \overline{\Phi^*(X)}$ because $A_i^T = - A_i$.
This implies that $f_k(X, \overline{Y}, z) = f_k(Y, \overline{X}, z)$ (see \eqref{eq:simple-capacity} for the definition of $f_k$).
Let $\tilde\Phi = \Phi_{g,h}$ be a $k$-scaling of $\Phi$ to $(\bfone_n, \bfone_n)$.
First observe that we can take $h = \overline{g}$.
This follows from the g-convexity
\[
    \frac{f_k(X,\overline{Y},z) + f_k(Y, \overline{X}, z)}{2} \geq f_k(X \# Y, \overline{Y} \# \overline{X}, z)
    = f_k(X \# Y, \overline{X \# Y}, z)
\]
where $X\# Y$ denotes the mid point of the geodesic between $X$ and $Y$, i.e., the matrix geometric mean of $X$ and $Y$.

By the scaling condition of $\Phi$, for any $\eps > 0$,
there exists $g$ such that
\begin{alignat}{2}\label{eq:fracmat-scaling}
    \sum_{i=1}^p \tilde A_i \tilde A_i^\dagger &\preceq (1+\eps)I_n,
    & \quad  \sum_{i=1}^p \tr(\tilde A_i \tilde A_i^\dagger) &\geq k,
\end{alignat}
where $\tilde A_i = g A_i g^T = g (a_i \wedge b_i) g^T$.
Denote $\tilde a_i = g a_i$  and $\tilde b_i = g b_i$ so that $\tilde A_i = \tilde a_i \wedge \tilde b_i$.
Apply the Gram-Schmidt orthogonalization to $\tilde a_i, \tilde b_i$ and obtain an orthonormal basis $u_i, v_i$ of $\Img(\tilde A_i)$.
This is equivalent to the thin QR-decomposition:
\[
    \begin{bmatrix}
        \tilde a_i & \tilde b_i
    \end{bmatrix}
    =
    \begin{bmatrix}
        u_i & v_i
    \end{bmatrix}
    \begin{bmatrix}
        r_{11}^{(i)} & r_{12}^{(i)} \\
                   0 & r_{22}^{(i)}
    \end{bmatrix}
    =: Q_i R_i.
\]
By $\tilde a_i = r^{(i)}_{11} u_i$ and $\tilde b_i = r^{(i)}_{12} u_i + r^{(i)}_{22} v_i$, we have $\tilde a_i \wedge \tilde b_i = (r^{(i)}_{11} u_i) \wedge (r^{(i)}_{12} u_i + r^{(i)}_{22} v_i) = r^{(i)}_{11}r^{(i)}_{22} u_i \wedge v_i = \det R_i u_i \wedge v_i$ since $u_i \wedge u_i = 0$.
Plugging this back, we have
\begin{alignat*}{2}
    \sum_{i=1}^p \abs{\det R_i}^2 (u_i \wedge v_i)(u_i \wedge v_i)^\dagger &\preceq (1+\eps)I_n,
    & \quad
    \sum_{i=1}^p \abs{\det R_i}^2 \tr((u_i \wedge v_i)(u_i \wedge v_i)^\dagger )&\geq k.
\end{alignat*}
Now observe that $(u_i \wedge v_i)(u_i \wedge v_i)^\dagger$ is the orthogonal projection onto $\Img(\tilde A_i)$.
This follows from
$(u_i \wedge v_i)(u_i \wedge v_i)^\dagger = (u_i v_i^T - v_i u_i^T)\overline{(v_i u_i^T - u_i v_i^T)} = u_iu_i^\dagger + v_iv_i^\dagger = Q_i Q_i^\dagger$ by the orthonomality of $u_i, v_i$.
So $\tr((u_i \wedge v_i)(u_i \wedge v_i)^\dagger) = \dim\Img(\tilde A_i) = 2$.
Letting $x_i := \abs{\det\tilde R_i}^2$, we have
\begin{align}\label{eq:fracmat-proj}
    \sum_{i=1}^p x_i \pi_{\Img(\tilde A_i)} \preceq (1+\eps)I_n,
\end{align}
and $\sum_{i=1}^p x_i \geq k/2$. We can show that $x$ satisfies the constraints of fractional linear matroid matching up to $(1+\eps)$ multiplicative error.
Here is a useful lemma for obtaining the linear inequalities for all subspaces from a matrix inequality.

\begin{lemma}\label{lem:proj-ineq}
    Let $\pi_{V_i}$ be the orthogonal projection onto a subspace $V_i$ and $x_i \geq 0$ ($i=1,\dots, m$), and $\alpha \geq 0$.
    If $\sum_{i=1}^m x_i \pi_{V_i} \preceq \alpha I_n$, then $\sum_{i=1}^m x_i \dim(V_i \cap V) \leq \alpha \dim V$ for any subspace $V$.
\end{lemma}
\begin{proof}
    For any subspace $V$, multiplying the orthogonal projection $\pi_V$ onto $V$, we obtain $\sum_{i=1}^m x_i \pi_{V_i}\pi_V \preceq \alpha \pi_V$.
    Taking the trace of both sides, $\sum_{i=1}^m x_i \tr(\pi_{V_i}\pi_V) \leq \alpha \dim V$.
    Now the claim follows from $\tr(\pi_{V_i}\pi_V) \geq \dim(V_i \cap V)$ and $x_i \geq 0$ for each $i$.
\end{proof}

For any subspace $V$, the subspace $\tilde{V} = g V$ satisfies $\Img(A_i) \cap V = \Img(\tilde{A}_i) \cap \tilde{V}$.
By Lemma~\ref{lem:proj-ineq} and \eqref{eq:fracmat-proj}, we have
\[
\sum_{i=1}^p x_i \dim(\Img(\tilde A_i)\cap\tilde{V}) \leq (1 + \eps)\dim\tilde{V} = (1 + \eps)\dim V,
\]
so $x/(1+\eps)$ is a fractional linear matroid matching of size $\frac{k}{2(1+\eps)}$.
Summarizing this argument, we obtain Algorithm~\ref{alg:fracmat}.

\begin{Algorithm}
Algorithm \textsc{FracLinearMatroidMatching}$(B_1, \dots, B_p)$:
\begin{description}
\item[\hspace{.2cm}\textbf{Input:}]
Lines $\{a_i, b_i\} \subseteq \C^n$ ($i = 1, \dots, p$), a parameter $\eps > 0$.
\item[\hspace{.2cm}\textbf{Output:}] A $(1-\eps)$-maximum fractional linear matroid matching $x$.
\item[\hspace{.2cm}\textbf{Algorithm:}]
\end{description}
\begin{enumerate}
    \item Compute $A_i = a_i \wedge b_i$ for each $i$. Let $\caA$ be the matrix space spanned by $A_i$.
    \item Determine the noncommutative rank $k$ of $\caA$ with \textsc{DecisionSinkhorn} (Algorithm~\ref{alg:sinkhorn-dec}) and binary search.
    \item Find a matrix $g$ that satisfies the scaling condition~\eqref{eq:fracmat-scaling} with Algorithm~\ref{alg:maj-scaling}.
    \item For each $i$, let $\tilde A_i = g A_i g^T$ and $x_i = \frac{\abs{\det R_i}^2}{1+\eps}$, where $R_i$ is the upper-triangular factor of the thin QR-decompostion of $\tilde A_i$.
    \item \textbf{output} $x$.
\end{enumerate}
\caption{Deterministic algorithm for finding an approximate maximum fractional linear matorid matching.}\label{alg:fracmat}
\end{Algorithm}

We obtain the following theorem from Theorems~\ref{thm:maj-sinkhorn} and \ref{thm:decision-correct}.

\begin{theorem}\label{thm:fracmat}
    For any $\eps > 0$, Algorithm~\ref{alg:fracmat} finds a $(1-\eps)$-maximum fractional linear matroid matching in $O(n^3(n+p)(n^2\log^2 n + \log n\log M + \frac{1}{\eps^2}(n\log n + \log Mp)))$ arithmetic operations.
    Furthermore, the bit complexity of intermediate numbers is $\poly(n,p,\log M)$.
\end{theorem}

For a constant $\eps > 0$, Algorithm~\ref{alg:fracmat} runs in $\tilde O(n^5(n+p))$ time, where the $\tilde O$ notation hides the $\poly(\log n, \log M, 1/\eps)$ factor, which improves $O(n^3p^4)$ time of \cite{Chang2001b} for fractional linear matroid matching.
Furthermore, we can also bound the bit complexity of intermediate numbers, whereas \cite{Chang2001b} did not analyze the bit complexity.

\FloatBarrier

\subsection{Scaling characterization of the fractional linear matroid matching polytope} \label{sec:scaling}
We show a scaling characterization of $P(B)$ analogously to Theorem \ref{thm:bl}. The characterization is identical to that of Theorem \ref{thm:bl} apart from the lower bound by $(1 - \eps)I_n$ in the second item and the constraint $X \succeq I_n$ in the third item.

\begin{theorem} \label{thm:bl-less}
Suppose each $B_i$ has rank $n_i$ and $x \in \R^p_{\geq 0}$. The following are equivalent:
\begin{enumerate}
\item $x \in P(B)$.
\item For every $\eps > 0$ there is an invertible matrix $g \in \R^{n \times n}$ and invertible matrices $h_i \in \R^{n_i \times n_i}$ such that $\tilde{B}_i:= h_i B g^T$ satisfies
\begin{align*}
\sum x_i \tilde{B_i}^\dagger \tilde{B_i} &\preceq (1 + \eps) I_n\\
\text{and } \tilde{B_i} \tilde{B_i}^\dagger &= I_{n_i}\;\forall i \in [p].\end{align*}
\item The following quantity is finite:
$$\capa_x B:= \inf_{X\succeq I_n, Z_i \succ 0} \sum_{i = 1}^n \tr B_i X^{-1} B_i^\dagger Z_i^{-1} + x_i \log\det Z_i + \log\det X.$$
\end{enumerate}
Moreover, suppose $x = c/d$ for $c \in \Z_{\geq 0}^p$, $d \in \Z_{> 0}$, and that the entries of $B$ are Gaussian integers with absolute value at most $\log M$. Then it is possible to decide if $x \in P(B)$ in time $\poly( k, d, n, \log M)$ where $k = \sum_i c_i n_i$.
\end{theorem}

\begin{proof}Let $S$ be the set of $x \geq 0$ for which the desired $g,h_i$ exist, and let $C$ be the set of $x \geq 0$ for which $\capa_x B$ is finite. It is straightforward from calculus $C \subseteq S$, i.e. (3) implies (2), because if the capacity is finite then we may take $g^\dagger g = X$ and $h_i^\dagger h_i = Z_i$ from $X,Z$ with subgradients of the objective function of norm $\eps$. Furthermore, $S$ is convex by the log-convexity of $\capa_x \Phi_B$ in $x$.

We may also easily check that $S \subseteq P(B)$, i.e. (2) implies (1), by checking that $x \in S$ implies the inequalities for $P(B)$. Let $ \eps > 0$ to be determined, and let $\tilde{B_i}$ be as in the theorem statement.
For any subspace $V$, the subspace $\tilde{V} = g^{\dagger} V$ satisfies $\row B_i \cap V = \row \tilde{B}_i \cap \tilde{V}$. Consider the orthogonal projection $\pi_{\tilde{V}}$ onto $\tilde{V}$.
By the condition on the $\tilde{B}_i$'s and Lemma~\ref{lem:proj-ineq}, we have that
\[
\sum x_i \dim \row B_i \cap V \leq (1 + \eps) \dim V.
\]
Taking $\eps \to 0$, we have $x \in P(B)$.

 In summary, so we have
 $$ C \subseteq S \subseteq P(B)$$
 or equivalently $(3) \implies (2) \implies (1)$. To show that these containments are in fact equalities, and so $(1) \implies (3)$, it suffices to show that every extreme point of $P(B)$ is contained in $C$. To do this, observe that every extreme point $x$ of $P(B)$ may be written as $c/d$ where $d \in \Z$ and $c \in \Z^p$. We may assume all $x_i > 0$, and let $k = \sum_i c_i n_i$. For contradiction, assume $\capa_x \Phi_B = -\infty.$ Construct a CP map $\Psi$ as in \cite{Garg2018} with $k \times nd$ Kraus operators with a block structure having $p$ groups of rows with $c_i n_i$ consecutive rows each, broken into $c_i$ groups of $n_i$ consecutive rows and columns broken into into consecutive groups of $n$ columns. Each of the $d \cdot \prod_i c_i$ many Kraus operators has exactly one of these $n_i \times n$ blocks filled with $B_i$ and the rest equal to zero. It is straightfoward to check that $\capa_k  \Psi = d \capa_x cB = -\infty$ where $cB = (c_1B_1, \dots, c_p B_p)$. By Theorem \ref{thm:k-gurv}, there is an independent set $L,R$ for $\Psi$ such that $\dim L + \dim R \geq k + nd - k = nd.$ We only increase the dimension of $L$ and $R$ by assuming $R = \oplus_{i = 1}^d V$ for $V \subseteq \R^n$ and $L = \oplus_{i = 1}^p \oplus_{j = 1}^{n_i} (B_i V)^\perp$. Thus $$\sum_i c_i \dim (B_i V)^\perp + d \dim V > nd.$$
Using $\dim (B_i V)^\perp = \dim V^\perp \cap \row B_i$ for $B_i$ full row rank and $x_i = c_i/d$, we have
$ \sum_i x_i \dim \row B_i \cap V^\perp  >  d \dim V^\perp$, contradicting $x \in P(B)$.
\end{proof}

The construction of $\Psi$ yields an algorithm for deciding membership in $P(B)$: simply construct $\Psi$ from $B,c,d$ and check if $\capa_k  \Psi > -\infty$ for $k=\sum_i c_i n_i$. By the argument above, this is the case if and only if $x \in P(B)$. The runtime follows from Theorem \ref{thm:decision-correct} applied to $\Psi$. A pseudocode is shown in Algorithm~\ref{alg:bl-membership}.

\begin{Algorithm}
Algorithm \textsc{MemEpsBL}$(B_1, \dots, B_p)$:
\begin{description}
\item[\hspace{.2cm}\textbf{Input:}]
A $p$-tuple of $n_i \times n$ complex matrices $B = (B_1, \dots, B_p)$, a vector $x \in \R^p_{> 0}$ and $\sum_i x_i \leq 1$, a running time $T$,
and a parameter $\eps > 0$.
\item[\hspace{.2cm}\textbf{Output:}] A feasible point $y$ such that $\|x - y\| \leq \eps$ or that $x$ is at least $\eps$-far from $P(B)$.
\item[\hspace{.2cm}\textbf{Algorithm:}] Set $X = I_n, Y_i = I_{n_i}$, and $\delta := \eps/ (2\sqrt{n})$.  \textbf{For} $t = 1, \dots, T$:
\end{description}
    \begin{enumerate}
    \item \textbf{Update }$Y$: Set $Y_i = x_i^{-1} B_i X^{-1} B_i^\dagger$ and then set $Y_i = Y_i/\det(Y_i)^{1/n_i}$.
    \item \textbf{Update }$X$:
    \begin{itemize}
    \item Let  $\mu = \lambda(\sum_i  X^{-1/2} B_i^\dagger Y_i^{-1} B_i X^{-1/2} )$. Set $\eps_X = \sum_{\mu_i > 1} \mu_i - 1 - \log \mu_i$.
    \end{itemize}
    \begin{description}
    \item[\textbf{If:}] $\eps_X \leq \frac{\delta^2}{2}$: \textbf{Output} $y_i = (e^{-\eps}/n_i) \tr X^{-1/2} B_i^\dagger Y_i^{-1} B_i X^{-1/2}$.

    \item[\textbf{Else: }] Choose $X$ in the same eigenbasis as $\sum_i  B_i^\dagger Y_i^{-1} B_i$ with eigenvalues $\lambda_i(X) = \max \{1, \mu_i\}$.
    \end{description}
    \end{enumerate}
\hspace{.2cm}\textbf{Output:} that $x$ is not in $P(B)$.
\caption{Algorithm for deciding $\eps$-membership in the Brascamp-Lieb polytope.}\label{alg:bl-membership}
\end{Algorithm}

We now outline why this algorithm works. We keep it brief because it is so similar to the proof of Corollary \ref{cor:sink-running}.
\begin{proof}[Proof outline of Theorem \ref{thm:bl-membership}]
The proof follows a similar structure to typical analyses of Sinkhorn-type algorithms. We use the objective function $f(X,Y):=\sum_{i = 1}^n \tr B_i X^{-1} B_i^\dagger Y_i^{-1} + x_i \log\det Y_i + \log\det X$ from Theorem \ref{thm:bl-less}.

First suppose that $x \in P(B)$. The capacity lower bound independent of $x_i$ is similar to the proof of \ref{thm:cap-lb}. It follows by using the reduction in the proof of Theorem \ref{thm:bl-less} on the vertices of $P(B)$ and applying concavity of $\capa_{x} (B)$ in $x$. The progress bound is very similar to the proof of Theorem \ref{thm:maj-sinkhorn}. This time the generalized KL projections are to the \emph{point} $\oplus_i x_i \mathbf 1_{n_i}$ (the $Y_i$ update) and the polytope $\{\mu: \mu \leq \mathbf 1_n\}$ (the $X$ update). The former makes progress because it is an alternating minimization step, and the latter was already analyzed Theorem \ref{thm:maj-sinkhorn}.

Finally, Lemma \ref{lem:kl-scale} shows that the output $y$ is in $P(B),$ and it is equal to $e^{-\eps} x$, hence the distance bound. On the other hand, if $x$ is at least $\eps$-far from $P(B)$, then the Sinkhorn steps will not terminate and hence the algorithm will correctly output that $x$ is not in the polytope.
\end{proof}

\begin{proof}[Proof of Theorem \ref{thm:bl-conp}]
As mentioned earlier, an \textsc{NP} certificate for $x^*\in BL(B)$ can be obtained by decomposing $x^*$ as a convex combination of half-integral vertices. These vertices and convex coefficients form a certificate with polynomial bit-complexity. Note that the membership of a half-integral vector can be checked in polynomial time by the result of \cite{Garg2018} because the common denominator is a constant.

For an \textsc{NP} certificate for $x^*\notin BL(B)$, we can proceed as follows. Wlog, we can assume that $\sum_i n_i x_i^* =n$, and therefore we need to certify that $x^*\notin P(B)$.  Not being in $P(B)$ means that there exists a facet-defining valid inequality $\alpha^T x\leq \beta$ for $P(B)$ with $\alpha\in\{0,1,2\}^p$ and $\beta\in \Z_+$ violated by $x^*$. The certificate consists of $p$ affinely independent half-integral points (or vertices) $x^{(i)}\in \frac{1}{2}\Z^p\cap BL(B)$ for $i=1,\cdots,p$ satisfying $\alpha^Tx^{(i)}=\beta$ for all $i$. Either $\alpha\in\{0,1,2\}^p$ and $\beta \in \Z_+$ are also part of the certificate, or the verifier first computes them (by solving a system of linear equations), and the verifier then checks that $\alpha^T x^* >\beta$. The next task for the verifier is to check that $\alpha$ and $\beta$ indeed defines a valid inequality for $P(B)$. For this purpose, the verifier computes 
$$y^*=\frac{1}{p} \sum_{i=1}^{p} x^{(i)} + \frac{1}{4p^2} \bfone.$$ We claim that the validity of $\alpha^T x\leq \beta$ over $P(B)$ can be checked by verifying that $x^{(i)}\in P(B)$ for every $i$ and that $y^*\notin P(B)$. This would complete the proof as these tasks can be done in polynomial time using the result of  \cite{Garg2018} since the common denominator of these points is either $2$ (for the $x^{(i)}$'s) or $4p^2$ (for $y^*$). 

To prove the claim, consider any facet-defining valid inequality $\gamma^T x \leq \delta$ for $P(B)$ with $\gamma\in\{0,1,2\}^p$ and $\delta\in \Z_+$ for which $\gamma^Ty^* >\delta$. Such inequality must exist. If this inequality is not $\alpha^Tx\leq \beta $ then one of the $x^{(i)}$'s would satisfy $\gamma^T x^{(i)} \leq \delta -\frac{1}{2}$ by half-integrality, and therefore
$$\gamma^Ty^* = \frac{1}{p} \sum_{i=1}^{p} \gamma^Tx^{(i)} + \frac{1}{4p^2} \gamma^T \bfone \leq \delta-\frac{1}{2p} + \frac{1}{2p} = \delta,$$ a contradiction. Thus, the unique $\{0,1,2\}$ facet-defining inequality separating the $x^{(i)}$'s from $y^*$ must be $\alpha^Tx\leq \beta$, proving its validity. 
\end{proof}

\subsection{Pseudopolynomial weighted optimization algorithm}\label{sec:pseudopoly}
In this section we describe algorithm \textsc{WeightedSinkhorn} (Algorithm \ref{alg:weighted-sinkhorn}) to find the optimum value $OPT$ of the linear program
\begin{alignat*}{2}
    &\text{maximize}  \quad && w^T x  \\
    &\text{subject to}\quad && \sum_{i=1}^p x_i \dim(\row B_i \cap V) \leq \dim V, \qquad \text{($V$: subspace of $\C^n$)} \\
    &\quad && x_i \geq 0 \qquad (i=1, \dots, p),
\end{alignat*}
i.e. $OPT = \max\{w^T x: x \in P(B)\}$, in time $\poly(\|w\|_1,m,n)$. We will find a feasible point $x$ such that $w^T x > OPT - 1/2$, which if all $n_i = 2$ implies $OPT =  \lceil 2x \rceil/2$ by half-integrality of the feasible region. Let $H_{w,k} = \{x:w^T x \geq k\},$ and let $D(H_{w,k}|| \alpha)$ denote the generalized KL distance from $\alpha$ to $H_{w,k}$. We assume that $B_i$ have full row rank, and so $\|x\|_\infty \leq 1$ for all $x \in P(B)$ and $e_i \in P(B)$ for all $i$. Thus $\|w\|_\infty \leq OPT \leq \|w\|_1$. We also assume $w\geq 0$ because the negative coordinates will be optimized at $x_i = 0$.

\begin{Algorithm}
Algorithm \textsc{WeightedSinkhorn}$(B_1, \dots, B_p)$:
\begin{description}
\item[\hspace{.2cm}\textbf{Input:}]
A $p$-tuple of $n_i \times n$ complex matrices $B = (B_1, \dots, B_p)$, a vector $w \in \R^p_{\geq 0}$, a number $k \in \Z_{> 0}$, a running time $T$,
and a parameter $0 < \eps$.
\item[\hspace{.2cm}\textbf{Output:}] A feasible point $x$ such that $w^T x > (1 - \eps) k$ if $OPT \geq k$, and that $OPT < k$ otherwise.
\item[\hspace{.2cm}\textbf{Algorithm:}] Set $X = I_n, Y_i = I_{n_i}, z = 0$, and $\delta := \eps/ (2\sqrt{n})$. Define $\tilde w_i := w_i/n_i$. \textbf{For} $t = 1, \dots, T$:
\end{description}
    \begin{enumerate}
    \item \textbf{Update }$Y$: Set $Y_i = e^{\tilde w_i z} B_i X^{-1} B_i^\dagger$ and then set $Y_i = Y_i/\det(Y_i)^{1/n_i}$.
    \item \textbf{Update }$X,z$:
    \begin{itemize}
    \item Let  $\mu = \lambda(\sum_i e^{\tilde w_i z} X^{-1/2} B_i^\dagger Y_i^{-1} B_i X^{-1/2} )$. Set $\eps_X = \sum_{\mu_i > 1} \mu_i - 1 - \log \mu_i$.
    \item For $i \in [p]$ set $\nu_i = \tr e^{\tilde w_i z} B_i^\dagger Y_i^{-1} B_i X^{-1}$
     and $\eps_Z =  D(H_{\tilde w,k} || \nu).$
    \end{itemize}
    \begin{description}
    \item[\textbf{If:}] $\eps_X, \eps_Z \leq \frac{\delta^2}{2}$: \textbf{Output} $x_i = e^{-\delta} \nu_i/n_i$.

    \item[\textbf{Else: }]$ $
    \begin{description}
    \item[If:] $\eps_Z$ is larger: set $\alpha_i = \tr  B_i^\dagger Y_i^{-1} B_i X^{-1}$ and choose $z$ to be the minimum value such that $\sum_{i = 1}^p \tilde w_i \alpha_i \geq k$.
    \item[If:] $\eps_X$ is larger: choose $X$ in the same eigenbasis as $\sum_i e^{\tilde w_i z} B_i^\dagger Y_i^{-1} B_i$ with eigenvalues $\lambda_i(X) = \max \{1, \mu_i\}$.
    \end{description}
    \end{description}
    \end{enumerate}
\hspace{.2cm}\textbf{Output} that $OPT < k$.
\caption{Algorithm for optimizing a linear functional over the Brascamp-Lieb polytope.}\label{alg:weighted-sinkhorn}
\end{Algorithm}

To analyze the algorithm, we define another capacity.
\begin{align}
f_{w,k} ( X, Y, z) %
= \sum_{i = 1}^p e^{w_i z/n_i} \tr  B_i X^{-1} B_i^T Y_i^{-1} - kz  + \log \det X. \label{eq:fwk}
\end{align}
Consider the geodesically convex domain $\caD' = \{(X,Y,z):X \succ I_n, \det Y_i = 1, z \geq 0\},$ and define
\begin{align} \capa_{w,k} B:= \inf_{\upsilon \in \caD'} f_{w,k} (\Upsilon).\label{eq:cwk}\end{align}

We will show that our modified capacity characterizes the optimum over the BL polytope. First, we show that it acts as a Lyapunov function much as in Theorem \ref{thm:maj-sinkhorn}.
\begin{theorem}\label{thm:bl-sinkhorn} Suppose $w \in \R^p_{\geq 0}$, $k \geq \|w\|_\infty$, and $0 \leq \eps \leq 1$.
\begin{itemize}
\item  If \textsc{WeightedSinkhorn} terminates before step $T$, then the output is a feasible point $x$ with $w^T x \geq (1 - \eps)k$.
\item If $\capa_{w,k} B$ is finite and $\sum_i \|B_i\|_F^2\leq 1$, \textsc{WeightedSinkhorn} terminates before step
$$T = 2 \cdot \frac{1 - \capa_{w,k} B }{ \eps^2 }.$$
\end{itemize}\end{theorem}
The proof is so similar to that of Theorem \ref{thm:maj-sinkhorn} that we omit many of the details, but we include a sketch below. %
We'll use a standard convex duality lemma for generalized KL projections.
\begin{lemma}\label{lem:1d-bregman}
Let $0 \neq \alpha \in \R^p_{> 0}$, $w \in \R^p_{\geq 0}$, and $k\geq 0$. Then
$$ D(H_{w,k}||\alpha) = \sum_{i = 1}^p \alpha_i - \min_{z \geq 0}\left( \sum_{i = 1}^p e^{w_i z} \alpha_i - kz\right).$$
where $D(H_{w,k}||\alpha)$ is the generalized Kullback-Leibler distance from $\alpha$ to the half-space $H_{w,k} = \{x:w^T x \geq k\}$. Furthermore, $\arg\min_{z \geq 0}\left( \sum_{i = 1}^p e^{w_i z} \alpha_i - kz\right) = \min\{z:\sum_i w_i e^{w_i z} \alpha_i \geq k\} \cup \{0\}$.
\end{lemma}
In particular this lemma shows that the $z$-update step minimizes $f_{w,k}$ in $z$ while holding the other variables fixed. If $w^T \alpha \geq k$, then $\inf_{z \geq 0} \sum_{i = 1}^p e^{w_i z} \alpha_i - kz$ is simply $\sum_i \alpha_i$ because $D(\alpha||\alpha) = 0$.
We also have an analogue of Lemma \ref{lem:prog-bound} for the $z$ update.
\begin{corollary}[$z$-update progress bound] \label{cor:z-prog}
Let $z \geq 0$, $\alpha_i \in \R^p_{> 0}$ and $w \in \R^p$. Let $\nu_i = e^{w_i z} \alpha_i$.
Then
$$\sum_{i = 1}^p e^{w_i z} \alpha_i - kz \geq D(H_{w,k}||\nu) + \sum_{i = 1}^p \alpha_i - D(H_{w,k}||\alpha).$$

In particular, the progress made in the $z$ update of \textsc{WeightedSinkhorn} is at least $D(H_{w,k}||\nu)\geq\delta^2/2$.

\end{corollary}
\begin{proof}
Like Lemma \ref{lem:prog-bound}, the proof is a change of variables. By Lemma \ref{lem:1d-bregman},
\begin{align*}\sum_i \alpha_i - D(H_{w,k}|| \alpha)
&= \min_{z' \geq 0}\left( \sum_{i = 1}^p e^{w_i z'} \alpha_i - kz'\right)\\
&\leq \min_{z' \geq 0} \left( \sum_{i = 1}^p e^{w_i (z' +z)} \alpha_i - k(z' +z)\right)\\
& = - kz + \min_{z' \geq 0} \left( \sum_{i = 1}^p e^{w_i z'} \nu_i - kz'\right)\\
&= - kz + \sum_{i = 1}^p e^{w_i z} \alpha_i - D(H_{w,k}|| \nu).
\end{align*}
Rearranging the terms yields the inequality. The progress bound follows from Lemma \ref{lem:1d-bregman}, because the terms that depend on $z$ change from $\sum_{i = 1}^p e^{w_i z} \alpha_i - kz$ to $\min_{z' \geq 0} \left( \sum_{i = 1}^p e^{w_i z'} \alpha_i - kz'\right) = \sum_i \alpha_i - D(H_{w,k}|| \alpha).$
\end{proof}
\begin{proof}[Proof of Theorem \ref{thm:bl-sinkhorn}]
We first verify that, assuming the algorithm terminates, the output is a feasible point. We claim that $\tilde B_i = \sqrt{ n_i/\nu_i}e^{w_i z/n_i} Y_i^{-1/2} B_i X^{-1/2}$ and $x$ satisfy the conditions in Theorem \ref{thm:bl-less} for feasibility of $x$.
As in the algorithm, let $\tilde w_i = w_i/n_i$. By the $Y$ update step previously, $\tilde B_i \tilde B_i^\dagger = I_{n_i}$. Hence $\tilde B_i \tilde B_i^\dagger/ (\nu_i/n_i) = I_{n_i}$. On the other hand, $\sum_i x_i \tilde B_i^\dagger \tilde B_i = e^{- \delta} \sum_i e^{\tilde w_i z} X^{-1/2} B_i^\dagger Y_i^{-1} B_i X^{-1/2},$ which has spectrum $e^{-\delta} \mu$ for $\mu$ as in the $X,Y$ update step of the algorithm. Note that $D(P_{\mathbf 1_n}|| \mu)= \sum_{\mu_i > 1} \mu_i - 1 - \log \mu_i,$ so Lemma \ref{lem:kl-scale} implies $e^{-\delta} \mu \leq \mathbf 1_n$. This implies $x$ is feasible.

Next we verify that $w^T x$ is at least $(1 - \eps) k$. By the termination condition, we know there is some $\beta \in \R^p_{\geq 0}$ such that $\tilde w^T \beta \geq k$ and $D( \beta || \nu) \leq \delta^2/2$. Then we have $w^Tx  = e^{-\delta} \tilde w^T \nu$, and
$$\tilde  w^T \nu = e^{-\delta} (\tilde w^T (\nu - \beta) + \tilde w^T\beta \geq k - \|\tilde w\|_\infty \|\nu - \beta\|_1. $$

By Lemma \ref{lem:pinsk}, $D(\beta|| \nu) \geq \min\left\{\frac{1}{4 \|\nu\|_1} \|\nu - \beta\|_1^2, (1 - \ln 2)\| \alpha - \beta\|_1.\right\}$. Because $e^{-\delta} \mu \preceq 1_n$, we know that $\|\nu\|_1 \leq  \sum_i \nu_i n_i \leq e^{\delta} n.$ Thus $\|\nu - \beta\|_1 \leq \max\{\frac{\delta^2}{(1 - \ln 2) 2}, \sqrt{2 n} e^{\delta/2} \delta \}.$ For $\eps \leq 1$, the latter argument dominates.
Combining the inequalities we have $w^T x \geq e^{-\delta} (k - \|w\|_\infty \sqrt{2 n} e^{\delta/2} \delta) \geq  (1 - (\sqrt{2n} + 1)\delta) k\geq (1 - \eps) k.$

We next verify that $f_{w,k}$ is monotone decreasing. For $X$ and $Y$ one can easily verify that each update step minimizes $f_{w,k}$ in the corresponding variable while holding the others fixed. In fact, the $X$ update is the same as the update in \textsc{MajSinkhorn} for $\alpha = \mathbf 1_n$.
For the $z$ update this follows from Lemma \ref{lem:1d-bregman}. It remains to check that the progress made each iteration is at least $\delta^2/2$ whenever the termination condition is not met. If the $X$ variable is updated, this follows from Lemma \ref{lem:prog-bound}, and if $z$ is updated it follows from Corollary \ref{cor:z-prog}.
\end{proof}

\begin{theorem}\label{thm:bl-capacity}
$\capa_{w,k} B > - \infty$ if and only if $OPT:=\max\{w^T x : x \in BL(B)\} \geq k$. Moreover, if $B$ has Gaussian integer entries then $OPT \geq k$ implies $-\capa_{w,k} \Phi_B  = O( k (n + m) \log (m + n)).$
\end{theorem}
We first prove a lemma relating $\capa_{w,k}$ to the capacity for majorization scaling.
\begin{lemma}\label{lem:bl-cap-ineq}
Suppose $x \in \R_{\geq 0}^p$ is such that $w^T x \geq k$. Then $\capa_{w,k} B \geq \capa_x B$.
\end{lemma}
\begin{proof}
We have
\begin{align*}
\capa_{w,k} B &= \inf_{(X,Y,z) \in \caD'}\sum_{i = 1}^p e^{w_i z/n_i} \tr  B_i X^{-1} B_i^T Y_i^{-1} - kz  + \log \det X \nonumber\\
& =  \inf_{(X,Y,z) \in \caD'} \sum_{i = 1}^p (\tr  B_i X^{-1} B_i^T (e^{w_i z/n_i} Y_i)^{-1}  +  x_i \log \det (e^{-w_i z/n_i} Y_i) + w_i x_i ) - kz  + \log \det X \nonumber\\
& \geq \inf_{X \succeq I_n, Z_i \succ 0} \sum_{i = 1}^p( \tr  B_i X^{-1} B_i^T Z_i^{-1}  +x_i  \log \det Z_i)   + \log \det X.%
\end{align*}
In the last inequality we used that $w^T x \geq k$ and $z\geq 0$.
\end{proof}
We can now easily prove Theorem \ref{thm:bl-capacity}.
\begin{proof}[Proof of Theorem \ref{thm:bl-capacity}] For the foward direction, if $\capa_{w,k} B > -\infty$ then for any $\eps > 0$ our algorithm produces a point with $P(B)$ with $w^T x > (1 - \eps) k$ by Theorem \ref{thm:bl-sinkhorn}; hence $OPT \geq k$. Now suppose that $OPT \geq k$. Then there is a point $x \in P(B)$ with $w^T x \geq k$. By Theorem \ref{thm:bl-less}, $\capa_x B \geq -\infty$. By a similar argument to the proof of Theorem \ref{thm:cap-lb} (which we omit to avoid redundancy), if $B$ has Gaussian integer entries then $-\capa_x B = O( k (n + m) \log (m + n) + k \log p) = O(k (m + n) \log (m+n))$.  \end{proof}

\begin{theorem}\label{thm:weighted-sinkhorn-decision} Suppose $B$ has Gaussian integer entries of absolute value at most $M$, that all $n_i \in \{1,2\}$, and that $k \geq \|w\|_\infty$ and $w\in\Z^p_{\geq 0}$. Then \textsc{WeightedSinkhorn} run on $B/\sqrt{mnM}$ terminates before step
$$T = O((k (m + n) \log (m+n) + n \log M)/\eps^2)$$
if $OPT \geq k$ and outputs a feasible point $x\in P(B)$ with $w^T x \geq k - \eps$, and does not terminate if $OPT \leq k - \eps$. As $OPT \leq \|w\|_1$, we may find $OPT$ to precision $2\eps$ by binary search over values of $k$ in $O(\log (\|w\|_1/\eps))$ calls to the algorithm and hence
$$ O((\|w\|_1 (m_1 + n) \log (m_1+n) + n \log M) (m_1 n^2 + n^3 +m_2 n + m_3) \log (\|w\|_1/\eps) /\eps^2)$$
arithmetic operations, where $m_1 := \sum_{i = 1}^p n_i$, $m_2 := \sum_{i = 1}^p n_i^2$, and $m_3 := \sum_{i = 1}^p n_i^3$. If all $n_i \leq 2$ and $\eps = 1/4k$ then the algorithm correctly computes $OPT = \max\{w^T x :x \in P(B)\}$ in
$$O((\|w\|_1(m + n) \log (m + n) + n \log M) \|w\|_1^2 \log \|w\|_1)$$
iterations and
$$ O((\|w\|_1(m + n) \log (m + n) + n \log M) (m_1 n^2 + n^3 +m_2 n + m_3)\|w\|_1^2 \log \|w\|_1)$$
arithmetic operations. \end{theorem}
We can check if $BL(B)$ is empty by computing the value for $w = (n_1, \dots, n_1)$. If $BL(B) \neq \emptyset$, optimizing $w^Tx$ over $BL(B)$ reduces to optimizing $2\|w\|_1 (n_1, \dots, n_p) + w$ over $P(B)$.

\begin{proof} First assume $OPT \geq k$. Note that $B/\sqrt{mnM}$ has size at most $1$. By the same proof as Lemma \ref{lem:cap-scalar}, $\capa_{w,k} B/\sqrt{mnM} \geq \capa_{w,k} B - n \log mnM$, which is $- O(k (m + n) \log (m+n) + n \log M)$ by Theorem \ref{thm:bl-capacity}. By Theorem \ref{thm:bl-sinkhorn}, if $OPT \geq k$ \textsc{WeightedSinkhorn} terminates before step
$$T = O((k (m + n) \log (m+n) + n \log M)/\eps^2) = O(k^3(m + n) \log (m + n) + k^2 n \log M)$$
and outputs a feasible point $x\in P(B)$ with $w^T x \geq k - \eps$. If $OPT \leq k - \eps$, the algorithm cannot terminate before step $T$ or else it outputs a feasible point with value $\geq k - \eps$, a contradiction. The case of $n_i \leq 2$ follows by half-integrality.\end{proof}
\begin{remark}
If we do not have all $n_i = 2$, the algorithm still runs in time $O((k (m + n) \log (m+n) + n \log M)/\eps^2)$ - only we may need to take $\eps$ exponentially small to certify $OPT \geq k$, because $OPT$ need not be half-integral. The number of arithmetic operations per step is $O(m n^2 + n^3 + \sum_i n_i^2 n + \sum_i n_i^3)$.

\end{remark}

We now give another proof of Theorem~\ref{thm:bl-conp}.

\begin{proof}[Another Proof of Theorem~\ref{thm:bl-conp}]
It suffices to prove the statement for $P(B)$; the claim for $BL(B)$ follows immediately because $BL(B) = P(B) \cap \{x: \sum_i n_i x_i = n\}$.

Suppose $x \in P(B)$. Then $x$ is in the convex hull of at most $p + 1$ vertices of $P(B)$, which are half-integral. By Theorem \ref{thm:bl-less}, there is a polynomial time algorithm to decide if a half-integral vector is in $P(B)$, so the $p + 1$ vertices are a polynomial length certificate that $x \in P(B)$.

Next suppose $x \not\in P(B)$. Then there is some subspace $V \subseteq \C^n$ such that $\sum_i x_i \dim \row B_i \cap V > \dim V$. If we let $w_i = \dim \row B_i \cap V$, then we have $OPT = \max\{w^T x' :x' \in P(B)\} < w^T x.$ We can verify this in polynomial time using Theorem \ref{thm:weighted-sinkhorn-decision}, so $w$ is a certificate that $x \not \in P(B)$. \end{proof}

\ifanonymous\else
\section*{Acknowledgements} CF acknowledges helpful conversations with Akshay Ramachandran and Harold Nieuwboer.
The authors thank Nikhil Srivastava for references to eigendecomposition algorithms. 
The authors also thank Koyo Hayashi and Hiroshi Hirai for references~\cite{Aas2014a,Gietl2013} and sharing a draft of~\cite{Hayashi2022}.%
\fi

\bibliographystyle{alpha}
\bibliography{rank2BL}

\clearpage
\appendix

\section{Bit complexity of linear algebra operations}

\subsection{Models of computations}
We need two models of computations of rational numbers, which we formalize below.

\paragraph{Rational Arithmetic.}
A rational number $r$ is represented as $r = p / q$, where $p \in \Z$ and $q \in \Z_{>0}$.
This model of arithmetic is suitable to represent rational numbers exactly such as results of Gaussian elimination and Gram-Schmidt orthogonalization.
A caveat of this model is that careless implementation of an algorithm may explode the bit length to represent integers $p, q$, even if it makes only polynomially many arithmetic operations.

\paragraph{Finite Precision Arithmetic.}
A rational number $r$ is rounded to a fixed length of bits, where the length may depend on the size of the problem and desired accuracy.
We call the length of bits the \emph{precision}.
This model of arithmetic is more suitable to numerical computations such as eigendecompositions.

\subsection{Gaussian elimination and pseudoinverse}
The following is a well-known theorem for Gaussian elimination.
\begin{theorem}[\cite{grotschel2012geometric}]
    For $A \in \Q[\sqrt{-1}]^{n \times n}$, the Gaussian elimination algorithm runs in $O(n^3)$ arithmetic operations.
    If $A$ have Gaussian integer entries with magnitude at most $M$, the output as well as intermediate numbers are rationals with numerator and denominator at most $e^{O(n\log (nM))}$.
\end{theorem}

A nonsingular matrix $A^+$ is called a \emph{pseudoinverse} of $A$ if $A^+ A$ is the identity map restricted on $\Img(A)$.
It is easy to construct $A^+$ by the result of Gaussian elimination.
More specifically, let
\[
    A = P \begin{bmatrix}
        I_r & O_{r, n-r} \\
        O_{n-r,r} & O_{n-r,n-r}
    \end{bmatrix} Q^{-1}
\]
be a rank normal form obtained by Gaussian elimination, where $P,Q$ are nonsingular and $r = \rk(A)$.
The first $r$ columns of $Q$ spans $\Img(A)$.
We can see that $QP^{-1}$ is a pseudoinverse of $A$.

\subsection{Gram-Schmidt orthogonalization}
Given an ordered basis $(a_1, \dots, a_n) \subseteq \Q[\sqrt{-1}]^m$ of a subspace $U$, the Gram-Schmidt orthogonalization computes orthogonal ordered basis $(b_1, \dots, b_n)$ of $U$.
The Gram-Schmidt algorithm can be carried out in the same time and bit complexity as Gaussian elimination~\cite{grotschel2012geometric}.

Once we compute $(b_1, \dots, b_n)$, the orthogonal projection onto $U$ is given by $\pi_U = \sum_{i=1}^n \frac{b_i b_i^\dagger}{\norm{b_i}^2}$.
The orthogonal projection onto $U^\perp$ is simply $I - \pi_U$.

\begin{theorem}\label{thm:Gram-Schmidt}
    Given a basis $a_1, \dots, a_n \in \Q[\sqrt{-1}]^m$ of a subspace $U$, the orthogonal projections onto $U$ and its orthogonal complements can be found in $O(mn^2)$ arithmetic operations.
    If $a_1, \dots, a_n$ have Gaussian integer entries with magnitude at most $M$, the output as well as intermediate numbers are rationals with numerator and denominator at most $e^{O(m\log (mM))}$.
\end{theorem}

\subsection{Eigendecomposition of Hermitian matrices}
Given a Hermitian matrix $A \in \Q[\sqrt{-1}]^{n \times n}$, we are to find all pairs of eigenvalues and eigenvectors.
Since eigenvalues can be irrational, the best we can hope for is to approximate them in desired accuracy.
The following result is essentially achieved by the QR algorithm with Wilkinson's shift and the standard $O(n^3)$-time matrix multiplication.
For more recent improvements, see \cite{Banks2020} and references therein.

\begin{theorem}[\cite{Parlett1998symmetric}]\label{thm:eigh}
Let $A \in \Q[\sqrt{-1}]^{n \times n}$ be a Hermitian matrix with distinct eigenvalues $\lambda_i$ and (normalized) eigenvector $v_i$ ($i=1,\dots,n$).

There exists an algorithm that for $\delta > 0$, outputs pairs $(\hat\lambda_i, \hat v_i)$ such that
\[
    \abs{\lambda_i - \hat\lambda_i} \leq \delta \norm{A}, \quad
    \norm{v_i - \hat v_i} \leq \delta \norm{A},
\]
by making $O(n^3 + n^2\log(1/\delta))$ arithmetic operations with $O(\log(n/\delta))$ precision.
\end{theorem}

\begin{remark}
    If $A$ has a repeated eigenvalue, the above definition of the approximation of eigenvectors must be replaced with the approximation of eigenspaces.
    The QR algorithm can also find an approximate eigenspace $\tilde V$ close to a true eigenspace $V$ with respect to the principal angle distance.
    More precisely, $\norm{\pi_{\tilde V} - \pi_V} \leq O(\kappa_A \delta \norm{A})$, where $\kappa_A$ is the condition number of eigenspaces, i.e., the minimum gap of distinct eigenvalues.
    If $A$ has Gaussian integer entries with magnitude at most $M$, we have $\kappa_A = e^{-O(n\log M + n\log n)}$~\cite{Mahler1964}.
    So this only adds an additional $O(n\log M + n\log n)$ factor to the precision.
    For simplicity, we assume that matrices have distinct eigenvalues in the following argument.
\end{remark}

\section{Finite precision algorithms}\label{sec:finite}

So far we have assumed exact arithmetic for computing each iterate in \textsc{MajSinkhorn} (our main algorithmic primitive in this work; Algorithm~\ref{alg:maj-scaling}). In order to obtain a polynomial time algorithm in the Turing model, we need to show that we can round the numbers in each iterate to finite precision. Morally this should be possible because we are optimizing a (geodesically) convex function, but making the argument rigorous still requires some work.

Algorithm \ref{alg:maj-scaling-finite} (\textsc{MajSinkhornFinite}) is the finite precision version of \textsc{MajSinkhorn}. Note that we must represent the eigendecomposition of a matrix implicitly, because diagonalizing unitaries of rational matrices are hardly ever rational.
\begin{Algorithm}
Algorithm \textsc{MajSinkhornFinite}$(\Phi, \alpha, \beta, \eps)$:
\begin{description}
\item[\hspace{.2cm}\textbf{Input:}] A CP map $\Phi:\C^{n\times n} \to \C^{m\times m}$ with Kraus operators in $\Q[\sqrt{-1}]^{m \times n}$, non-increasing vectors $\alpha \in \Q^n, \beta \in \Q^m$, and a parameter $\eps > 0$.

\item[\hspace{.2cm}\textbf{Output:}] $X,Y,q$ such that the scaling $q\Phi_{X^{-1/2}, Y^{-1/2}} $ that is $(\alpha, \beta)$-majorized and has size at least $(1 -  \eps)^2k$.

\item[\hspace{.2cm}\textbf{Algorithm:}] Set $X = I_n, Y = I_m, q = 1$. Implicitly represent $X$ with the pair $x \in \Q^n,0 \preceq A \in \Q[\sqrt{-1}]^{n \times n}$, where $X = U \diag(x) U^*$ for $U$ diagonalizing $A$. Represent $Y$ with $(y,B)$ analogously.
\item[\hspace{.4cm}\textbf{For}] $t \in \{0,..,T\}$: \end{description}
\begin{enumerate}
\item Using the implicit representations of $X$ and $Y$, compute rational vectors $ \mu \approx_\delta \lambda( q X^{-1} \Phi^*(Y^{-1}))$ and $ \nu \approx_\delta \lambda( q Y^{-1} \Phi(X^{-1}))$. Using \textsc{KLProject}, compute the divergences $\eps_X := D(P_\alpha || \mu)$ and $\eps_Y:=D( P_\beta || \nu)$, where $P_\alpha$ and $P_\beta$ are the down-closure of the permutahedron of $\alpha$ and $\beta$, respectively. Compute also $\eps_q \approx_\delta q a - k \log q$ where $a = \tr \Phi(X^{-1})Y^{-1}$.
\item \textbf{If:} $\eps_X, \eps_Y, \eps_q  \leq \min\{\alpha_1, \beta_1\} \frac{\eps^2}{2}$, \textbf{Output:} $X,Y,(1 - \eps) q$.\\
\textbf{Else: } Update either $X,Y$ or $q$ depending on which of $\eps_X, \eps_Y, \eps_q$ is largest:
\begin{enumerate}
\item \textbf{$X$ update:}\begin{enumerate}
\item\label{it:finite-x-C} Using the implicit representation $(y,B)$ of $Y$, compute $0 \preceq \tilde C \in \Q^{n\times n}$ such that $\tilde C \approx_\delta C:=e^z \Phi^*( Y^{-1})$ in the trace norm. Compute $\lambda \in \Q^{n}$ diagonal such that $U \diag(\lambda) U^* \approx_{\delta} \tilde C$ in the trace norm, where $U$ diagonalizes $\tilde C$.
\item \label{it:finite-x-opt} Compute $ x \in \Q^{n}$ to be a minimizer of $x^{-1} \cdot \lambda + \alpha \cdot (\log x^\downarrow)$ subject to $x \geq \mathbf 1$ using \textsc{KLProject}.
\item Return $X$ represented implicitly as $(x, \tilde C)$, and set $\tilde a = (\lambda \cdot x^{-1})/q$ for use in the next $z$ update. \end{enumerate}
\item \textbf{$Y$ update:} analogous to $X$.
\item \textbf{$q$ update:} take $\tilde a$ computed in the $X$ or $Y$ update (or $\tilde a = \tr \Phi(I_n)$ if this is the first step). If $\tilde a \geq k$, set $q = e^z = 1$. Else set $q = e^z =  k/\tilde a$.
\end{enumerate}
\end{enumerate}
\caption{Finite precision algorithm for majorization scaling.}\label{alg:maj-scaling-finite}
\end{Algorithm}

Next we argue that this algorithm has the same time complexity and guarantees as \textsc{MajSinkhorn}. First we argue that, for sufficient precision, the finite precision steps still make at least half the progress that the infinite precision steps make. Next we argue that  \textsc{MajSinkhorn} only ever results in $\|X\|_2,\|Y\|_2,e^z \leq e^{O(T)},$ and that this remains true even with the finite precision steps. Finally, we show that these radius bounds imply that each step of the finite precision algorithm can be performed in polynomial time with polynomial bit complexity.

\begin{theorem}[Finite precision progress bound]\label{thm:finite-precision-progress}
Suppose $\delta \leq \min\{\alpha_1, \beta_1\} \frac{\eps^2}{8}$. Then the amount of progress made by the $X$ or $Y$ update in \textsc{MajSinkhornFinite} is at least half the progress made by the corresponding step in \textsc{MajSinkhorn}. If $\delta q \leq \frac{\eps^2}{8}$, where $q$ the updated value, the $q$ update decreases the function value by at least half as much as the infinite precision $z$ update.
\end{theorem}
\begin{proof}
Consider the $X$ update; the $Y$ update is analogous. The infinite precision version would set $X = \arg\min_{X \succ I_n} (f_C(X):=\tr X^{-1} C + \alpha \cdot \lambda( \log X))$. In the finite precision version we rather set $X$ to a minimizer of $f_{\tilde C}$ where we have chosen $\tilde C$ such that $\|C - \tilde C\|_{1} \leq \delta$. Note that $|f_{\tilde C}(X) - f_{C} (X)| \leq \delta$ for $X \succeq I_n.$ Therefore the optimal value $OPT$ and $\overline {OPT}$ of $f_C$ and $f_{\tilde C}$ are within $\delta$ of one another and if $X$ is an optimizer of $f_{\tilde C}$ we have $f_C(X) \leq OPT + 2\delta$.

In the $q$ update we choose $q \geq 1$ to minimize $\tilde a q - k \log q$, but the new function value becomes $ a q - k \log q$ where $a = \tr \Phi(X^{-1} )Y^{-1}$. These functions differ by $|a - \tilde a|q$, and $|a - \tilde a| = |\tr (\tilde C/q') X^{-1} - (C/q') X^{-1}| \leq \delta$ where $q' \geq 1$ is the previous value of $q$. By similar reasoning to the $X$ update, we must have that $ a q - k \log q$ decreases by at least $\eps_q - 2 \delta q$ which is at least the decrease for the infinite precision $z$ update minus $3 \delta q$.
\end{proof}
We need a small lemma about the KL projection to the down-closure of the permutahedron.

\begin{lemma}\label{lem:projection-bounds}
Suppose $x$ is a minimizer of $x^{-1} \cdot \lambda + \alpha \cdot (\log x^\downarrow)$ subject to $x \geq \mathbf 1$. Then $\|x\|_\infty$ is at most $\max \lambda/\min \alpha$, and $\lambda \cdot x^{-1} \geq \alpha_1$.
\end{lemma}
\begin{proof}
We know from the analysis of \textsc{KLProject} that we may assume $x$ and $\lambda$ are in non-increasing order. Writing $x_i = e^{\sum_{j = 1}^i \Delta_i}$ (or $x = e^{- \sum_{i = 1}^n \Delta_i \mathbf 1_{[i]}})$ for $\Delta_i \geq 0$, by minimality of $x$ we must have that for any $\Delta_i > 0$, in particular for the first $\Delta_i > 0$, that $\partial_{\Delta_i} e^{- \sum_{i = 1}^n \Delta_i \mathbf 1_{[i]}} \cdot \lambda + \alpha \cdot (\log x^\downarrow) = 0,$ or $\sum_{j=1}^i x_i^{-1} \lambda_j = \sum_{j =1} \alpha_i$. For $j \leq i$ we have $x_j = x_1$, so $x_1 = \sum_{j = 1}^i \lambda_j/ \sum_{j = 1}^i \alpha_i$. The first bound follows. We also see that $\sum_{j=1}^i x_i^{-1} \lambda_j = \sum_{j =1} \alpha_i$ implies $\lambda \cdot x^{-1} \geq \alpha_1$.
\end{proof}

\begin{theorem} \label{thm:sinkhorn-bound}
Set $S = \max\{k/s,k/\alpha_1,k/\beta_1,1\}$ where $s = \tr \Phi(I_n)$.
The iterates $X,Y,z$ of \textsc{MajSinkhorn} satisfy $\|X\|_2 \leq S^T s/\alpha_n ,\|Y\|_2 \leq S^T s/\beta_m $ and $e^z \leq S^T.$ Moreover, provided $\delta \leq .5 \min\{\alpha_n,\beta_m\}$, we have $\|X\|_2 \leq 3(3S)^T s/\alpha_n ,\|Y\|_2 \leq 3(3S)^T s/\beta_m $ and $e^z \leq (3S)^T.$

\end{theorem}
\begin{proof}
We first prove by induction that after the $t$ many $z$ updates in \textsc{MajSinkhorn} we have $e^z \leq S^{t}$. For the base case, observe that if the $z$ update happens in the first step then $e^z = \max\{k/s,1\}$. Now suppose $t - 1 \geq 0$ many $z$ updates have happened before, and that $z$ update is to happen in the current step that is not the first step. There must have been an $X$ or $Y$ update in the previous step, because otherwise there will be no $z$ update in the current step. The value of $X$ or $Y$ was chosen so that $e^{z} \tr \Phi(X^{-1}) Y^{-1} \geq \min\{\alpha_1,\beta_1\}$ by the properties of the KL projection. Therefore $\tr \Phi (X^{-1} Y^{-1}) \geq  \min\{\alpha_1,\beta_1\} S^{t - 1}$ by induction, so after the $z$ update we will have $e^{z} = k / \tr \Phi (X^{-1} Y^{-1}) \leq (k/ \min\{\alpha_1,\beta_1\}) S^{t - 1} \leq S^t$.

Next we argue that after $t$ many $z$ updates, $\|X\|_2 \leq S^t s/\alpha_n.$ By Lemma \ref{lem:projection-bounds}, we have $\|X\|_2 \leq \lambda_1(C)/ \alpha_n$. But then $\lambda_1 (C) \leq \tr e^{z} \Phi^*(Y^{-1}) \leq e^{z} s \leq S^t s$ because $Y \succeq I_m$. The argument for $Y$ is analogous.

We now show how to slightly modify the argument for \textsc{MajSinkhornFinite}. For the $q$ update step, the previous finite precision $X$ or $Y$ update will result in $q \tr \Phi(X^{-1}) Y^{-1} \geq \lambda \cdot x - \delta \geq \min\{\alpha_1,\beta_1\} - \delta \geq .5 \min\{\alpha_1,\beta_1\}$ by Lemma \ref{lem:projection-bounds}, so by induction after the $z$ update $q \leq 2 k / \min\{\alpha_1,\beta_1\} (3S)^{t-1} + \delta \leq (3S)^t$. The arguments for $X$ and $Y$ are the same, except for $X$ we have $\|X_2\| \leq \delta + \lambda_1(\tilde C)/ \alpha_n \leq \delta + (\lambda_1(C) + \delta)/ \alpha_n \leq \delta + ( s (3S)^t + \delta)/\alpha_n$. Analogous reasoning proves the bound for $Y$.
\end{proof}

\subsubsection*{Implementation of \textsc{MajSinkhornFinite}}

We now show why the intermediate numbers in each step are only polynomial bit complexity and can be computed in polynomial time. In order to ensure at least half the progress is made, the most onerous requirement is in the $z$ update. For this we need to take $\delta \leq \eps^2/8q,$ but by Theorem \ref{thm:sinkhorn-bound} we have $q \leq 3(3S)^T$ which is merely exponential. Therefore it suffices to take $\delta$ inverse exponential in the input size and $\log(1/\eps)$.

We begin with the $X$ update. Given an implicit representation of $Y$ as $(y,B)$ and $e^z \in \Q$, we need to compute $e^z \Phi^*(Y)$ to precision $\delta$ in the trace norm. Let $B = V \diag(b) V^\dagger$. By the particular choice of $B$ ($\Phi(X)e^{z}$ for the previous $X$), any choice of $V$ diagonalizing $B$ will give rise to $Y$ with the same function value. Therefore it suffices to find a $\tilde V$ which is $\delta$-close to some $V$ diagonalizing $B$ and set $\tilde C = e^z \Phi^*( \tilde V \diag(r) \tilde V^\dagger)$ for $e^{y} \approx r \in \Q^{m}.$ Standard eigendecomposition algorithms output $\tilde V$ deterministically with $\poly(m, \log 1/\delta, \log \|B\|_2)$ bits in time $\poly(m, \log 1/\delta, \log \|B\|_2)$ (Theorem \ref{thm:eigh}). We know that $\|B\|_2 \leq M \sqrt{m np } (3S)^T$ by Theorem \ref{thm:sinkhorn-bound}, so this takes polynomial time. The $Y$ update is analogous, and the $z$ update is trivial. Approximating $\eps_X, \eps_Y, \eps_q$ is similar. This concludes the argument.

\subsubsection*{Calling \textsc{MajSinkhornFinite} as a subroutine}

As \textsc{MajSinkhornFinite} with appropriate $\delta$ makes at least half the progress of \textsc{MajSinkhorn} in every step, it only needs to be run for twice as many iterations. Therefore we can use \textsc{MajSinkhornFinite} as a replacement for \textsc{MajSinkhorn} in all the algorithms in this paper.

The only algorithm that requires further discussion is \textsc{ApproximateIndep}. We are given $X$ and $Y$ implicitly as $(x,A)$, $(y,B)$. We need to check if some pairs $R = \langle u_1, \dots, u_i \rangle$, $L = \langle v_1, \dots, v_j \rangle $ are $\eps$-shrunk subspaces for $\eps$ inverse exponential, where $u_i, v_j$ are \emph{some} eigenbases of $A$,$B$ in order of decreasing eigenvalue. We merely need to compute approximations $\tilde u_i \approx_\delta u_i, \tilde v_j \approx_\delta v_j$ for some such eigenbases, which again is possible in polynomial time by standard eigendecomposition algorithms (Theorem \ref{thm:eigh}). Taking $\delta$ inverse exponential will suffice because $\tr \Phi(\pi_R) \pi_L$ is sufficiently Lipschitz in $R$ and $L$.

\FloatBarrier

\end{document}